\documentclass[11pt]{article}

\newcommand{\blind}{1}

\usepackage[ruled,linesnumbered]{algorithm2e}
\usepackage{xcolor}
\usepackage{graphicx,natbib,amsfonts,amsthm,amssymb}
\usepackage{hyperref}
\usepackage{amsmath,mathrsfs}
\usepackage{wrapfig,lipsum,booktabs}
\usepackage[scr=rsfs,cal=boondox]{mathalfa}
\usepackage{appendix}
\usepackage{enumitem}
\usepackage{multirow}
\usepackage{comment}
\newtheorem{theorem}{Theorem}

\newtheorem{lemma}{Lemma}
\newtheorem{proposition}{Proposition}

\newcommand{\bA}{\mbox{\bf A}}

\newcommand{\bp}{\mbox{\bf p}}

\newcommand{\bAs}{\mbox{\bf A}^s}
\newcommand{\cI}{\mathcal{I}}
\newcommand{\cC}{\mathcal{C}}
\newcommand{\sA}{\mathscr{A}}

\newcommand{\bbR}{\mathbb{R}}
\newcommand{\cd}{\mathcal{d}}

\newcommand{\bone}{\mbox{\bf 1}}
\newcommand{\tS}{\Tilde{S}}
\newcommand{\tR}{\Tilde{R}}
\newcommand{\tG}{\Tilde{G}}

\renewcommand{\appendix}{
 \setcounter{section}{0}%
  \setcounter{subsection}{0}%
  \renewcommand\thesection{\Alph{section}}
  \setcounter{equation}{0}
  \renewcommand{\theequation}{S.\arabic{equation}}
  \setcounter{figure}{0}
  \renewcommand\thefigure{S\arabic{figure}}
  \setcounter{table}{0}
  \renewcommand\thetable{S\arabic{table}}  
  }

\usepackage{caption}
\usepackage{subcaption}

\def\spacingset#1{\renewcommand{\baselinestretch}%
{#1}\small\normalsize} \spacingset{1}

\begin{document}

\if1\blind
{
  \title{\bf Statistics in everyone's backyard: an impact study via citation network analysis}
  \author{Lijia Wang \hspace{.2cm}\\
    Department of Mathematics, University of Southern California\\
    and \\
    Xin Tong \thanks{
    Correspondence should be addressed to Xin Tong (xint@marshall.usc.edu) and Y. X. Rachel Wang (rachel.wang@sydney.edu.au)}  \hspace{.2cm}\\
    Department of Data Sciences and Operations, University of Southern California\\
    and \\
   Y. X. Rachel Wang \footnotemark[1]  \\
    School of Mathematics and Statistics, University of Sydney}
  \maketitle
} \fi

\if0\blind
{
  \bigskip
  \bigskip
  \bigskip
  \begin{center}
    {\LARGE\bf Title}
\end{center}
  \medskip
} \fi

\bigskip
\begin{abstract}
The increasing availability of curated citation data provides a wealth of resources for analyzing and understanding the intellectual influence of scientific publications. In the field of statistics, current studies of citation data have mostly focused on the interactions between statistical journals and papers, limiting the measure of influence to mainly within statistics itself. In this paper, we take the first step towards understanding the impact statistics has made on other scientific fields in the era of Big Data. By collecting comprehensive bibliometric data from the \textit{Web of Science} database for selected statistical journals, we investigate the citation trends and compositions of citing fields over time to show that their diversity has been increasing. Furthermore, we use the local clustering technique involving personalized PageRank with conductance for size selection to find the most relevant statistical research area for a given external topic of interest. We provide theoretical guarantees for the procedure and, through a number of case studies, show the results from our citation data align well with our knowledge and intuition about these external topics. Overall, we have found that the statistical theory and methods recently invented by the statistics community have made increasing impact on other scientific fields.


\end{abstract}

\noindent%
\vfill

\newpage

\section{Introduction}

As a discipline that focuses on the collection, analysis and interpretation of data, statistics is outward facing and often serves as a
tool in other scientific investigations. The age of Big Data has brought about new challenges and opportunities in many fields, where the postulation, verification and refinement of scientific models rely on empirical data. In this sense, one would expect statistics to play an increasingly important role in these fields as the need for methods and tools for handling large, complex data increases. On the other hand, much of the fundamental research in statistical theory and methods requires rigorous mathematical arguments and abstract formulations for generalizability. It can be argued that the technical nature of such works serves as a barrier, making direct adoption of research developments difficult in other fields. In this paper, we consider measuring the impact of theoretical and methodological research in statistics on other scientific disciplines in recent decades. As John Tukey deftly put it: ``the best thing about being a statistician is that you get to play in everyone's backyard."      

One direct way to measure the impact of academic works is through citation data. In the digital age, comprehensive bibliometric studies have been made possible by the existence citation databases such as \textit{Web of Science} and \textit{Scopus}. From these databases, citations between papers can be extracted, represented as a network, and studied using network analysis techniques. These citation networks have been used to track the movements of ideas and measure the distance between different scientific fields \citep{shi2015weaving,V19}. Coauthorship networks can also be constructed from publication records for studying the structure of collaboration patterns \citep{newman2001structure, martin2013coauthorship}. More specifically in statistics, \cite{stigler1994citation} and \cite{varin2016statistical} used the Bradley-Terry model to measure the import and export of knowledge between statistical journals. \cite{ji2016coauthorship} collected and analyzed citation and coauthorship networks for papers in top statistical journals. Rather than focusing on the structure of citation patterns inside statistics, we provide the first comprehensive study analyzing the connections \textit{between statistics and other fields}.  

We collect citation information for papers published in selected statistical journals from the \textit{Web of Science} (WOS) Core Collection. These published papers are termed \textit{source papers} for being the source of knowledge export; our complete data contains citations between source papers as well as their citations by papers (termed \textit{citing papers}) in other journals and fields. Using descriptive statistics, we characterize the trends of citation volumes and compositions of citing fields for the source papers over time, paying attention to fields external to statistics. We compare the internal and external citations for highly cited source papers and identify the corresponding statistical research areas highly ranked by both criteria. Citation trend analysis of these areas allows us to associate them with external fields on which they have made an intellectual impact.      

Given a network, one of the most commonly used analysis techniques is community detection, also known as node clustering. On the citation network for source papers, \textit{global} clustering techniques can be used to partition the nodes into densely connected communities as has been done in \cite{ji2016coauthorship}, offering a global view of various research areas within statistics. However, in this paper, we are more interested in connecting these communities in statistics with research topics in other disciplines they have cast an influence on. That is, given an external research topic (e.g., Covid-19), we consider finding the most relevant community in statistics, with relevance measured by the citation data. A \textit{local} clustering perspective is particularly suitable in this case since i) we expect the relevant community to be small compared with the whole network, making it  challenging to detect by global clustering methods; and ii) the citations between the source papers and the citing papers give a natural way of finding ``seed nodes'' for local clustering algorithms.

A large class of local clustering algorithms is based on seed expansion: given a small subset of seed nodes from a community of interest, the rest of the community is detected by ranking the other nodes according to the landing probabilities of random walks started from the seeds. Different classes of algorithms correspond to different ways of combining these probabilities for random walks of different lengths, with the most popular ones being versions of personalized PageRank (PPR) \citep{andersen2006communities,whang2013overlapping, kloumann2014community} and heat kernels \citep{chung2009local, kloster2014heat}. These algorithms have been widely applied to large-scale real networks with much empirical success. More recently, attention has been paid to studying their theoretical properties on networks generated from the stochastic block model (SBM, \cite{holland1983stochastic}) and its variant. \cite{KUK17} showed that PPR corresponds to the optimal linear classifier under a suitable two-block SBM. Using the more general degree-corrected SBM (DC-SBM, \cite{KN11}), \cite{CZR20} showed that PPR can include high-degree nodes outside the community of interest, while using the adjusted PPR (aPPR) algorithm in \cite{andersen2006local} can correct the degree bias, achieving consistency in the detection of the target community with high probability.  

After the nodes have been ranked in terms of their relevance to the target community, it remains to choose the size of the local cluster and cut the sorted list of nodes at the desired size. A scoring function is thus needed to evaluate the quality of the communities found along the sorted list. One of the most widely used scoring functions is conductance \citep{andersen2006local, YL15, VM16}, which measures the fraction of total edge volume that points outside the cluster. A smaller conductance indicates the cluster is more separated from the rest of the network, hence more likely to be a community on its own. Assessing the performance of various scoring functions on a large number of real networks, \cite{YL15} showed conductance consistently gives
good performance in identifying ground-truth communities. The theoretical properties of conductance, however, has not been investigated under the local clustering setting with generative network models. For our local clustering procedure, we adopt aPPR followed by conductance local minimization. Under the DC-SBM, we show that with high probability, this procedure finds all the nodes in the community to which the seed nodes belong.

The rest of the paper is organized as follows. In Section \ref{sec:data_summary}, we  describe the data collection procedure and various covariates used in our analysis. We provide a summary of citation trends over time and citation distributions for each journal. In Section \ref{sec:compare}, we study the diversity of citing fields by grouping the citing papers according to the research areas they belong to. In particular, we determine which highly cited source papers have high citations both within statistics and outside statistics, as well as those that appear to have a larger impact on one side of the audience. In Section \ref{sec:local_cluster}, we describe our local clustering procedure for finding the statistical community most relevant to an external research topic. We provide theoretical analysis of its behavior under the DC-SBM and demonstrate its performance on simulated data and a number of case studies from our citation network. We end the paper with a discussion of the merits and limitations of our study, pointing to directions in which it can be extended in the future.

\section{Data collection and overview of citation trends}
\label{sec:data_summary}
\subsection{Data collection}\label{sec:data}

We conducted our study on all the papers published from 1995 to 2018 in five influential statistics journals: \textit{Annals of Applied Statistics} (AOAS), \textit{Annals of Statistics} (AOS),  \textit{Biometrika}, \textit{Journal of the American Statistical Association} (JASA) and \textit{Journal of the Royal Statistical Society: Series B} (JRSSB)\footnote{We include both publication names JRSSB used during 1995-1997. 
}. Using a Python script, we crawled the bibliographic database \textit{Web of Science (WoS) Core Collection} to collect citation data for a total of  $9{,}338$ papers published in these journals in the time span considered. We only included publications whose document types are listed as ``article" in WoS.  We call these publications \textit{source papers} since they act as a source of knowledge for papers citing them. Among our selected journals, AOS, Biometrika, JASA, and JRSSB are considered by many researchers in the statistics community as top outlets for theory and method works. We have also included AOAS as a representative journal with a broad applied focus. 



For each source paper, the WoS database provides a list of papers citing it and the corresponding publication information. We finished extracting these lists before December 2020. In addition to the citations between the source papers, $264{,}356$ papers from other journals (or from the selected five statistics journals but published in 2019 and 2020) cited these source papers; these papers are called \textit{citing papers}.\footnote{The accessibility of citing papers depends on the university library VPN used to access the WoS database.} Rather than limiting to ``article" as we did for the source papers, the citing papers can be of any document type. Based on the lists of citations, we build a citation network that consists of $273{,}694$ nodes including all the source and citing papers, and edges representing citations between the source papers and from the citing papers to the source papers.

The above citation network can be represented by a binary adjacency matrix $\bA \in \{0, 1\}^{273694\times 9338}$, in which
\begin{equation}\label{eq:adj}
A_{ij} =\begin{cases}
1\,,&\mbox{$i$ cites $j$;}\\
0\,,&\mbox{otherwise.}
\end{cases}
\end{equation}
In this matrix, we assign each source paper to an index in $\cI_s = \{1,\ldots,9338\}$ and each citing paper to an index in $\cI_c =\{9339,\ldots,273694\}$. Our current study does not contain citations from the source papers to the citing papers since we are primarily interested in the impact of source papers on other scientific works.


We obtained the publication information for both source and citing papers from the WoS database. In particular, the following variables are central to our analysis: (1) article title,  (2) publication source title (e.g., journal or conference names), (3) publication year, (4) author keywords, (5) abstract, (6) WoS categories (e.g., ``Statistics \& Probability" and ``Mathematical \& Computational Biology"), and (7) research areas (e.g., ``Mathematics"). In our dataset, only 98 of all the papers do not have any specified categories (nor research areas), thus we label their categories (and research areas) as ``NA". We use the broad research areas to classify the general field of each paper and the WoS categories to provide finer classifications when statistics needs to be distinguished from other research fields. More discussions on the division and field classification can be found in Section~\ref{sec:compare}.

Furthermore, in Section~\ref{sec:compare}, we use the variables (2) publication source title, (3) publication year and (7) research areas to illustrate the change of impact on external and internal areas over time for the selected five journals and their highly cited papers. In Section~\ref{sec:local_cluster}, we select papers from target topics based on (1) article title and (5) abstract, and validate the local community found using (4) author keywords. More details about the usage of these variables will be presented in the respective sections.

\subsection{Citation distributions and trends}


As shown in Figure~\ref{fig:histogram_table}, the vast majority of source papers have fewer than $500$ citations with $1.92\%$ of the source papers receiving zero citations. Figure~\ref{fig:histogram} further plots the distribution of the citation counts for source papers with citations from 0 to 500. We observe that removing the zero-citation papers would lead to a power-law distribution of the citation counts. 
Notably, only $0.06\%$ papers (6 papers) received more than $5{,}000$ citations. Highly cited papers like these will be discussed in more details in Section~\ref{sec:top_20}.

\spacingset{1}
\begin{figure}[!ht]
     \centering
     \begin{subfigure}[b]{0.42\textwidth}
    \centering \includegraphics[width=\textwidth]{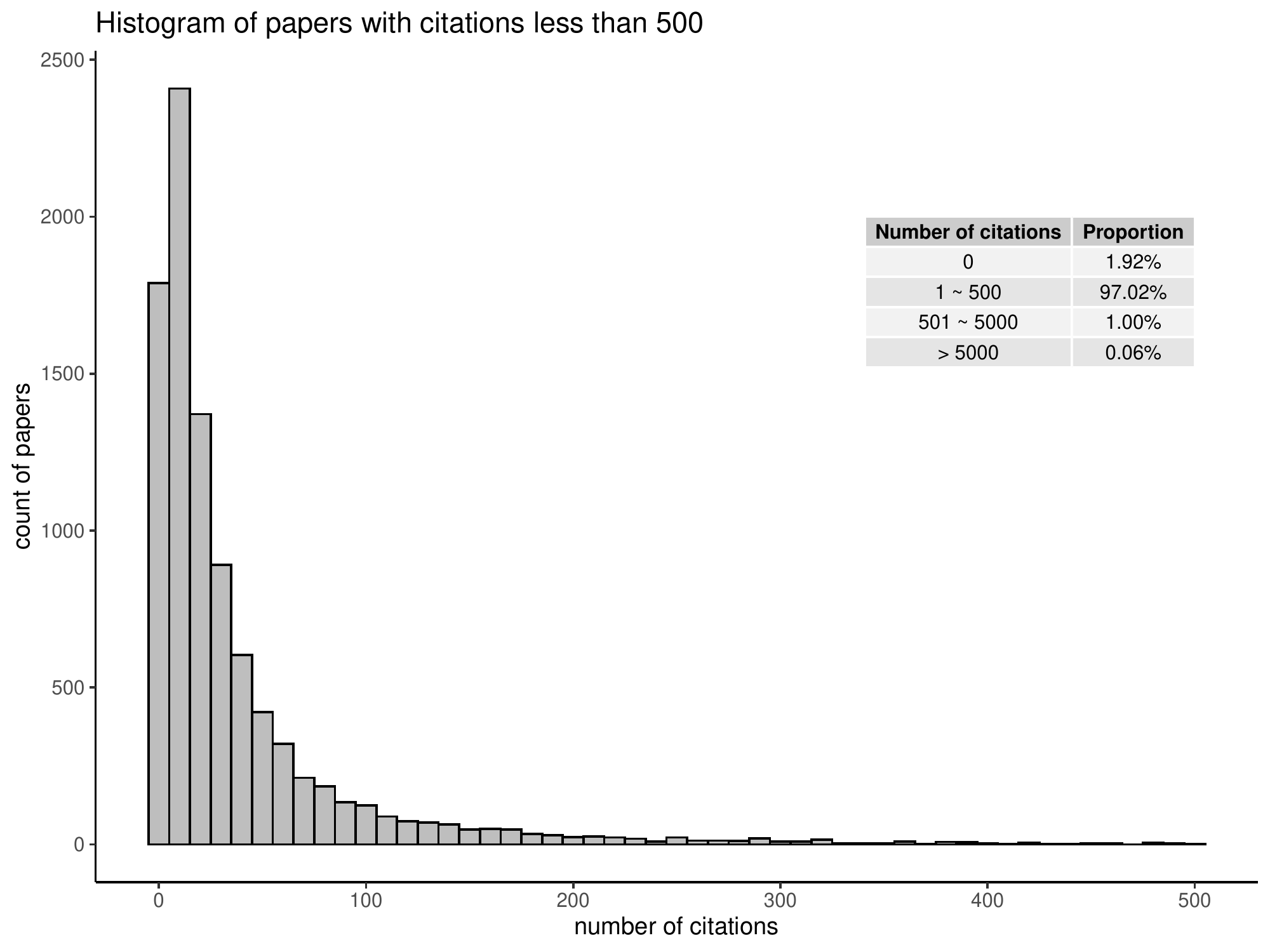}
          \caption{ }\label{fig:histogram_table}
     \end{subfigure}
     \hfill
     \begin{subfigure}[b]{0.55\textwidth}
         \centering
         \includegraphics[width=\textwidth]{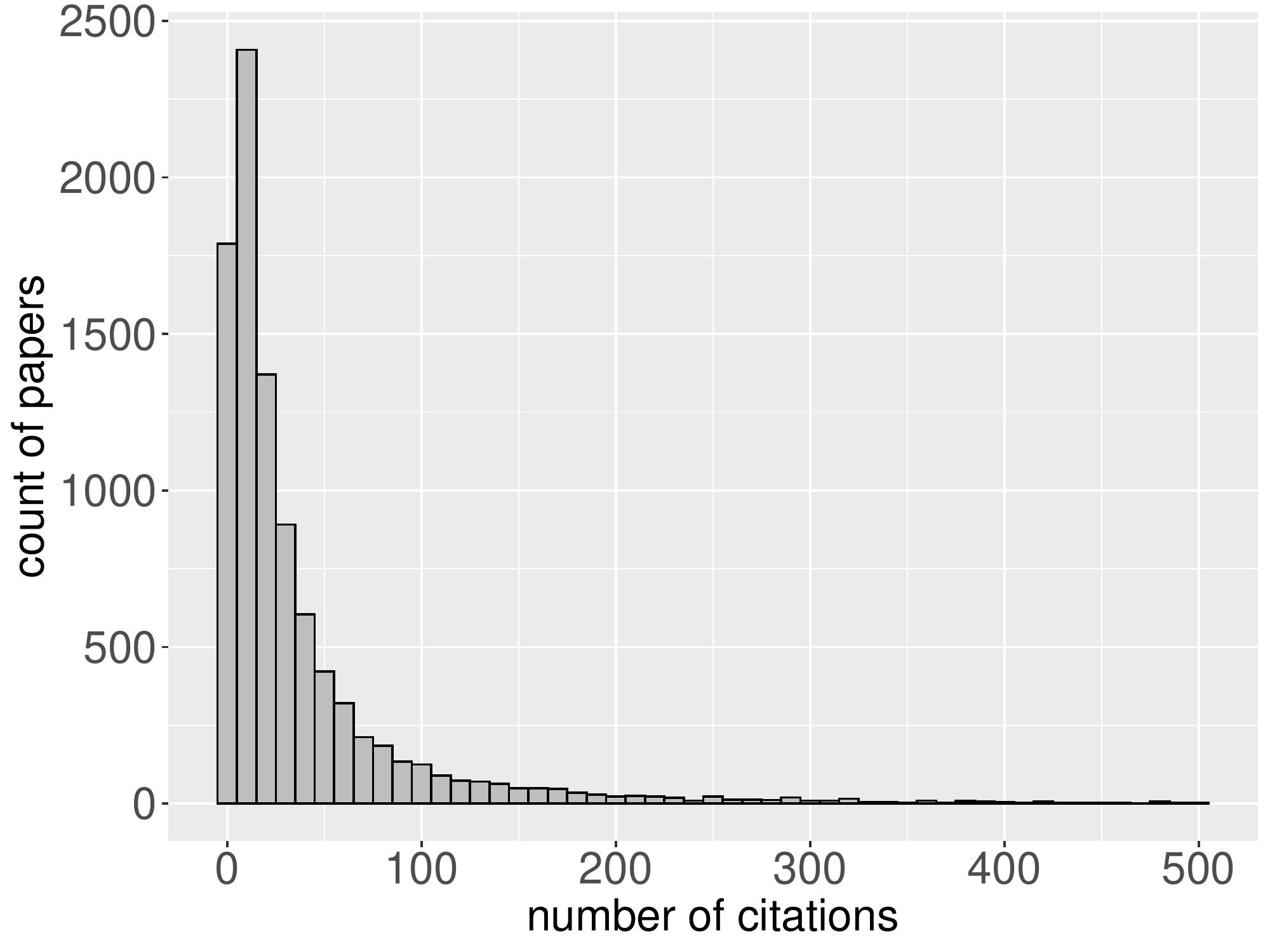}
         \caption{ }\label{fig:histogram}
     \end{subfigure}
     \hfill
 \caption{ The distribution of citation counts for the source papers: (a) the proportions of papers falling into different citation brackets; (b) the histogram of papers with citations $\leq 500$.}
\end{figure}

Looking at the trends over the years, the total number of citations for each journal grows consistently (Supplementary Figure~\ref{fig:total_citation}), and the growth is not due to the journals expanding their volumes of publications. In fact, there was no significant increase in the annual number of publications in each journal (Supplementary Figure~\ref{fig:volume}) except AOAS. AOAS was established in 2007 and subsequently went through fast growth period before stabilizing. 
To account for the effect of publication numbers, for each year $T$, we normalize the annual citation count for each journal by the total number of published papers from 1995 to $T$ in that journal, since any citing paper published in year $T$ is free to cite source papers in the period 1995-$T$. Figure~\ref{fig:avg_cite} shows that the normalized citations still increase consistently over the years for all the journals, among which JRSSB enjoys substantially more citations per article after 2002. AOAS' normalized citations have been growing quickly as a relatively new journal.

\spacingset{1}
\begin{figure}[!ht]
     \centering
     \begin{subfigure}[b]{0.48\textwidth}
    \centering
         \includegraphics[width=\textwidth]{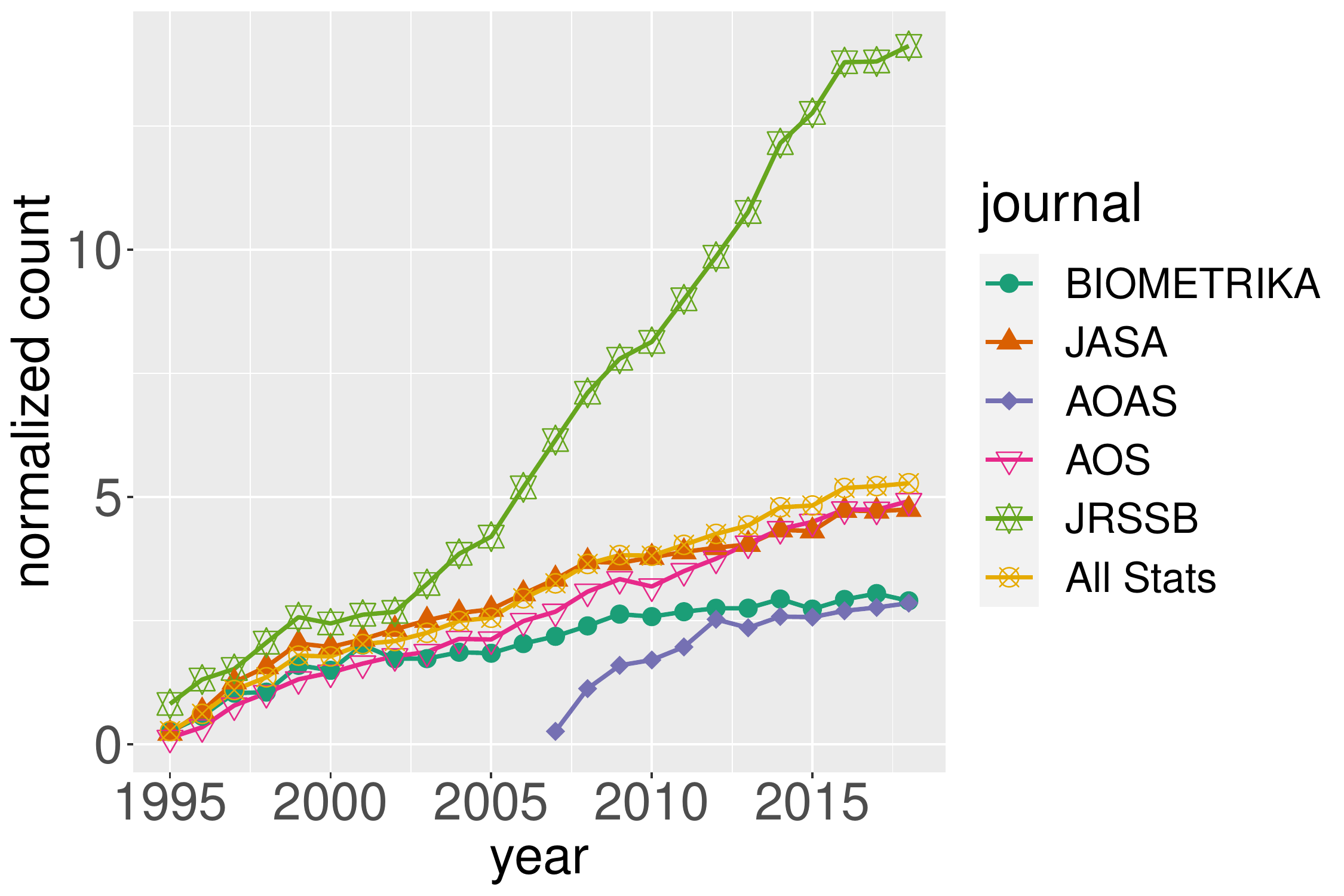}
         \caption{}\label{fig:avg_cite}
     \end{subfigure}
     \hfill
     \begin{subfigure}[b]{0.48\textwidth}
         \centering
         \includegraphics[width=\textwidth]{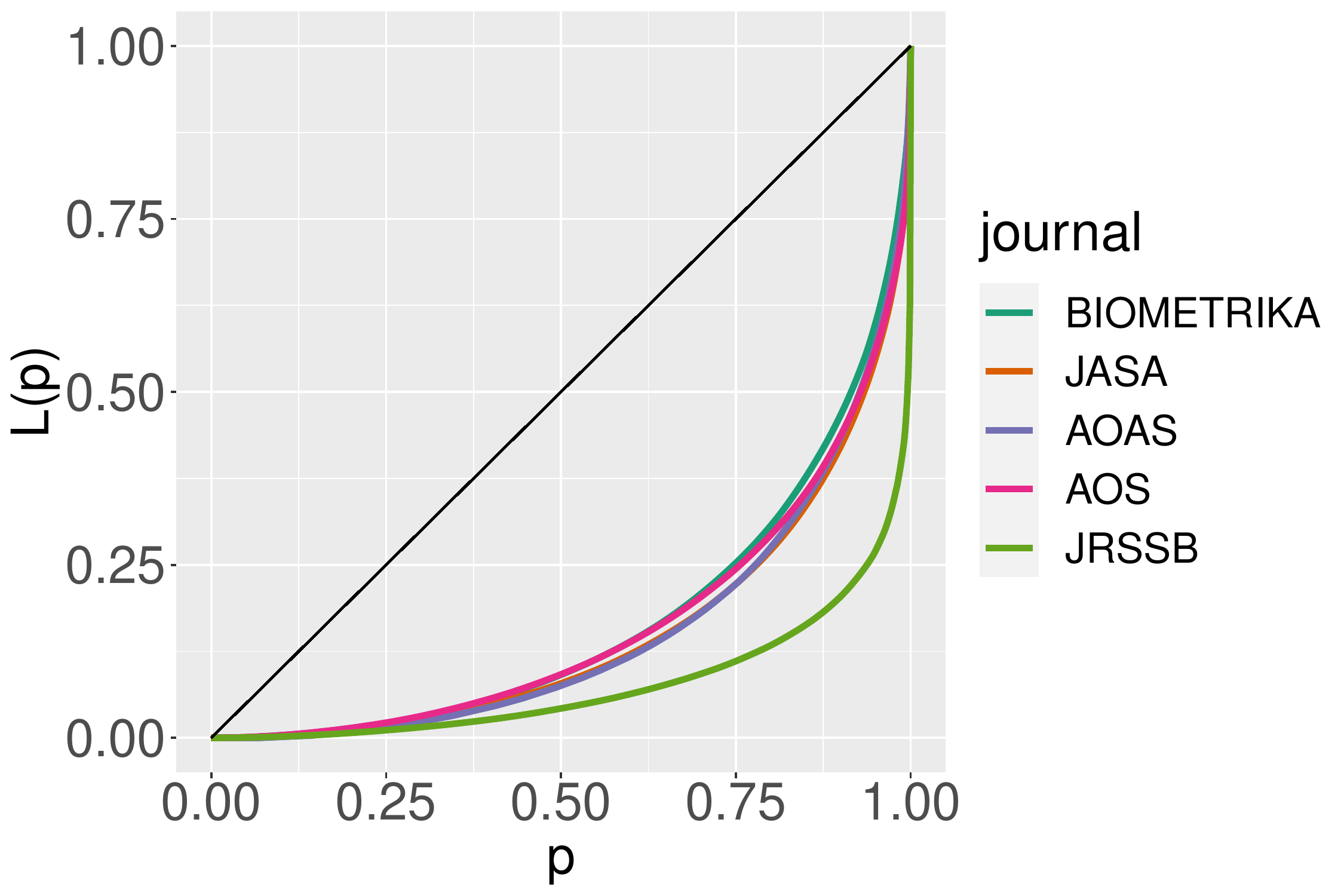}
          \caption{}\label{fig:lorenz}
     \end{subfigure}
     \hfill
 \caption{Citation trends and distributions for each journal: (a) the normalized number of citations over the years. ``All Stats" refers to all the source papers; (b) the Lorenz curve for each  journal.}\label{fig:two_gini}
\end{figure}




It is clear that citation counts are not distributed equally across all the papers, and one possible way to measure citation inequality is through the Lorenz curves \citep{V19,ji2016coauthorship}. For journal $j$, define
$$L(p) = \frac{\sum^{\lfloor p\times N_j \rfloor }_{i = 1} d_{(i)}}{\sum^{N_j}_{i = 1} d_{(i)}}\,,$$
where $N_j$ is the number of publications, $p$ is the percentage, and $ d_{(1)}, d_{(2)}, \ldots, d_{(N_j)}$ are the citation numbers in a non-decreasing order of papers in journal $j$ published in 1995-2018. $L(p)$ calculates the percentage of citations shared by the bottom $p$ percent of papers as a measure of inequality. Figure~\ref{fig:lorenz} plots $L(p)$ as a Lorenz curve for each journal, with curves closer to the bottom right corner indicating greater extent of inequality. Most journals have highly similar curves, while JRSSB appears to have the most significant inequality. This can be explained by the fact that there are four papers that each received more than $5{,}000$ citations, accounting for $49.5\%$ of the total citations towards JRSSB in this period. (Recall that we only have six papers in our dataset with citation numbers exceeding $5{,}000$.) After removing these four papers, the normalized citation counts for JRSSB become much closer to the other journals but remain the highest of all journals (Supplementary Figure~\ref{fig:normal_2}). 


\section{Comparison between internal and external citations}\label{sec:compare}


As the overall citations for each journal increase over the years, how much of the increase can be attributed to research fields outside statistics? In this section, we break down the citations by their research fields, paying attention to the distinction between internal and external citations. 

As mentioned Section~\ref{sec:data}, even though the WoS categories and research areas can help us identify the research field each paper belongs to, we still have to make a decision about whether a citation should be considered inside (internal) or outside (external) of statistics. This is a subjective decision in some sense given the interdisciplinary nature of many research topics in statistics and the overlap of statistics with fields such as mathematics, computational biology and econometrics. We take the following approach, which perhaps can be viewed as conservative in estimating external impact. We consider two types of internal papers. The first type includes papers containing the tag ``Statistics \& Probability" in their WoS categories, which applies to all the papers published in common statistics and/or probability journals. These papers are labeled as ``STATS" in our subsequent plots. The second type includes papers whose WoS categories contain the keyword ``math" (e.g., ``Mathematics" and ``Mathematical \& Computational Biology"). Additional papers selected by this step are published in journals such as \textit{Journal of Econometrics, BMC Bioinformatics}, thus from fields reasonably close to statistics. In what follows, these papers are labeled as ``MATH" and counted as internal citations. The rest of the papers are considered as external. This procedure divides our dataset into $83{,}503$ internal and $190{,}191$ external papers.


We use the research areas\footnote{\url{https://images.webofknowledge.com/images/help/WOS/hp_research_areas_easca.html}} provided by WoS to classify the external papers into five broad categories: arts \& humanities (ART), life sciences \& biomedicine (BIO), physical sciences (PHY), social sciences (SOC), and technology (TECH). We note that the finer divisions under research areas could also be used to search for the second type of internal papers described above with ``math" as a keyword, but doing so would lead to a subset of the papers already selected by the WoS categories.


\subsection{Diversity of citing fields over time}\label{sec:diversity}

Using the category labels discussed above, Figure~\ref{fig:overall_diversity} shows the research area breakdowns for all the citations over the years. If an external paper  
lists multiple research areas, each area is weighted equally and contributes a fractional count to the total. As expected, in the earlier years of our period of study, most of the citations are from within statistics. However, the proportion of external citations soon begins to increase at a fast pace and finally exceeds half. Among the external citations, BIO and TECH have heavy weights. The same trend for each journal separately is presented in Supplementary Figure~\ref{fig:diversity_journal}. The proportion of external citations also increases over time for all the journals, with AOAS and JRSSB having larger  proportions than the others.

\spacingset{1}
\begin{figure}[!ht]
     \centering
     \begin{subfigure}[b]{0.55\textwidth}
    \centering
    \includegraphics[width=\textwidth]{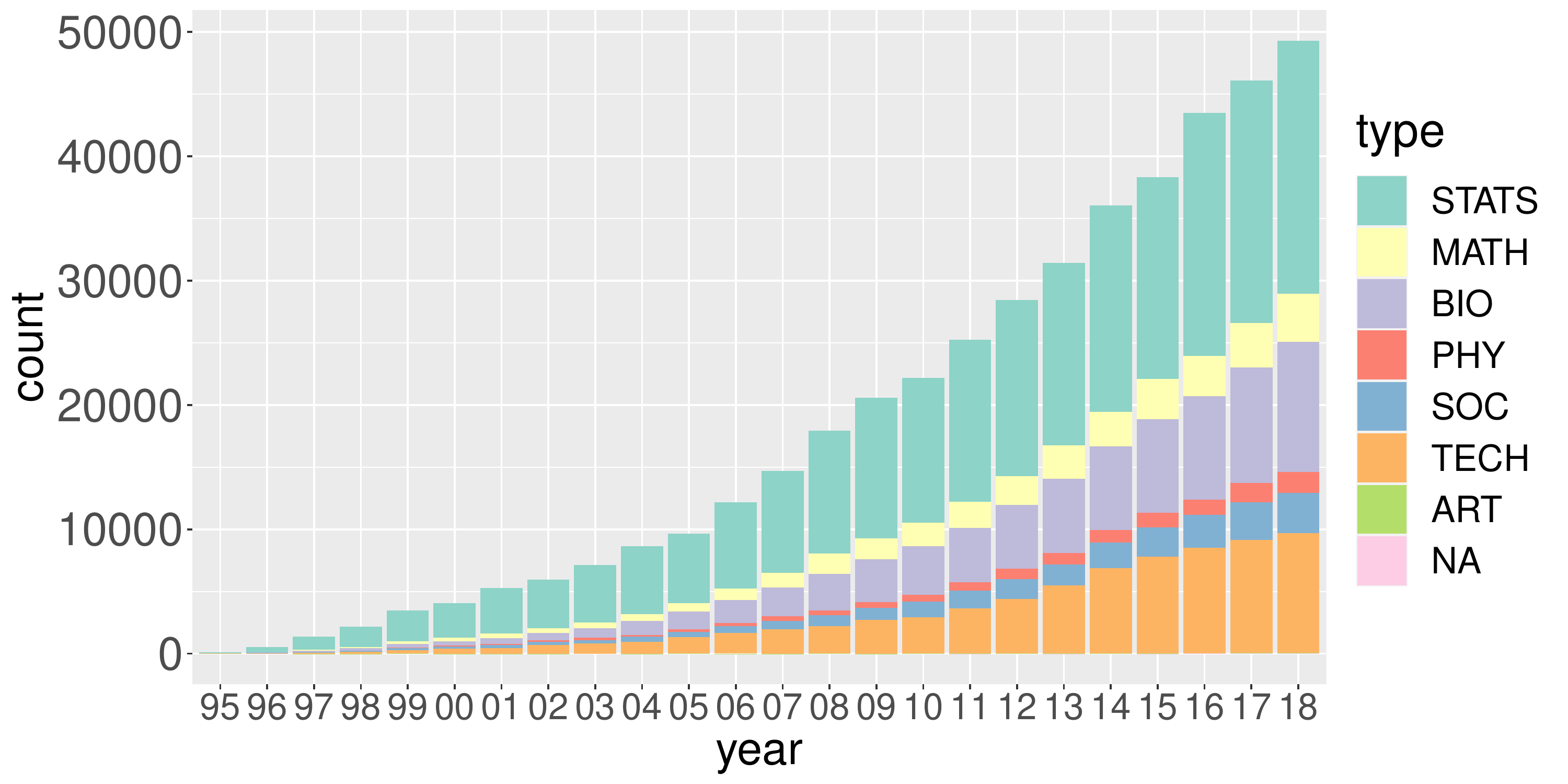}
    \vspace{-0.3cm}
    \caption{}\label{fig:overall_diversity}
     \end{subfigure}
     \hfill
     \begin{subfigure}[b]{0.4\textwidth}
         \centering
    \includegraphics[width=\textwidth]{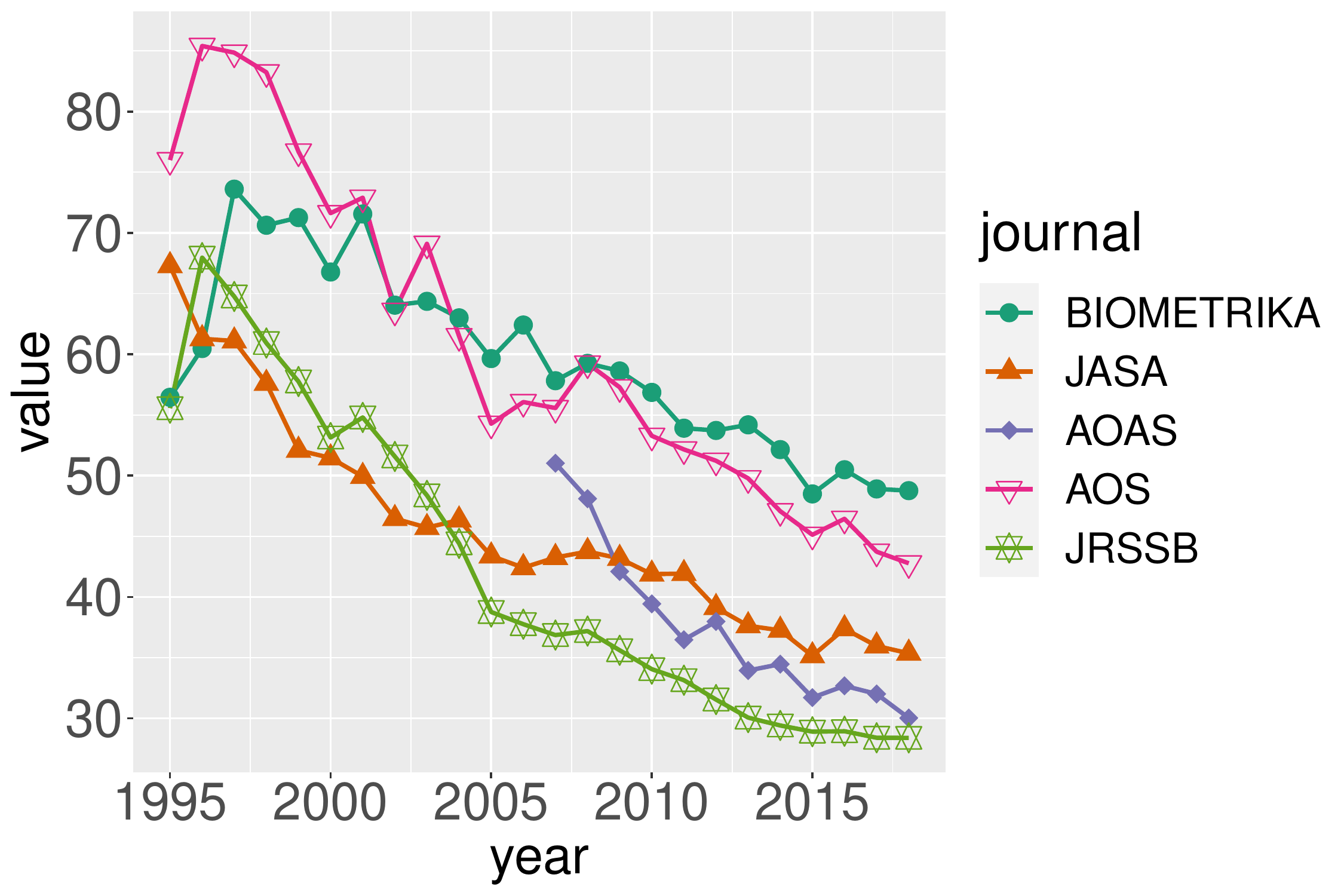}
    \vspace{-0.3cm}
    \caption{}\label{fig:gini_con}
     \end{subfigure}
     \hfill
 \caption{Diversity of citing fields in statistics journals: (a) the annual counts of internal and external citations for all the source papers; (b) the yearly Gini concentration for each journal.}
\end{figure}


One way to summarize the distribution of proportions and put the diversity measure for each journal on the same scale is through the use of Gini concentration \citep{stigler1994citation}. Let
$$\mbox{Gini Concentration} = 100 \times \sum_i s_i^2\,,$$
where $s_i$ is the proportion of citations from research category $i$, and we consider the same categories as shown in Figure~\ref{fig:overall_diversity} except that we combine STATS and MATH into one internal category. Journals with more diverse citations by external categories have lower Gini concentrations. Figure~\ref{fig:gini_con} plots the change in the Gini concentration for each journal over the years. Overall the trends agree with our results in Figure~\ref{fig:overall_diversity} and Supplementary Figure~\ref{fig:diversity_journal}. All the journals have demonstrated increasing connections with external fields, with AOAS, JASA, and JRSSB being more diverse than the others.


\subsection{Internal and external impact of most highly cited papers }
\label{sec:top_20}

In the previous section, we compared the proportions of internal and external citations at an aggregated level within each journal. Now we turn to examine the internal and external impact of some specific source papers selected based on their high citation counts. Do highly cited papers always have high impact both internally and externally? To this end, we first rank the source papers according to their internal and external citation counts separately. Focusing on papers in the top 20 list by either internal or external counts, Figure~\ref{fig:rank_comp} shows their respective ranks internally and externally. One can see that most of these papers are ranked high under both criteria except for a few outliers. We focus on the most obvious two (boxed in red) and provide their information in Table~\ref{table:top_outlier} and further analysis below.

\spacingset{1}
\begin{figure}[!ht]
\centering
\includegraphics[width=10cm]{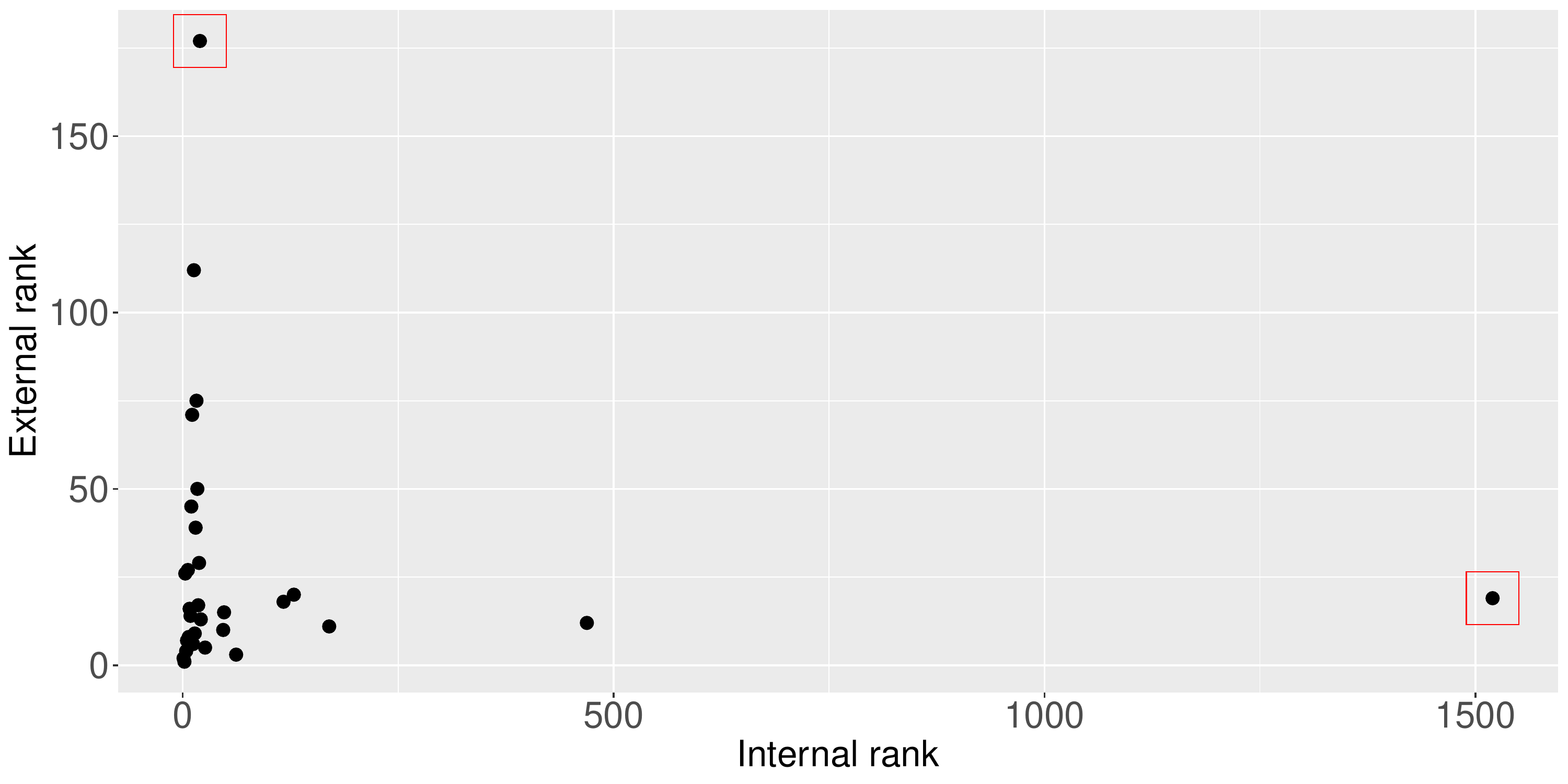}
\vspace{-0.3cm}
\caption{The comparison of external and internal ranks for highly cited papers.}\label{fig:rank_comp}
\end{figure}

\begin{table}[!ht]
\centering
\footnotesize
\begin{tabular}{p{9cm}|c|c|c}
\hline\hline
Title  [author, year] & \begin{tabular}[c]{@{}l@{}}Rank\\ (internal)\end{tabular} & \begin{tabular}[c]{@{}l@{}}Rank\\ (external)\end{tabular} & \begin{tabular}[c]{@{}l@{}}\# of\\ Citations\end{tabular}   \\ \hline
The multivariate skew-normal distribution \citep{azzalini1996multivariate} & 20  & 177 & $749$     \\ \hline
 A nonparametric trim and fill method of accounting for publication bias in meta-analysis \citep{duval2000nonparametric}  & $1{,}520$ & 19 & $1{,}362$    \\ \hline
\end{tabular}
\caption{Papers with significantly different  internal and external ranks.}\label{table:top_outlier}
\end{table}

The first paper \citep{azzalini1996multivariate} in Table~\ref{table:top_outlier} ranks in the top 20 based on the internal citation counts, but its external rank is relatively lower in comparison. Since the paper is about distribution theory, unsurprisingly we find most of the citations come from fields closely related to statistics. Supplementary Table~\ref{table:hl_paper_1} provides the top 10 WoS categories and their number of occurrences among the citations, with ``Statistics \& Probability" appearing most often. Also, most of these categories contain the keyword ``math", which explains the higher internal rank. The other categories (e.g, ``Computer Science, Interdisciplinary Applications") are still closely related to statistics or mathematics. Upon removing the internal papers, the occurrences of these categories other than statistics and mathematics decrease significantly (Supplementary Table~\ref{table:hl_paper_2}), suggesting many of the previous counts are contributed by internal papers with multiple category labels. Overall, the paper has reached a larger audience within statistics and mathematics, most likely due to its technical nature.  


The second paper \citep{duval2000nonparametric} in Table~\ref{table:top_outlier} demonstrates the opposite pattern, with a high external rank but a low internal rank. This paper proposes a practical method of evaluating and adjusting for the possibility of publication bias (e.g., a preference for positive results), a well-known phenomenon in published academic research especially in meta-analysis, and thus has attracted wide scientific interests. Supplementary Table~\ref{table:lh_paper_1} lists the top 10 most frequent WoS categories among all the citations. One can see that list is dominated by psychiatry and psychology, while statistics or mathematics related categories are not present. This list remains almost unchanged after removing all the internal papers from the citations (Supplementary Table~\ref{table:lh_paper_2}). We have additionally searched for keywords related to publication bias in the title and author keywords of the internal papers. The search only returns 59 papers, confirming the topic is less explored internally and could be a potential area for further theoretical and methodological development in statistics. We note that Figure~\ref{fig:rank_comp} has another paper \citep{lo2001testing} with a low internal rank (469) and a high external rank (12). The paper has a similar category profile to \citep{duval2000nonparametric} (Supplementary Table~\ref{table:lh_paper_3}), hence detailed discussion is omitted.

\spacingset{1}
\begin{table}[!ht]
\centering
\footnotesize
\begin{tabular}{p{8cm}|p{1.5cm}|c|c|c}
\hline\hline
Title [author, year]& \begin{tabular}[c]{@{}l@{}}Area\\ (statistics)\end{tabular}  & \begin{tabular}[c]{@{}l@{}}Rank\\ (internal)\end{tabular} & \begin{tabular}[c]{@{}l@{}}Rank\\ (external)\end{tabular} & \begin{tabular}[c]{@{}l@{}}\# of\\ Citations\end{tabular}  \\ \hline
Reversible jump Markov chain Monte Carlo computation and Bayesian model determination \citep{green1995reversible} & MCMC & 8  & 16 & $2{,}868$     \\ \hline
Identification of causal effects using instrumental variables \citep{angrist1996identification} & Causal & 18 & 17 & $2{,}125$     \\ \hline
Least angle regression \citep{efron2004least} & Penalized regression  & 7 & 8   & $4{,}252$     \\ \hline
The control of the false discovery rate in multiple testing under dependency \citep{benjamini2001control} & FDR & 12  & 6 & $5{,}062$     \\ \hline
Model selection and estimation in regression with grouped variables \citep{yuan2006model} & Penalized regression    & 9 & 14   & $2{,}935$     \\ \hline
Regularization and variable selection via the elastic net \citep{zou2005regularization} & Penalized regression    & 5 & 7 & $5{,}790$     \\ \hline
A direct approach to false discovery rates \citep{storey2002direct} & FDR & 14 & 9   & $3{,}186$     \\ \hline
Bayesian measures of model complexity and fit \citep{spiegelhalter2002bayesian} & Bayesian model selection   & 4     & 4    & $6{,}743$     \\ \hline
Controlling the false discovery rate: a practical and powerful approach to multiple testing \citep{benjamini1995controlling} & FDR & 2   & 1    & $46{,}899$    \\ \hline
Regression shrinkage and selection via the Lasso \citep{tibshirani1996regression} & Penalized regression     & 1   & 2     & $16{,}905$    \\ \hline
\end{tabular}
\caption{Papers whose  internal and external citations both rank in the top 20.}\label{table:top_20}
\end{table}


\begin{figure}[!ht]
     \centering
     \begin{subfigure}[b]{0.45\textwidth}
    \centering
    \includegraphics[width=\textwidth]{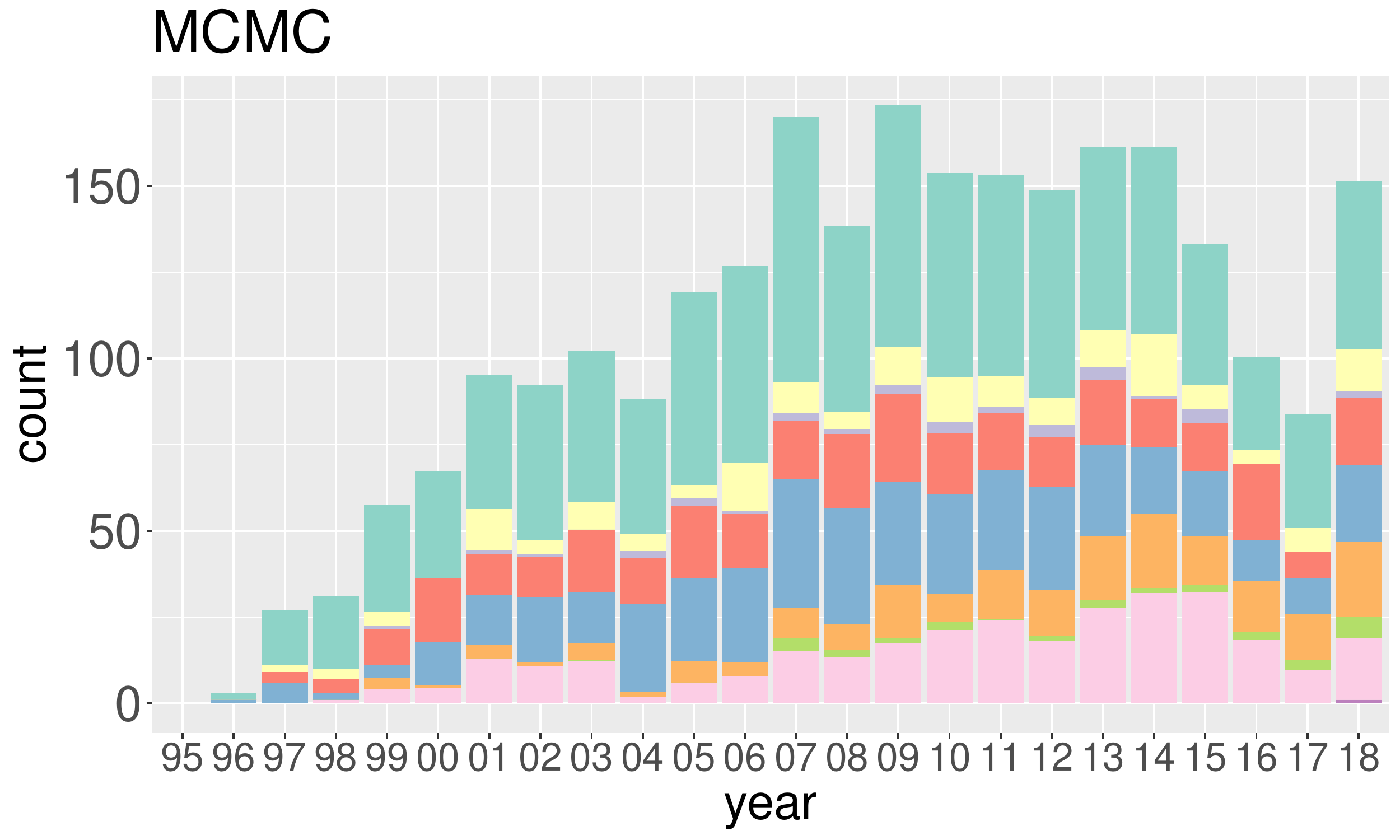}
    \caption{}
     \end{subfigure}
     \begin{subfigure}[b]{0.45\textwidth}
    \centering
    \includegraphics[width=\textwidth]{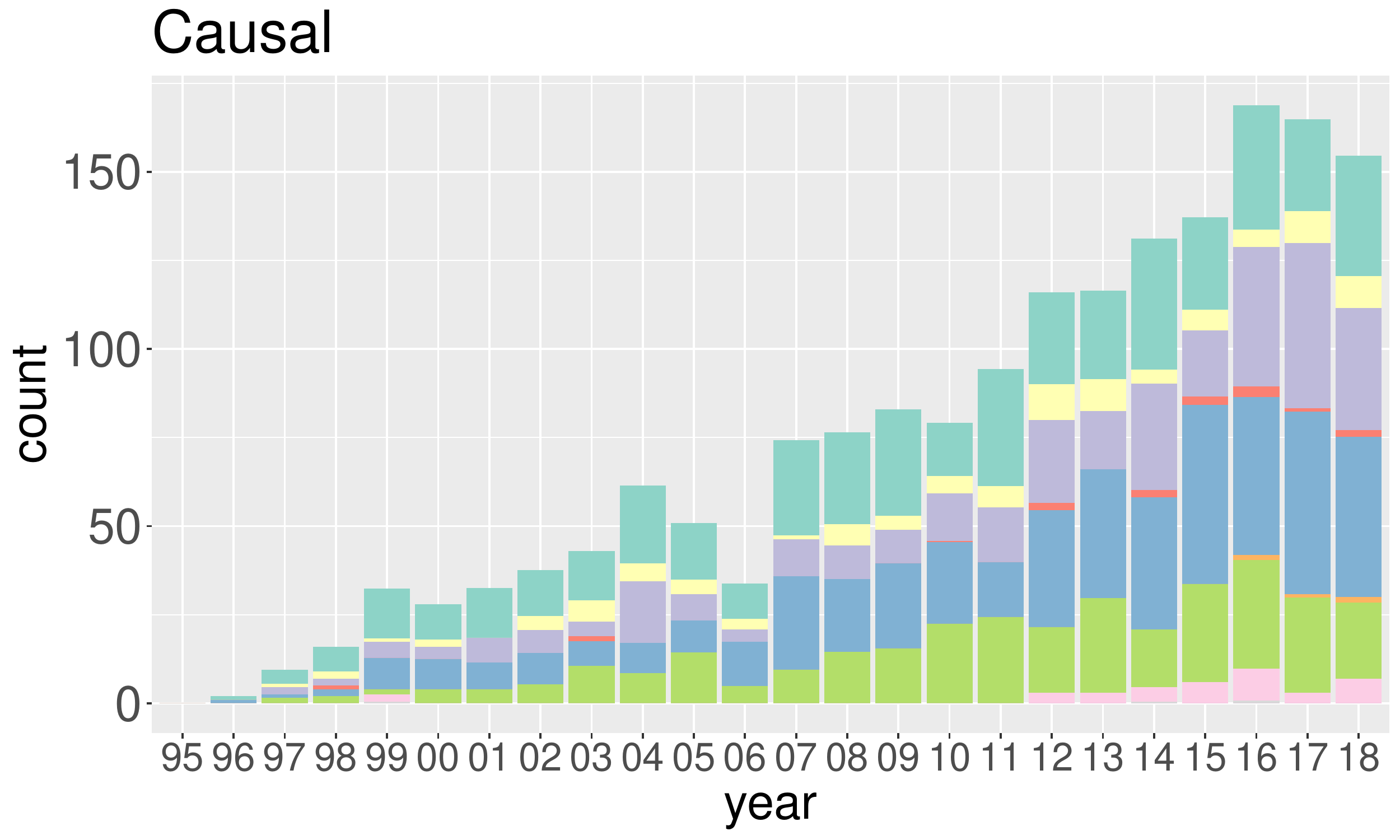}
    \caption{}
     \end{subfigure}
     \begin{subfigure}[b]{0.45\textwidth}
    \centering
    \includegraphics[width=\textwidth]{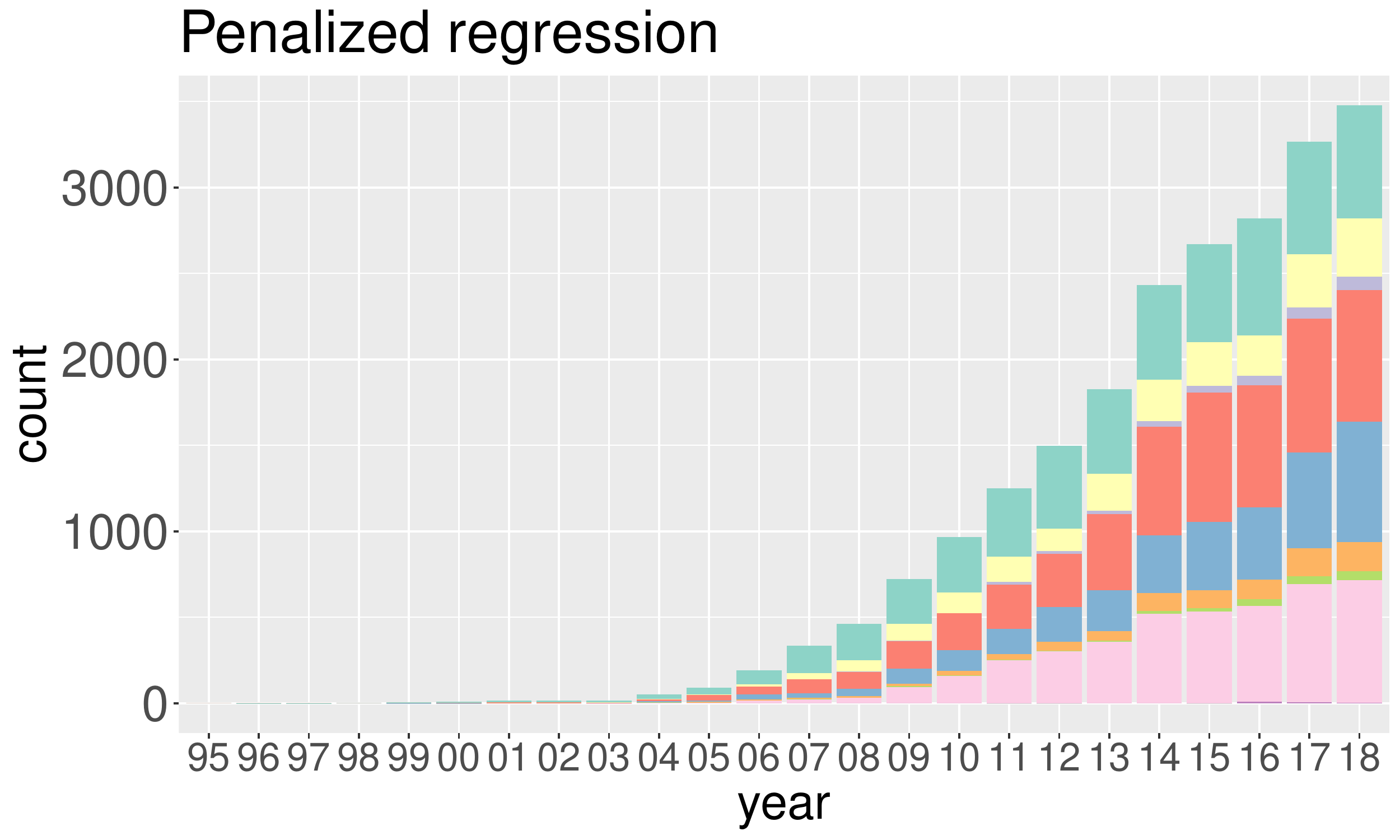}
    \caption{}
     \end{subfigure}
     \begin{subfigure}[b]{0.45\textwidth}
    \centering
    \includegraphics[width=\textwidth]{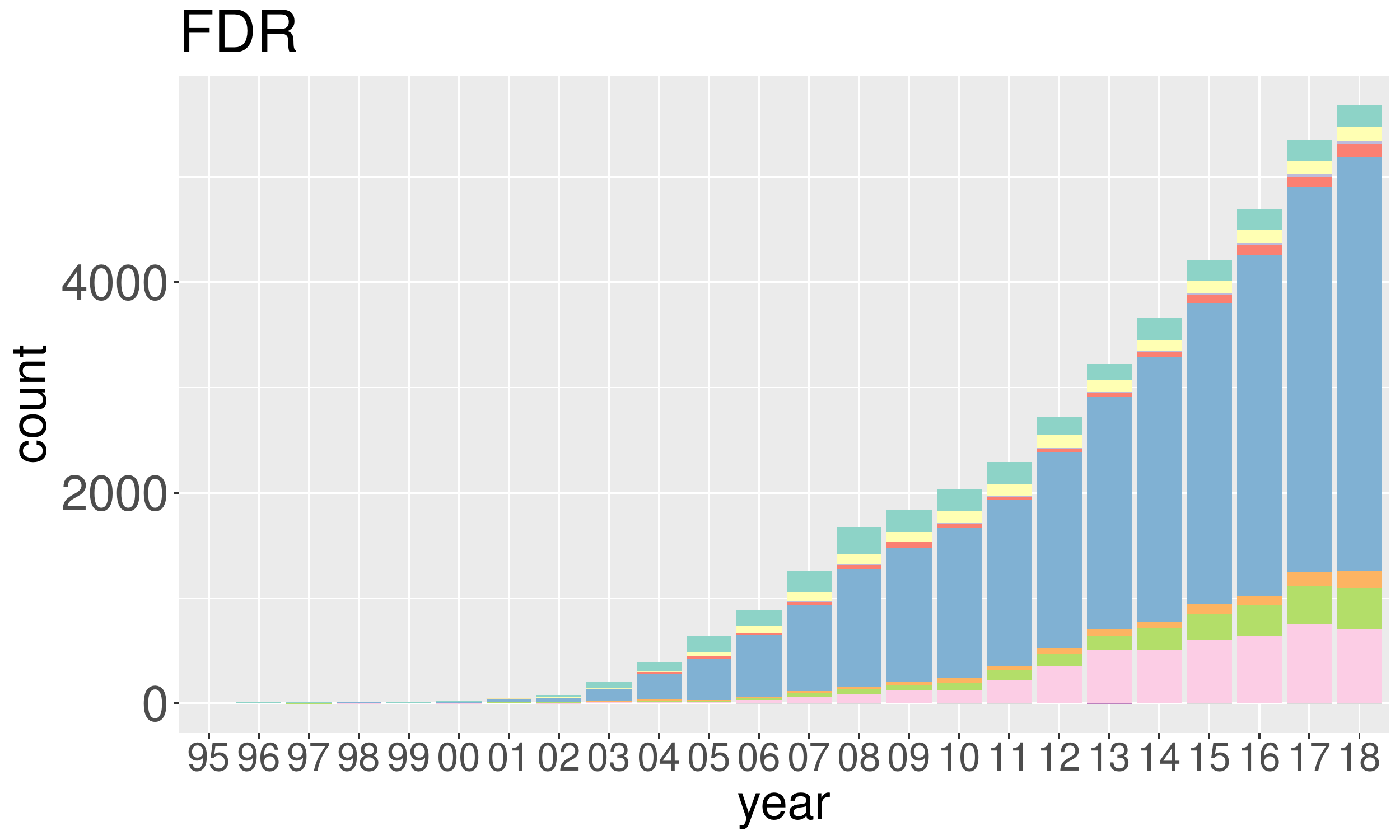}
    \caption{}
     \end{subfigure}
     \begin{subfigure}[b]{0.5\textwidth}
    \centering
    \includegraphics[width=\textwidth]{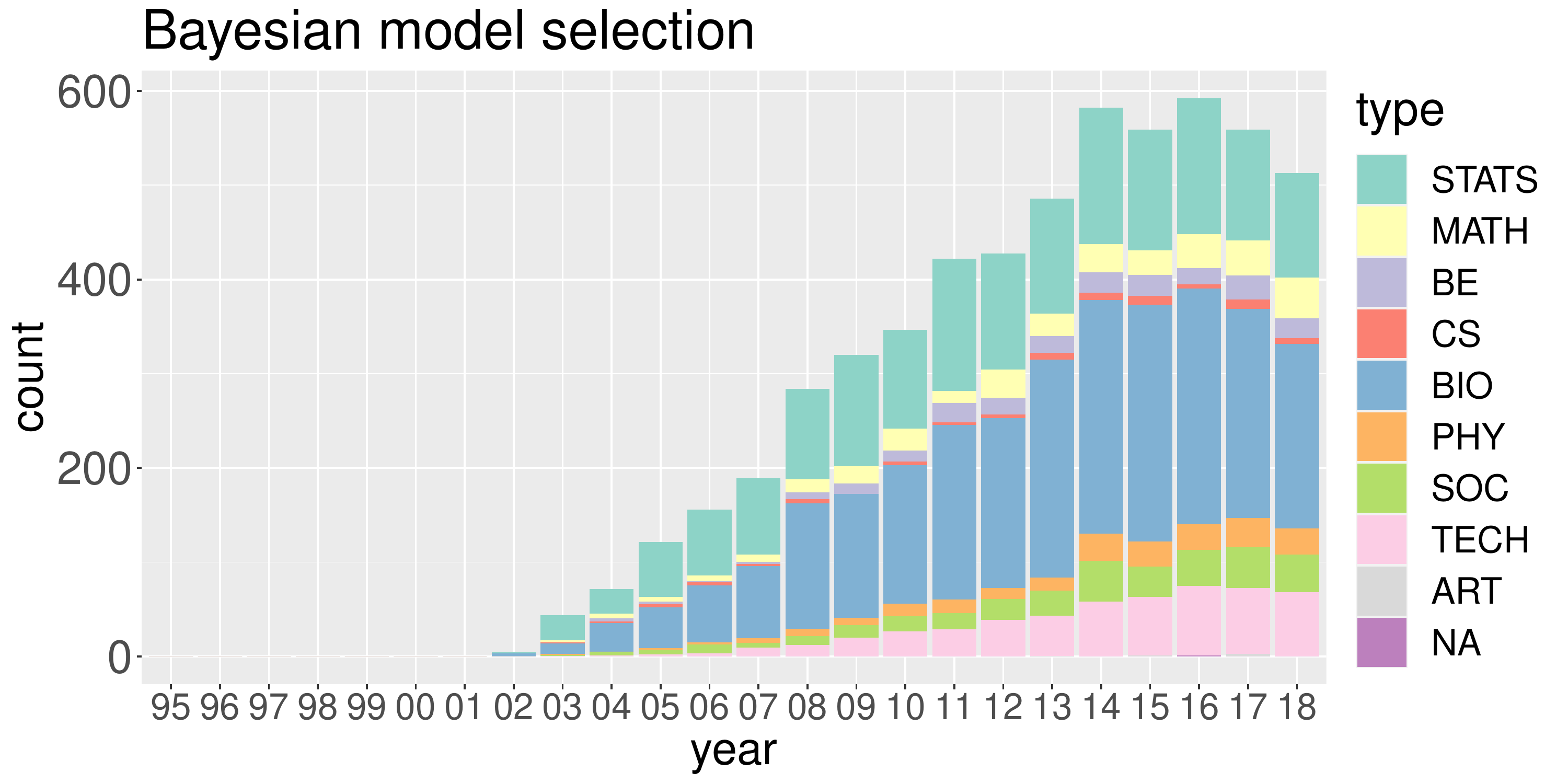}
    \caption{}
     \end{subfigure}
     \caption{Breakdown of citations for the papers in Table~\ref{table:top_20} aggregated by the five statistical topics: (a) MCMC, (b) Causal, (c) Penalized regression, (d) FDR, and (e) Bayesian model selection.}\label{fig:diversity_area}
\end{figure}

As can be observed in Figure~\ref{fig:rank_comp}, most papers have both high internal and external ranks. Table~\ref{table:top_20} lists all the papers that are ranked in the top 20 both internally and externally. 
We classify these papers roughly into five topics: Markov chain Monte Carlo (MCMC), causal inference (causal), penalized regression, false discovery rate (FDR), Bayesian model selection. To investigate the influence of these papers on other fields, we consider the aggregated citations by the five topics and break down the citations by category labels, similar to Figure~\ref{fig:overall_diversity}. In this case, we have added two category labels: ``BE" for the research area ``Business \& Economics" and ``CS" for the research area ``Computer Science", since we notice a considerable number of citations are from these two areas, especially for causal inference and penalized regression. To avoid double counting, papers with the BE (or CS) label will not be counted in SOC (or TECH), which is the broad category BE (or CS) belongs to in WoS. Similar to before, multiple labels for one paper are weighted equally. Figure~\ref{fig:diversity_area} shows that the influence on other fields differs by statistical research topics. FDR and Bayesian model selection have always attracted a substantial proportion of citations from BIO, even from the earlier years. MCMC and penalized regression have more citations from CS than the others. On the other hand, causal inference has the largest proportion of citations from SOC and BE among the five topics.

\section{Connecting statistical research communities to external topics by local clustering}\label{sec:local_cluster}

We have seen that different research topics in statistics often have different citation profiles by external fields, indicating they may have a heavier influence on some fields and topics and less so on others.  This prompts us to consider the question, given a specific external research topic, can we identify the most relevant statistical research topic (with relevance measured by our collected citation data)? This section investigates a local clustering approach by aPPR followed by appropriate cutoff selection. We present theoretical studies under the DC-SBM and results on simulated data. More importantly, we demonstrate the efficacy of the procedure on our citation data through several detailed case studies.

\subsection{Methods and theoretical guarantees}
\label{sec:main_method}

A typical local clustering method starts from one or multiple seed nodes and performs a random walk in the neighborhood of the seeds to gather other relevant nodes. In our setting, we first use keyword search to select a subset of citing papers, $\cI_t\subset \cI_c$, from an external topic of interest (see details in Section~\ref{sec:case_studies}). The seed nodes are constructed using citation information between the source papers $\cI_s$ and the topic papers in $\cI_t$, and the local clustering is performed on $\cI_s$ and their network $\bAs$. For clustering purpose, we consider two papers as related in content if a citation exists between them; the direction of this citation is less important if we think of it as a form of association. For this reason, we treat $\bAs$ as an undirected network in this section. That is, 



\begin{equation}\label{eq:undirect.adj}
A^s_{ij} =\begin{cases}
1\,,&\mbox{there is a citation between $i$ and $j$\,;}\\
0\,,&\mbox{otherwise}
\end{cases}
\end{equation}
for $i,j \in \cI_s = \{1, \ldots, 9338\}$. 

Next we present the details of the local clustering procedure and its theoretical properties under a network model with community structure. Standard order notations $\mathcal{O}$, $\Omega$,  $\mathcal{O_p}$ and $\Omega_p$ will be used throughout.

\subsubsection{Preliminaries and the DC-SBM}\label{sec:intro_SBM}


In order to analyze the behavior of local clustering, we adopt the popular DC-SBM \citep{KN11}, which captures both node heterogeneity and community structure, as the underlying network model. While such a model may not capture all the features of our citation network, the presence of node heterogeneity is reflected by the uneven distribution of citation counts, and it is plausible to assume the underlying communities correspond to different research topics. For convenience of notation, we will describe the DC-SBM and local clustering procedure using a general symmetric adjacency matrix $A$ and a general set of nodes $\cI$, with the understanding that they refer to $\bAs$ and $\cI_s$ in our data analysis. 

In the original SBM 
\citep{holland1983stochastic}, $N$ nodes are assigned to $K$ blocks or communities, and the probability of an edge between two nodes only depends on their community memberships. To abbreviate notations, write the set $\{1, \ldots, n\}$ as $[n]$ for any integer $n$. The set of nodes 
$\cI =[N]$ is partitioned into K blocks by the function $g: [N]\rightarrow [K]$. Let $n_k$ denote the size of block $k$, $\cI_k$ denote the set of nodes in block $k$ for $k \in[K]$. The proportion of members in block $k$  is $\tau_k = n_k/N$. We consider the case that the number of blocks $K$ is fixed, and $\tau_k$ is bounded below by a constant for all the $k \in [K]$. The probability of an edge between nodes $i$ and $j$ is 
$$A_{ij} \mid g \, \overset{\mbox{ind.}}{\sim} \, \mbox{Bernoulli}\left( B_{g(i)g(j)} \right),\, \forall i,j \in \cI,\, \, i\neq j\,,$$
where $B \in [0,1]^{K \times K}$ is the connectivity matrix. We adopt the common parametrization for $B$ as $B = \rho_N S$, where $S$ is a fixed $K \times K$ matrix, and $\rho_N$ is the average edge density satisfying $\rho_N \rightarrow 0$ at some rate as $N \rightarrow \infty$.
 
DC-SBM introduces node heterogeneity by adding a degree parameter $\theta_i$ for each node $i$, so that the probability of an edge between $i$ and $j$ becomes
 \begin{equation}\label{eq:dc_sbm}
 A_{ij} \mid g,\, \theta \, \overset{\mbox{ind.}}{\sim} \, \mbox{Bernoulli}\left( \theta_i \theta_j B_{g(i)g(j)} \right),\, \forall i,j \in \cI,\,\, i\neq j. 
 \end{equation}
Some constraint is needed on $\theta_i$ for identifiability, and we adopt the constraint
$\sum_{i \in \cI_k} \theta_i = n_k$ for all $k\in [K]$ following \citet{KN11}. 

 The degree of node $i$ is defined as $d_i = \sum_{j \in \cI} A_{ij}$. The population adjacency matrix is the conditional expectation of $A$, i.e.,
 $$\sA = \mathbb{E} [ A \mid g,\, \theta]\,. $$
It follows then the population node degrees are $\cd_i = \sum_{j \in \cI} \sA_{ij}$, and the expected degree is $\lambda_N = \frac{1}{N}\sum_{i \in \cI} \cd_i$.  Note that we also have $\lambda_N = N \rho_N$.

\subsubsection{Adjusted Personalized PageRank under the DC-SBM}\label{sec:aPPR}

Given an adjacency matrix $A$, define the diagonal matrix $D = \mbox{diag}(d_1, \ldots, d_N)$ and the graph transition matrix $P= D^{-1}A$. The \textit{personalized PageRank (PPR) vector} $p \in [0,1]^N$ is the stationary distribution of the process
$$ p^\top = \alpha \pi^\top + ( 1 - \alpha) p^\top P \,,$$
where  $\alpha \in (0,1]$ is the teleportation constant, and $\pi \in  [0,1]^N$ is a probability vector called the \textit{preference vector} encoding one or multiple seed nodes. For example, if there is one seed node $v_0=1$, $\pi = (1,0, \ldots, 0)^\top$. Under a network model with community structure such as SBM or DC-SBM, the goal is to recover all the nodes with the same community membership as $v_0$ by ranking the elements in the PPR vector $p$.


In our setting, we choose source papers that have high citation counts by a set of topic papers as the seed nodes.  For a source paper $j \in \cI_s$ and a set of topic papers $\cI_t$, its citation count is $ a_j = \sum_{i \in \cI_t} A_{ij}$, where $\bA$ is the citation network defined in Eq~\eqref{eq:adj}. The preference vector $\pi \in [0,1]^{9338}$ is calculated as 
\begin{equation}\label{eq:pi_vector}
\pi_k = \frac{a'_{k}}{\sum_{ j \in \cI_s}a'_{j}} \mbox{ where } a'_j = \begin{cases}
    a_j \,,&  a_j \geq t\,;\\
    0 \,,& a_j < t\,.
    \end{cases} 
\end{equation}
Here $t$ is a chosen threshold constant. We extend the setting of a single seed node in \cite{KUK17} and \cite{CZR20} to multiple seed nodes, but still make the assumption that they all belong to the same community. While it is unlikely that all papers cited by a specific topic come from the same community, the threshold $t$ helps us prune the vector $\pi$ and make the assumption more reasonable.  


Related to PPR, the \textit{adjusted personalized PageRank} (aPPR) vector is defined as
\begin{equation}\label{eq:appr_def}
 p^*_i = \frac{p_i}{d_i} \mbox{ for } i = 1,\,\ldots, N\,,   
\end{equation}
where $p_i$ is the $i$th entry in the PPR vector. \cite{CZR20}
showed that under the DC-SBM, adjusting by the degrees leads to a consistent ordering of the entries in $p^*$ so that  entries with the highest values belong to the target community. Formally, let $n$ be a community size cutoff. Then $n$ nodes with the largest $p^*_i$ values are selected as members in the target community, that is
\begin{equation}\label{eq:appr_select}
 \cC_{n} = \{ i \mid p^*_i \geq p^*_{(n)}\}\,,
\end{equation}
where $p^*_{(1)}, \ldots, p^*_{(N)}$ is the sorted list of  $p^*$ in a  non-increasing order.

Corollary 1  in \citet{CZR20} shows that with $v_0=1$ (assuming without loss of generality it belongs to block 1), provided we know the correct size cutoff $n=n_1$ (recall $n_1 = |\cI_1|$), then the aPPR clustering can recover all the nodes in block 1 with high probability, i.e., $\cC_{n_1} = \cI_1$. Since our setting makes use of multiple seed nodes, the follow proposition extends their result to better fit our situation.

\begin{proposition}\label{prop:aPPR_clustering}
Under the DC-SBM, given a set of seed nodes from the same block, say block 1, and a teleportation constant $\alpha$, assume that 
\begin{enumerate}[label=(c.\arabic*)]
    \item\label{prop.dcsbm.assp.3}  $\min_{u \in [N]} \theta_u \geq L_\theta$ and $\max_{u \in [N]} \theta_u \leq U_\theta$  where $L_\theta$ and $U_\theta$ are positive constants. 
    \item\label{prop.dcsbm.assp.4} For some sufficient large constant $c_1 >0$, $$\lambda_N > c_1 \left(\frac{1 - \alpha}{\Delta_\alpha} \right)^2 \log N\,. $$
\end{enumerate}
Then for sufficiently large $N$, with probability at least $1 - \mathcal{O}(N^{-5})$, we have $\cC_{n_1} = \cI_1$. 
\end{proposition}

Here $\Delta_\alpha$ is an increasing function of $\alpha$. The exact form of $\Delta_\alpha$ and the proof of Proposition~\ref{prop:aPPR_clustering} are deferred to Supplementary Materials~\ref{app.sec:prop_1}.

\subsubsection{Conductance}\label{sec:conductance}

Given the result in Proposition~\ref{prop:aPPR_clustering}, it remains to choose the correct size $n$ for $\cC_n$ to fully recover the target community (block 1). To achieve this, an objective function is needed to evaluate the quality of the clusters found. Conductance is a popular objective function to be optimized either globally or locally \citep{VM16, YL15} and often used in conjunction with a local clustering algorithm like PPR \citep{andersen2006communities, wu2012learning}. It tends to favor small clusters weakly connected to the rest of the graph, and one would expect such an assortative structure in citation networks with communities defined by research topics.

For a set of nodes $\cI' \subseteq \cI$, we define its conductance $\phi$ as
\begin{equation}\label{eq:conductance}
  \phi(\cI')  = \frac{\sum_{i \in \cI' }\sum_{j \notin \cI' } A_{ij} }{\sum_{i \in \cI' } A_{i\cdot}}\,,   
\end{equation}
where $A_{i\cdot} = \sum_{j \in \cI} A_{ij}$. The numerator is known as the cut of the graph partitioned by $\cI'$ and its complement $(\cI')^c$, while the denominator represents the volume $vol(\cI', \cI)$. We note that an alternative form of conductance has $\min\{ vol(\cI', \cI), vol((\cI')^c, \cI) \}$ in the denominator. The two forms are equivalent when the size of $\cI'$ is smaller than $(\cI')^c$, a condition we expect to hold for $\cC_{n_1}$ and its neighborhood, given $n_1$ is small compared with $N$. Hence we choose the form in Eq~\eqref{eq:conductance} for easier bookkeeping.


Proposition~\ref{prop:aPPR_clustering} demonstrates that the aPPR vector sorts the nodes in terms their relevance to the target community with high probability. The sorted list of nodes leads to a sequence of clusters $\{\cC_n\}_{n=1}^N$ and their conductance values $\{\phi(\cC_n)\}_{n=1}^N$. Our next theorem establishes that the correct choice of $n$ occurs at a local optimum along this sequence, justifying the practice of choosing the community size cutoff by inspecting the conductance plot.
 
\begin{theorem}\label{thm:min}
Under the DC-SBM, suppose that \ref{prop.dcsbm.assp.3} and \ref{prop.dcsbm.assp.4} hold. Recall $S$ is the fixed component in the connectivity matrix $B$, and $\tau$ is the set of block proportions (Section~\ref{sec:intro_SBM}).
Further assume that
\begin{enumerate}[label=(c.3)]
    \item \label{prop.sbm.assp.2} given $S$ and $\tau$, 
    \begin{equation}\label{eq:assp}
     S_{11}\min_{i > 1}\tS_{i\cdot} >  2  \max_{j > 1} S_{1j} \tS_{1\cdot} \mbox{ where } \tS_{i\cdot} = \sum_{j} S_{ij} \tau_j\,.
    \end{equation}
\end{enumerate}
Then for sufficiently large $N$, there exits $n'$ with $n' - n_1 = \Omega(N)$ such that
\begin{equation}\label{eq:local_opt}
    \phi(\cC_{n_1}) -  \phi(\cC_n) \leq - \frac{1}{N} \Omega_P( | n - n_1|)
\end{equation}
uniformly for $n \in [n']$.
\end{theorem}

The proof of Theorem~\ref{thm:min} has two major parts. The first part in Supplementary Materials~\ref{app.sec:prop_SBM} analyzes the optimality properties of $\phi$ at the population level with the help of the result in Proposition~\ref{prop:aPPR_clustering}. The second part in Supplementary Materials~\ref{app.sec:thm_min} incorporates noise from the adjacency matrix and proves the local optimality result in the theorem. We make two remarks as follows.
\begin{enumerate}[label=\alph*)]
   \item 
   The bound in Eq~\eqref{eq:local_opt} and the lower bound on $n'-n_1$ guarantees the optimum at $n_1$ is well separated from its neighborhood and this neighborhood is wide enough to be observed in a conductance plot. 
   
   
   \item
   As can be seen from the proofs in Supplementary Materials~\ref{app.sec:prop_SBM}, the larger the gap between $S_{11}\min_{i > 1}\tS_{i\cdot}$ and $2  \max_{j > 1} S_{1j} \tS_{1\cdot}$, the more peaked and easier to spot the local optimum is. In an assortative graph with $S_{ii} > \max_{j\neq i}S_{ij}$, a smaller $\tau_1$ will lead to a smaller $\tS_{1\cdot}$ in Eq~\eqref{eq:assp}, making the inequaltiy more easily satisfied. Thus using conductance as a objective function is well suited to the situation where $n_1$ is a small fraction of $N$.
   
\end{enumerate}

\spacingset{1}  
\begin{algorithm}[htb!]
\caption{Local clustering} \label{alg:local_clustering}
\SetKw{KwBy}{by}
\SetKwInOut{Input}{Input}\SetKwInOut{Output}{Output}

\SetAlgoLined

\Input{adjacency matrix $A$, preference vector $\pi$, and teleportation constant $\alpha$.
}

Compute the aPPR vector $p^*$ in Eq~\eqref{eq:appr_def} based on $(A, \pi, \alpha)$.

Construct the sequence of clusters $\{\cC_n\}^N_{n = 1}$ according to Eq~\eqref{eq:appr_select} and $p^*$. 

Calculate conductance values  $\{\phi(\cC_n)\}^N_{n = 1}$ by Eq~\eqref{eq:conductance}.

Find the first local minimum $\phi(\cC_{n^*})$ in $\{\phi(\cC_n)\}^N_{n = 1}$.

\Output{local cluster  $\cC_{n^*}$.}
\end{algorithm}

\subsection{Simulations}\label{sec:simulation}

 In this section, we examine the performance of our local clustering procedure as summarized in Algorithm~\ref{alg:local_clustering} on data simulated from the DC-SBM. We focus on the case where the target community (block 1) is small compared with the whole graph as we consider it to be more relevant to our real data structure.

Consider a DC-SBM with $K = 2$, $n_1 = 50$, $n_2 = 3000$, and \begin{equation*}
B = \left(\begin{array}{cc}  0.05  &   0.01  \\  0.01  & 0.05 \end{array}\right). \,
\end{equation*}
To simulate the degree parameters, let $\eta_i \sim \mbox{Uniform} (1, 10)$ for $i = 1, \ldots, n_1 + n_2$, and
\begin{equation*}
   \theta_i =  \begin{cases}
   n_1 \eta_i / \sum^{n_1}_{j = 1} \eta_j \,,& i = 1, \ldots, n_1 \,;\\
   n_2 \eta_i / \sum^{n_1 + n_2}_{j = n_1 + 1} \eta_j \,,& i = n_1 +1 , \ldots, n_1 + n_2  \,
    \end{cases}
\end{equation*}
so that they satisfy the identifiability constraint on $\theta_i$. 

\spacingset{1}
\begin{table}[!ht]
\centering
\footnotesize
\begin{tabular}{l l|llll}
\hline\hline
&  & \multicolumn{2}{c}{Precision}  & \multicolumn{2}{c}{Recall}  \\ \hline
 $\alpha$ & \# of seeds & Mean(\%) & SD(\%) & Mean(\%) & SD(\%) \\ \hline

 \multirow{4}{*}{0.15}  & 1   & 97.00 & 0.15 & 98.77 & 1.54 \\
                        & 5   & 97.49 & 0.22 & 99.23 & 0.94\\
                        & 10  & 97.93 & 0.22 & 99.42 & 0.63 \\
                        & 15  & 98.30 & 0.30 & 99.50 & 0.46 \\ 
                        & 20  & 98.60 & 0.40 & 99.56 & 0.37 \\\hline \hline
\end{tabular}
\caption{The means and standard deviations of precision and recall rates for local clustering under $\alpha=0.15$ and different numbers of seeds. Each setting is simulated 50 times.}\label{table:sim_results_1}
\end{table}

We investigate the effect of the teleportation constant $\alpha$ and the number of seed nodes, denoted $m$, on the accuracy of the local clustering results. 
When using $m$ seed nodes, the corresponding preference vector have $\pi_i = 1/m$ for $i\in[m]$, and $\pi_i = 0$ for the other entries. To determine the community size, we search for the first obvious local minimum in the conductance plot, and we find these optimal points usually occur before $n<55$. Supplementary Figure~\ref{fig:conduct_exp} provides examples of the conductance plots for $\alpha=0.15$ and different $m$ values; the cases for other $\alpha$ are similar. 


\spacingset{1}

\begin{figure}
 \centering
     \begin{subfigure}[b]{0.45\textwidth}
    \centering
    \includegraphics[width=\textwidth]{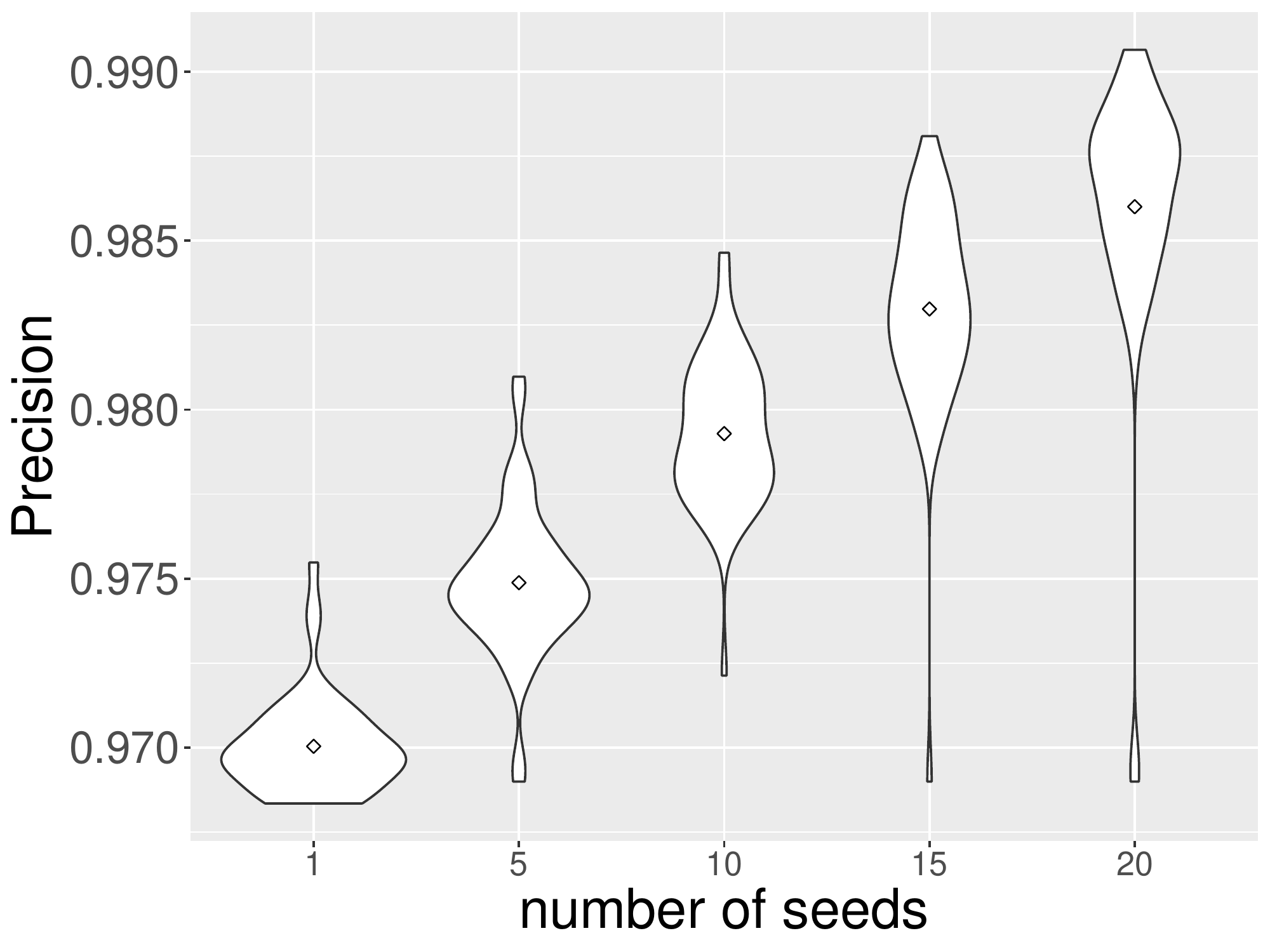}
    \caption{}
     \end{subfigure}
     \begin{subfigure}[b]{0.45\textwidth}
    \centering
    \includegraphics[width=\textwidth]{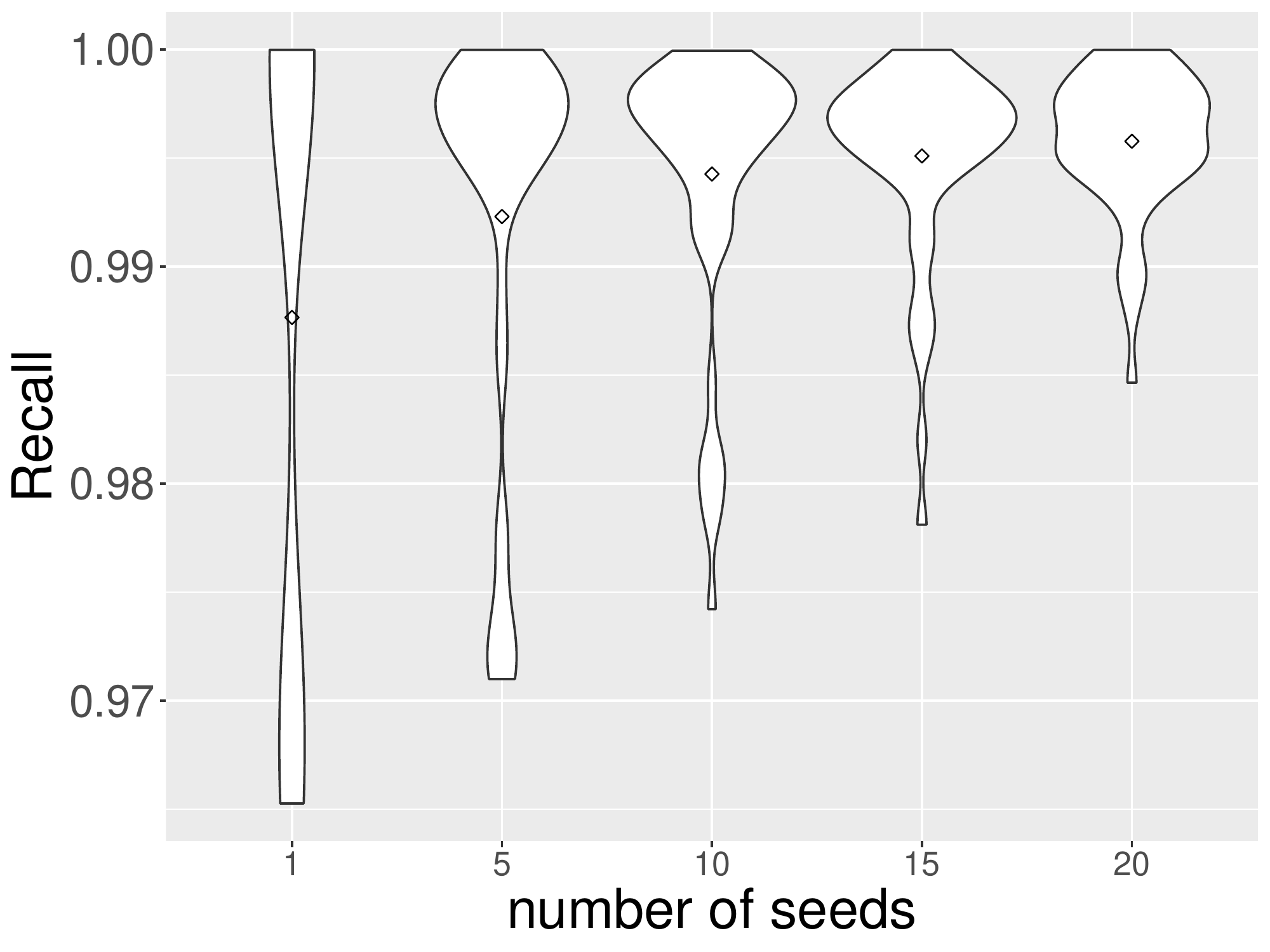}
    \caption{}
     \end{subfigure}
  \caption{The violin plots of (a) precision and (b) recall for local clustering with $\alpha = 0.15$. }\label{fig:violin_2}
\end{figure}

Table~\ref{table:sim_results_1} shows the average precision and recall rates and their standard deviations for finding members in block 1 with $\alpha = 0.15$ and five seed counts (1, 5, 10, 15, 20); each setting is repeated for 50 simulations. We can see that the precision increases as the number of seeds increases, since more seeds will provide more initial information for the clustering. More seeds also help to stabilize the variance of recall and increase the mean recall by a smaller margin. On the other hand, the influence of $\alpha$ is rather minimal. The results from different $\alpha$ values (0.05, 0.25) are  presented in Supplementary Table~\ref{table:sim_results_2}. In Figure~\ref{fig:violin_2}, we further illustrate the distributions of these precision and recall values for the case $\alpha=0.15$. For all the case studies from the citation data in the next section, we set $\alpha = 0.15$ as in \citet{CZR20}.


For the sake of completeness, in Supplementary Materials~\ref{app.sec:more_simulation}, we compare the local clustering procedure with commonly used global clustering techniques including spectral clustering and SCORE \citep{J15} for different values of $n_1$. As expected, local clustering is better suited to the situation with smaller $n_1$.

\subsection{Case studies from the citation network} \label{sec:case_studies}

Next we apply local clustering to our citation data and use the procedure to find the most relevant statistical research areas for given external topics. We choose three external topics (single-cell transcriptomics, labor economics and flu) of high general interests spanning biology, economics and epidemiology, and discuss the results in detail. More examples of topics and their clustering results can be found in Supplementary Materials~\ref{sec:more_cases}.

Before applying Algorithm~\ref{alg:local_clustering}, it remains to describe the selection of topic papers $\cI_t$ and the construction of the preference vector $\pi$. (Recall that the adjacency matrix used here is $\bAs$ as described in Eq~\eqref{eq:undirect.adj}.) For each external topic, papers in $\cI_t$ are chosen by keyword searches among the citing papers. More concretely, for the topics of single-cell transcriptomics and labor economics, we find citing papers that contain the relevant keywords \footnote{``Single-cell" (or ``single cell") and ``RNA-seq" for the topic of single-cell; ``labor" for the topic of labor economics.} in their abstracts. For a more accurate search result, we further restrict the labor economics papers to the category SOC using the labels in Section~\ref{sec:diversity}. The single-cell papers can come from a more diverse set of categories, and as shown in Supplementary Figure~\ref{fig:pie_single_cell}, most of our selected papers are from BIO. For the topic of flu, we note that many papers may use flu datasets as examples of their analytic methods instead of focusing on the topic itself. To select papers with a sharper focus on the topic, we search for papers with ``flu'' or ``influenza'' in their titles instead of abstracts. The proportions of category labels for the flu papers are illustrated in Supplementary Figure~\ref{fig:pie_flu}, which indicates most of them are from BIO. Having constructed $\cI_t$, the seed nodes in $\pi$ are chosen from the source papers with high citation counts by $\cI_t$. More details on the choice of $t$ in Eq~\eqref{eq:pi_vector} for different topics can be found in
Supplementary Materials~\ref{app.sec:case_study}.

Supplementary Figure~\ref{fig:conductance_topic} contains the conductance plot for each topic and our choices of the local minimum. The size of the target community found for each topic is listed in Table~\ref{table:network_topic}. We can see that these subnetworks indeed have significantly denser connections (and in some cases, higher clustering coefficients) than the whole network. The subnetworks and the word clouds generated from the keywords of the subnetwork papers can be found in Figure~\ref{fig:case_1}. We discuss these in more details below, interpreting the results with our understanding of the topics. 

\spacingset{1}
\begin{table}[!ht]
\centering
\begin{tabular}{c|r|c | c}
\hline\hline
Topic & Size & Graph density  & Average clustering coefficient \\ \hline
single-cell & 79 & 0.031  & 0.608 \\ 
economic labor& 108 & 0.039   & 0.402    \\ 
flu & 30 & 0.73 & 0.232 \\
all source papers & $9{,}338$  & 0.001  & 0.252     \\ \hline
\end{tabular}
\caption{Summary statistics for the subnetworks in Figure~\ref{fig:case_1} compared with the global graph $\bAs$. }\label{table:network_topic}
\end{table}

\noindent\textbf{Single-cell transcriptomics}

Rapid advances in single-cell sequencing technologies in the past decade have enabled researchers to profile different aspects of an individual cell, in particular its transcriptome. After appropriate preprocessing,  a single-cell transcriptomic data usually takes the form of a large, sparse matrix, with tens of thousands of rows representing genes and columns representing cells. The sparse, noisy and heterogeneous nature of such data has proved a fertile ground for the development of statistical and computational methods (see e.g. \cite{Kh21} for a review). Inspecting the subnetwork and word cloud in Figure~\ref{fig:case_single_cell}, perhaps unsurprisingly, a significant fraction of the papers selected are concerned with multiple testing and connected to the hub node 79 \citep{benjamini1995controlling}. As an example, multiple testing is routinely performed in the analysis of single-cell RNA-seq (scRNA-seq) data for identifying differentially expressed genes, which involves applying a statistical test to a large number of genes to determine if their expression levels are significantly different between two sets of cells. 
The word cloud also suggests clustering as another main keyword; in the subnetwork, clustering is a topic shared by the set of papers tightly knit around node 35 \citep{sugar2003finding} and 78 \citep{tibshirani2001estimating}. In the analysis pipeline of scRNA-seq data, clustering is applied to a dimension-reduced scRNA-seq matrix to identify distinct subpopulations of cells, which can correspond to different cell types or states. The related feature selection and model selection problems are highly relevant in this context, as they help researchers determine genes (features) that distinguish these subpopulations and the total number of subpopulations observed.

\noindent\textbf{Labor economics} 

Labor economics aims to understand the functioning and dynamics of the markets for wage labor. Many fundamental questions in this subject---How does education affect income? How does healthcare affect income?---are of causal nature. Economists and governments would like to design policies that might achieve certain economic and social welfare goals based on causal analysis.  Randomized controlled trials (RCT) are usually not available for Labor Economics problems.  Therefore, it is not surprising to see that an overwhelming majority of the statistics papers selected in the subnetwork and word cloud in Figure~\ref{fig:case_labor} are in the realm of causal inference. Concretely, in the the word cloud, the frequently appearing keywords (minus ``test") are all technical terms in causal inference---``propensity score", ``instrumental variable", ``structure model", ``matched sampling", ``treatment effect", ``matching", and ``observational study". Notably, the node 78 \citep{angrist1996identification}, a hub in the subnetwork, links the structural equations framework in econometrics and the potential outcomes framework in statistics.  The paper  provides conditions for a causal interpretation of the instrumental variable (IV) estimand, and quantifies the bias of violations of the critical assumptions.
Moreover, many cited statistics papers (node 16 \citep{ganong2018permutation} and node 18 \citep{li2018balancing}) are rather recent, and they also appear in the subnetwork. This coincides with the recent surge of the study of causal inference in the statistical community in the last few years and offers some evidence that the new developments quickly penetrate into other research fields.

\noindent\textbf{Flu} 

The global pandemic of Covid-19 has further ignited wide research interests in the modeling and prediction of the spread of an epidemic. We choose flu as an example of epidemics due to its longer history of study and frequent appearance in the literature of epidemiology. (The results from using Covid-19 as the topic are presented in Supplementary Materials~\ref{sec:more_cases}.) As expected, many of the keywords in Figure~\ref{fig:case_flu} are related to stochastic processes and state-space modeling. The word MCMC appears the most often being a commonly used technique for parameter estimation in these epidemic models. Looking more closely at the subnetwork, many of the papers focus on refining the susceptible-infectious-recovered (SIR) model for infectious diseases including flu and SARS. For the two hub nodes 28 \citep{britton1998estimation} and 3 \citep{dukic2012tracking}, the former is concerned with the parameter estimation problem for different types of observed data, while the latter extends the SIR model by incorporating incubation stage and time dynamics to track the spread of flu.


\spacingset{1}
\begin{figure}
    \begin{subfigure}[b]{1\textwidth}
    \centering
    \vspace{-0.2cm}
     \scalebox{0.35}{\includegraphics{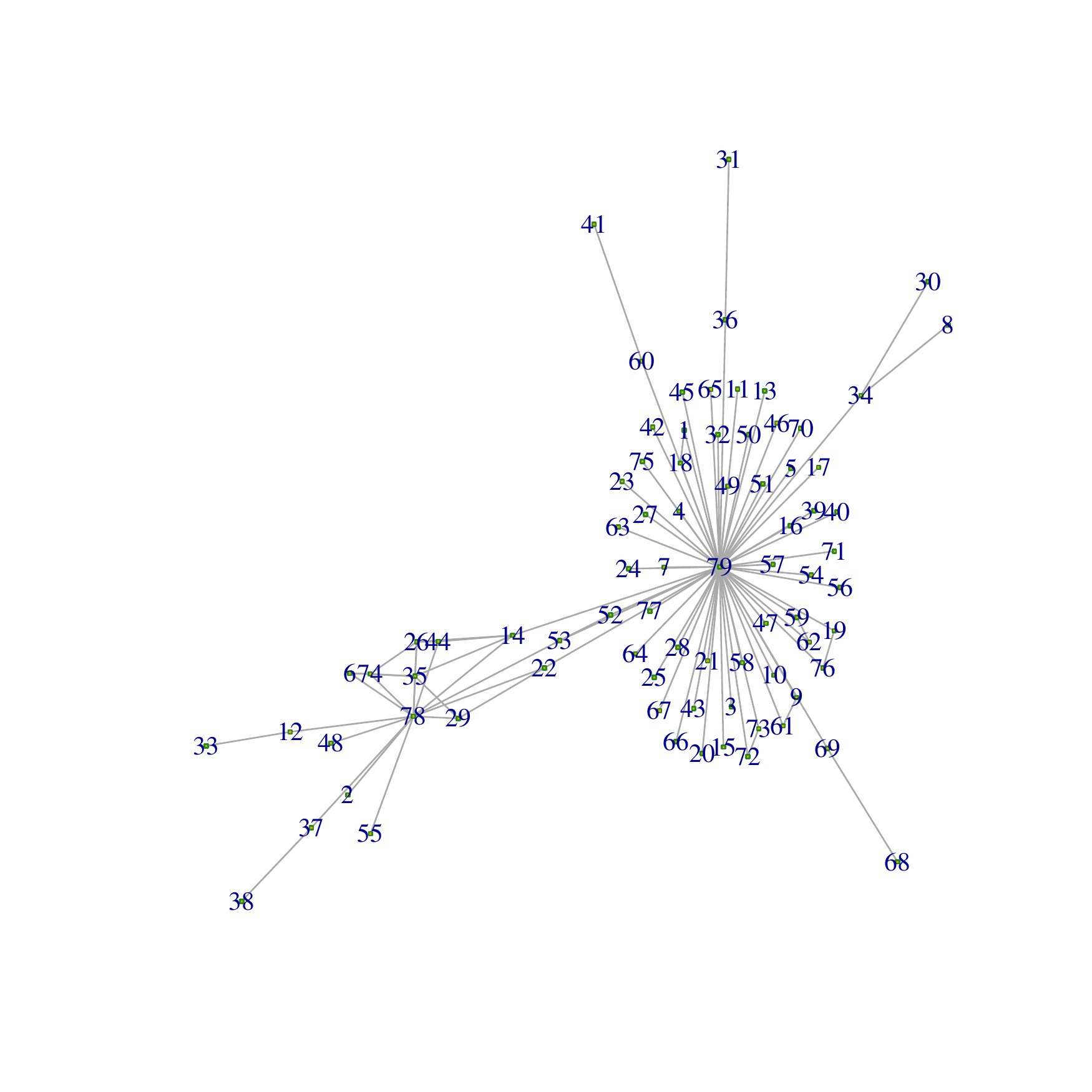}}
    \scalebox{0.4}{\includegraphics{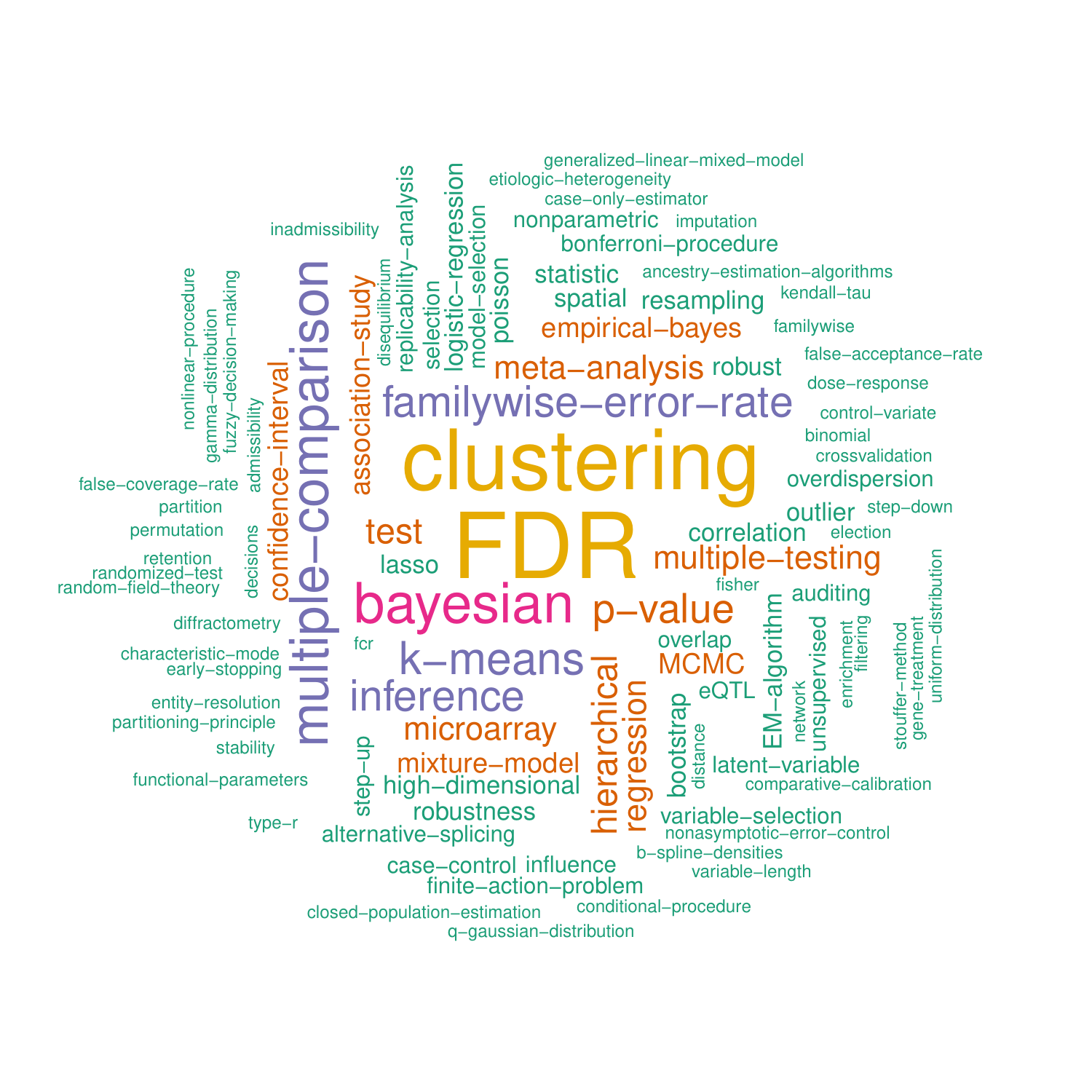}}  \vspace{-1cm}
    \caption{single-cell transcriptomics}\label{fig:case_single_cell} 
    \end{subfigure}
 \begin{subfigure}[b]{1\textwidth}

 \centering
     \scalebox{0.35}{\includegraphics{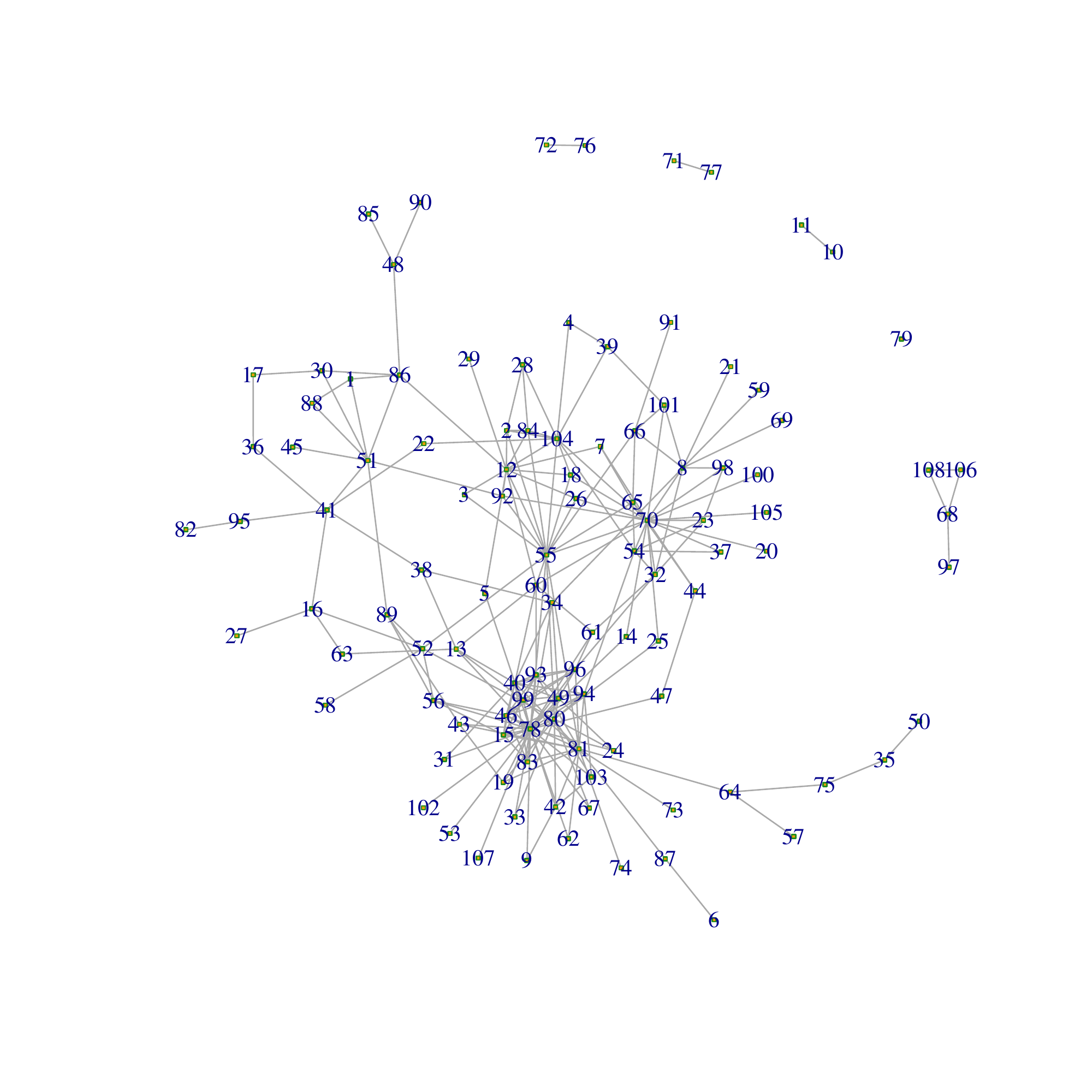}}
    \scalebox{0.4}{\includegraphics{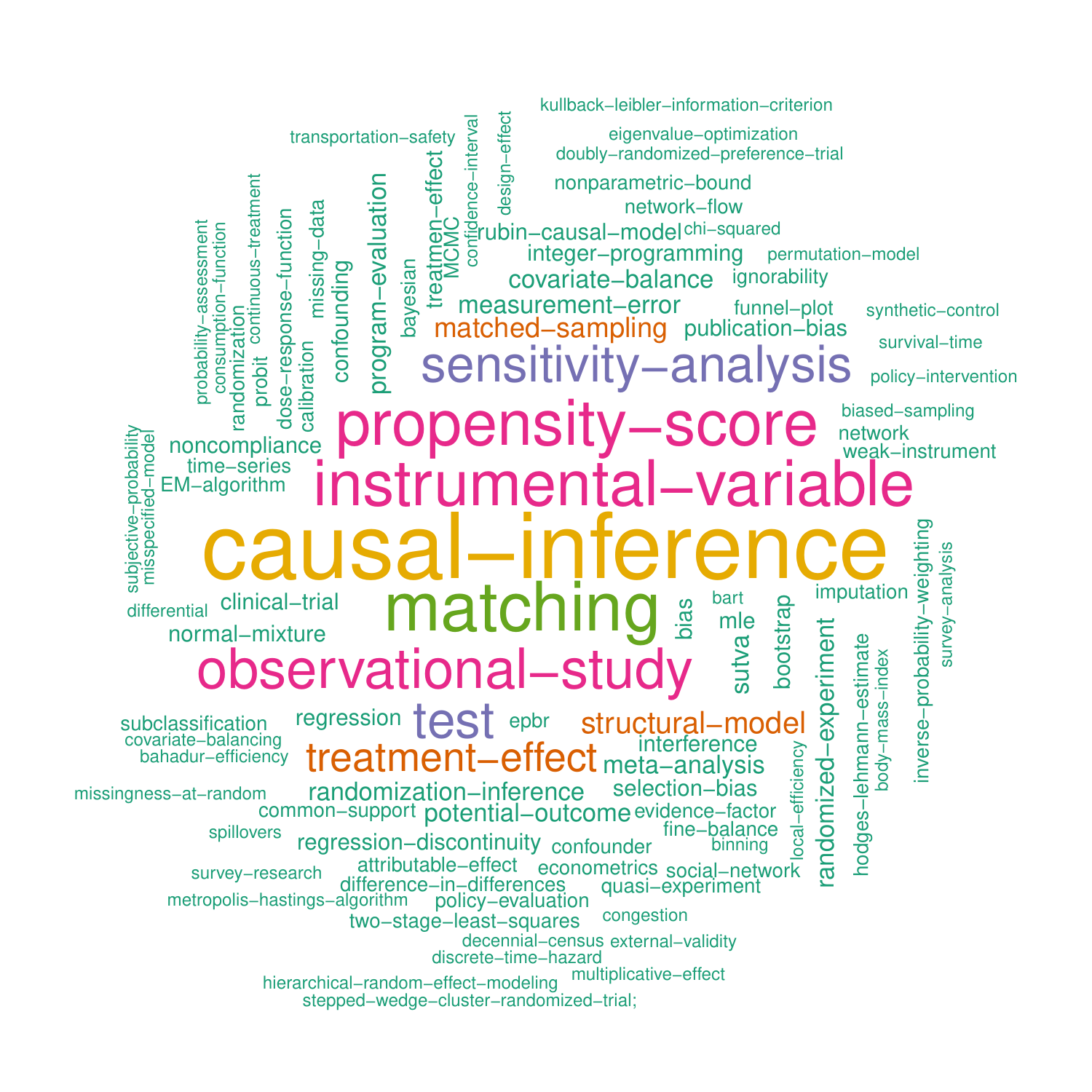}} \vspace{-0.5cm}
     \caption{labor economics}\label{fig:case_labor}
    \end{subfigure}
    \begin{subfigure}[b]{1\textwidth}
 \centering
     \scalebox{0.35}{\includegraphics{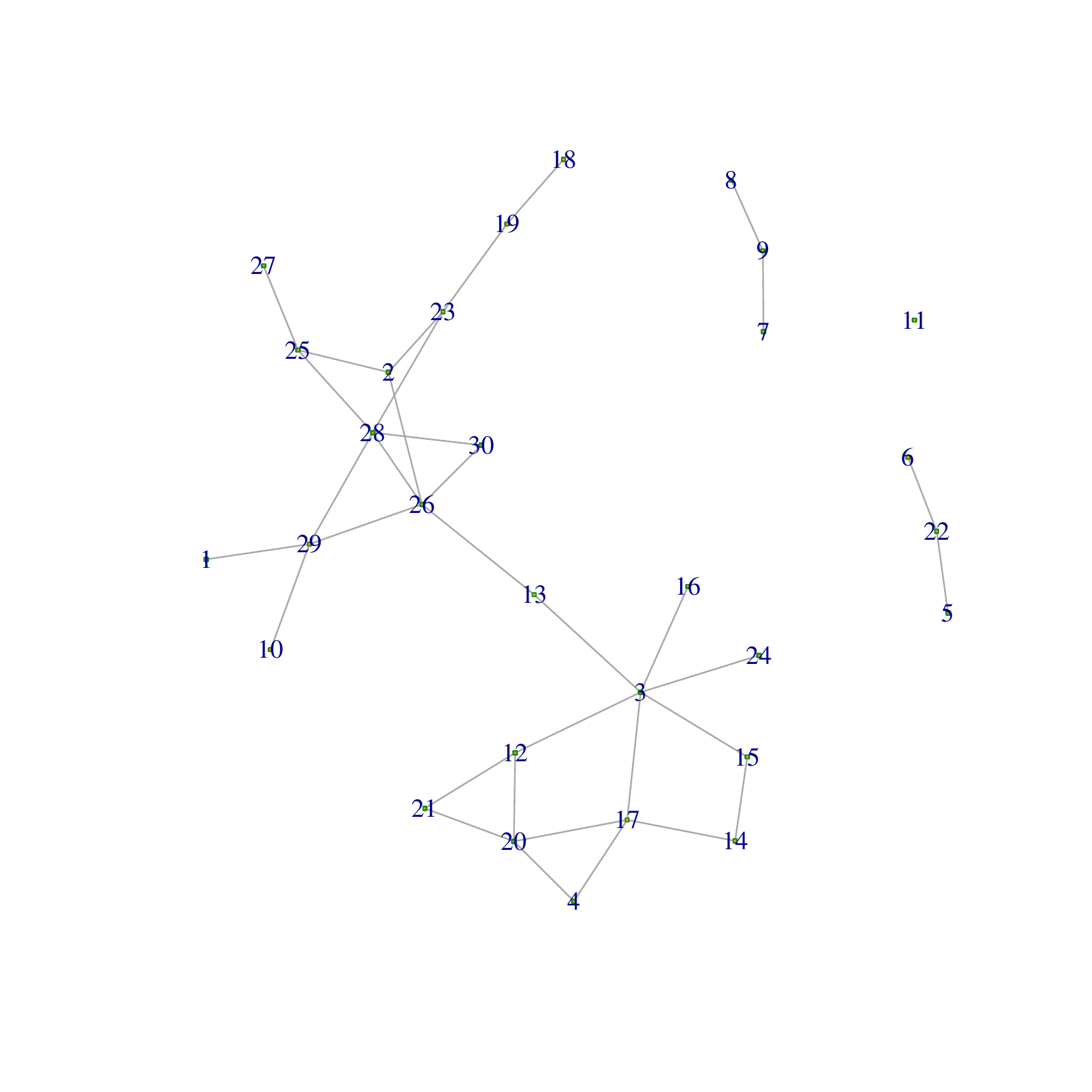}}
    \scalebox{0.4}{\includegraphics{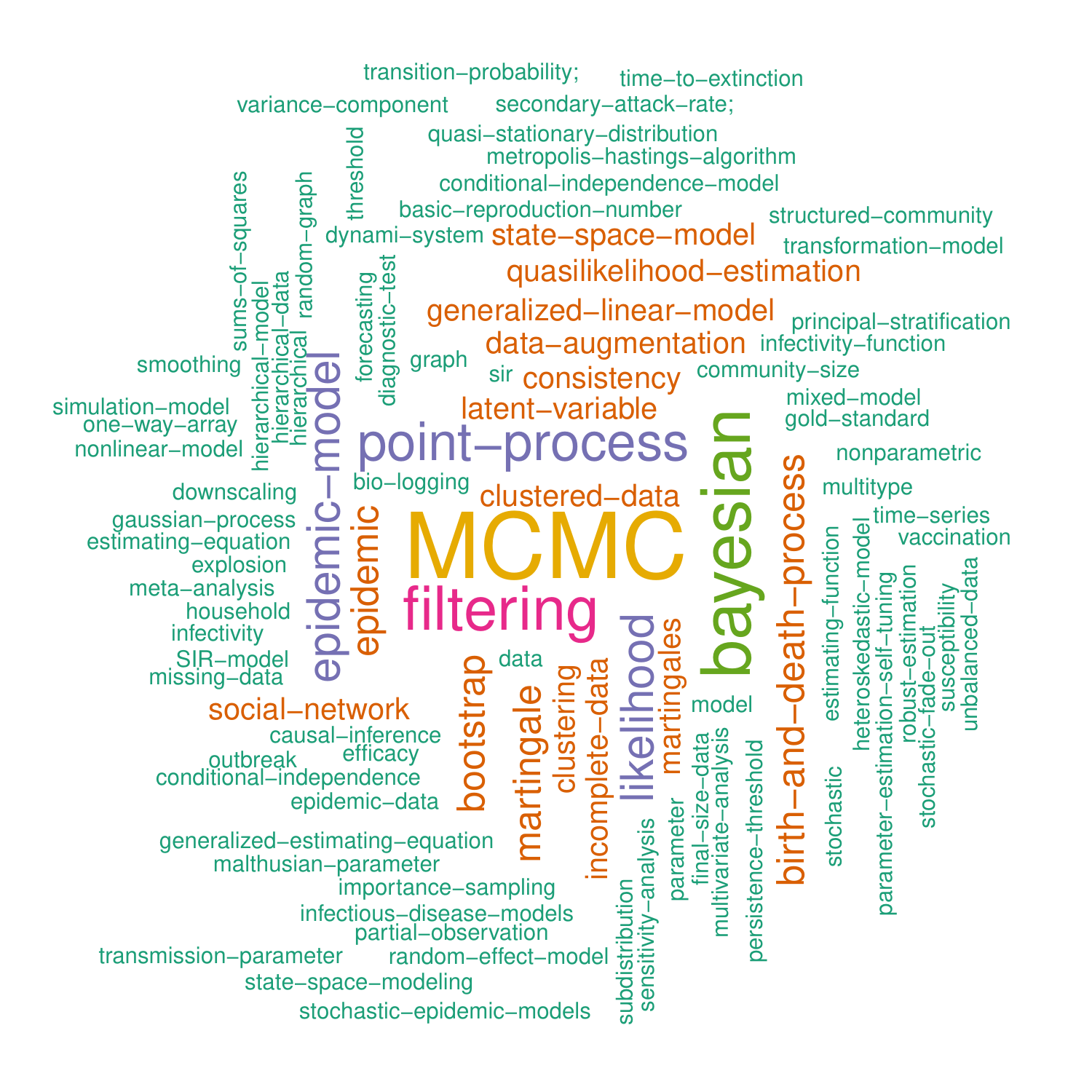}} \vspace{-0.5cm}
     \caption{flu}\label{fig:case_flu}
    \end{subfigure}
    \caption{Networks and word clouds generated from the source papers found by local clustering for each topic.}\label{fig:case_1}
\end{figure}

\section{Discussion}

In this paper, we study the citation network arising from selected statistical papers in the past two decades, a period coinciding with the rise of Big Data and statistics being perceived to play increasingly important roles in many scientific disciplines. Unlike previous studies on statistics citation networks, we focus on the connections between statistics and other disciplines and use citation data to investigate the external influence of various statistical works.  

First performing descriptive analysis, we show that both the overall volume of citations and the diversity of citing fields have been increasing over time for all the journals considered. Even typical theoretical journals such as AOS have been attracting a significant proportion of external citations in recent years, which is quite encouraging. Next by distinguishing between internal and external citations, we identify research areas in statistics that have high impact under both criteria. The most highly cited papers are ranked high both internally and externally. On the other hand, papers with a large number of external citations but relatively fewer internal citations can point to areas where future development in relevant theory and methods may be rewarded by immediate visibility outside statistics. Lastly, using the technique of local clustering, we identify the statistical research communities most relevant to various external topics of interest. Under the DC-SBM, we prove the combination of aPPR and conductance selects all nodes in the target community with high probability. We demonstrate the performance of the algorithm using simulated data, examining its stability with respect to the number of seeds and the teleportation constant. Presenting a number of case studies using external topics of high general interests, we show that the communities selected align well with our intuition and understanding of the topics.  

Our study takes the first step toward understanding the influence of statistical works on other disciplines that use tools and methods from statistics to aid their  discoveries. The data we have collected can be of independent interests, opening opportunities for further modeling and analysis from different perspectives. We also note that some of the limitations in our current study can be addressed by expanding the scope of the data. For example, in analyzing the trend of diversity of citing  fields, it would be ideal to collect information about the number of published papers in each citing field and include it as a normalization factor. The data could also be expanded to include more journals and other types of source publications, such as conferences and books, over a longer period of time to allow for a more comprehensive historical view and richer analysis. 
We leave the collection and analysis of these more extensive data as future work. 

Compared with global clustering, the theoretical properties of local clustering techniques are less well characterized under generative network models. Our application and theoretical results of local clustering can be extended to incorporate mixed membership modeling and temporal changes in the evolution of communities. We have currently used textual data (e.g., keywords) as a way to validate the target communities found; it would be more interesting to include such data as covariates in the network model subject to clustering analysis. 

We end the discussion by acknowledging the limitations of citation itself as a form of data measuring intellectual influence, some of which have already been pointed out in previous studies \citep{stigler1994citation, varin2016statistical}. Not all citations carry the same weight -- a paper could be mentioned just in the literature review or serve as the foundation that inspired the paper citing it; arguably the latter type of citation is more important. Citations are not always attributed to the correct source, and modern day style of research relying on search engines such as Google is likely to bias toward papers already with high citation counts. Many data scientists and practitioners in industry do not necessarily publish their works but can still make use of ideas and tools in statistical papers, resulting in missing citations.   
Nevertheless, despite these limitations, citation data provide a useful and necessary first passage into investigating the intellectual influence of scientific works.  

\section*{Acknowledgments}
The authors would like to thank Dr. Tung-Yu Wu for help with the data collection process and Prof. Peter J. Bickel, Prof. Jingyi Jessica Li for many fruitful discussions. Y.X.R.W. gratefully acknowledges funding from the Australian Research Council DECRA Fellowship (DE180101252).

\bibliographystyle{plain}
\bibliography{citationref}

\appendix
\appendixpage

\section{Additional results for descriptive analysis of the citation network}\label{app.sec:summary}

\spacingset{1}
\begin{figure}[h]
\centering
 \includegraphics[width=13cm]{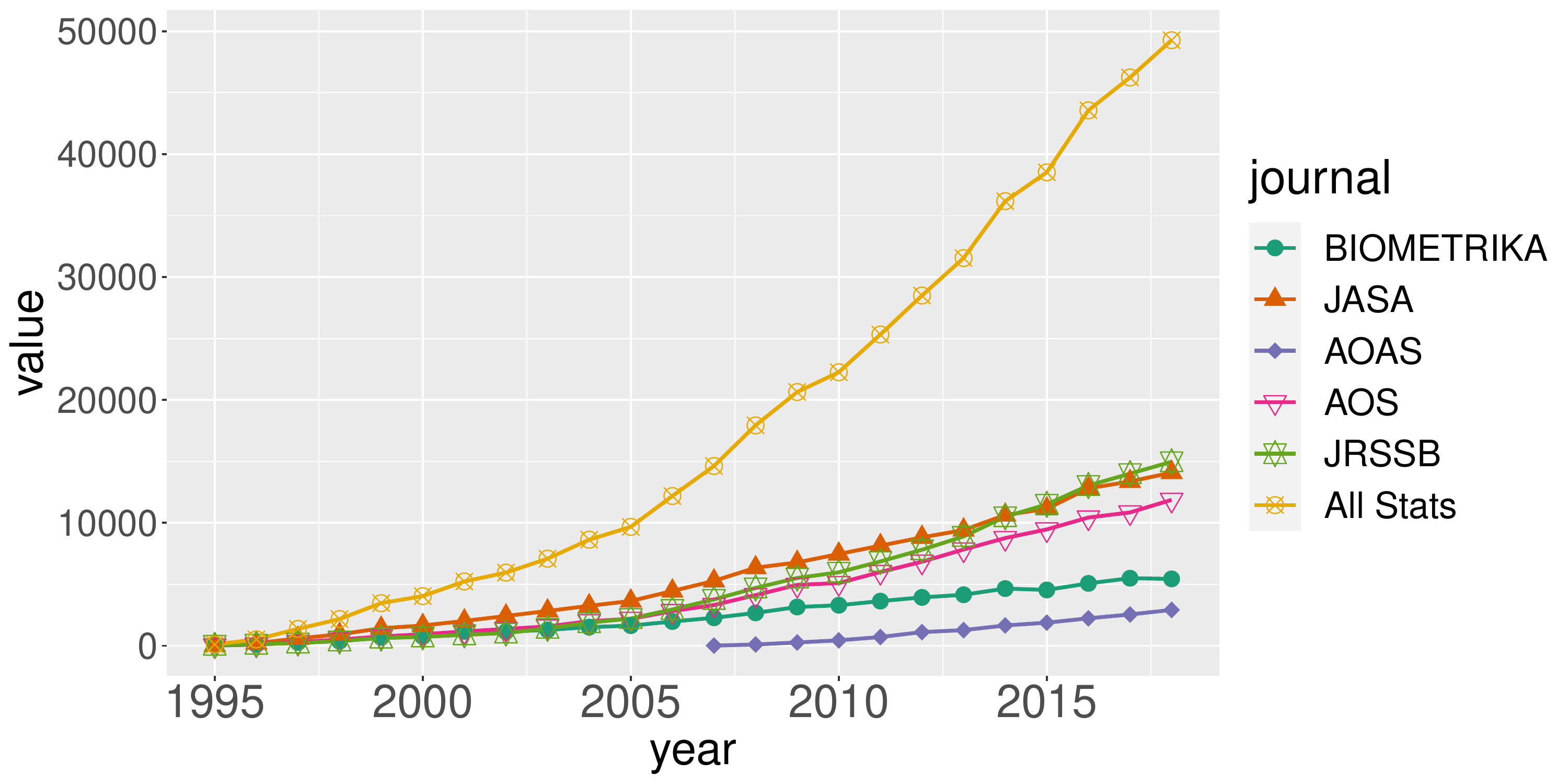}
 \caption{The total number of citations per year for each journal.}\label{fig:total_citation}
\end{figure}

\begin{figure}[h]
\centering
\includegraphics[width=13cm]{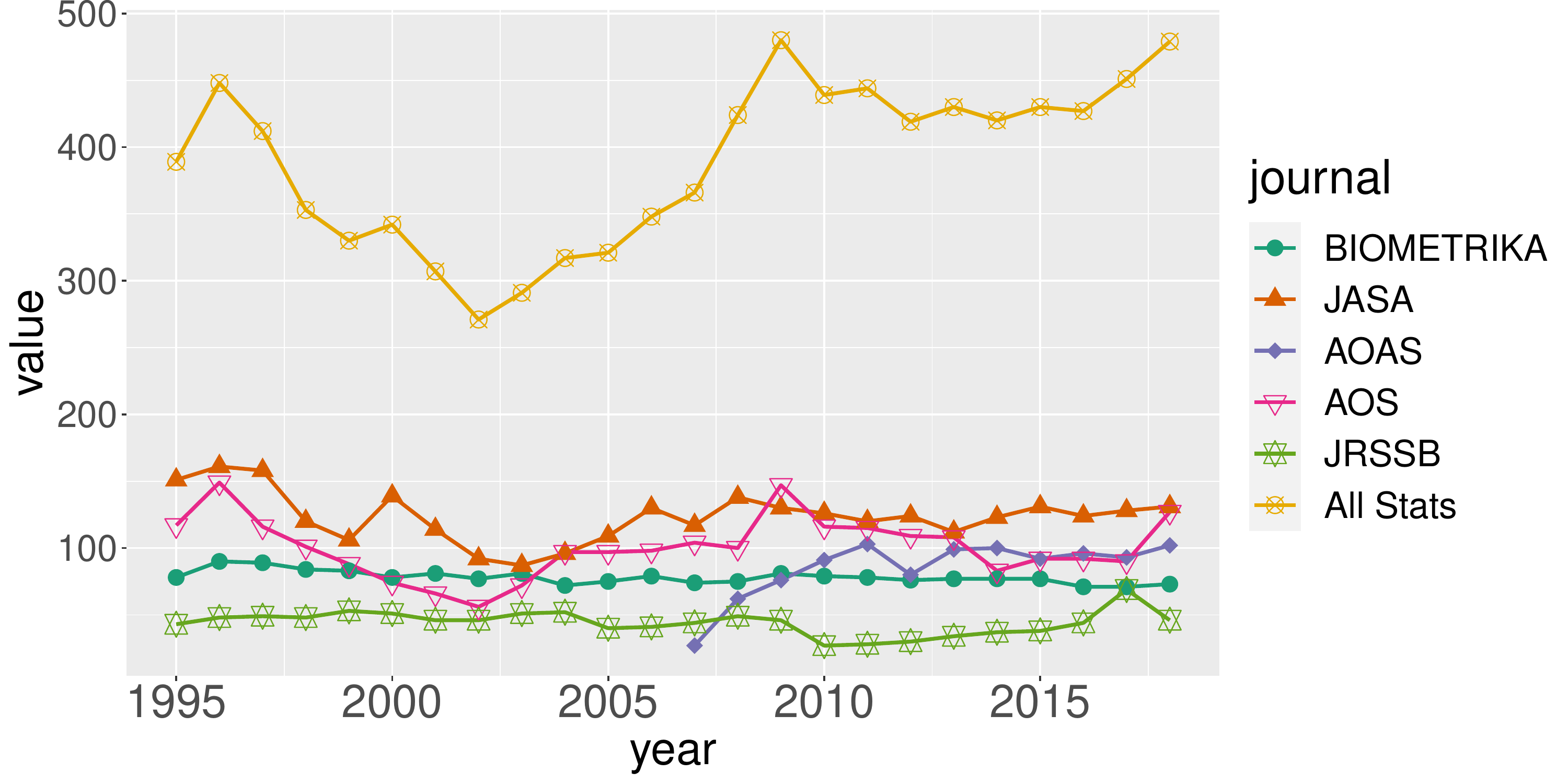}
\caption{The total number of publications per year for each journal.}
\label{fig:volume}
\end{figure}

\vspace{-0.5cm}
\begin{figure}[h]
\centering
 \includegraphics[width=12cm]{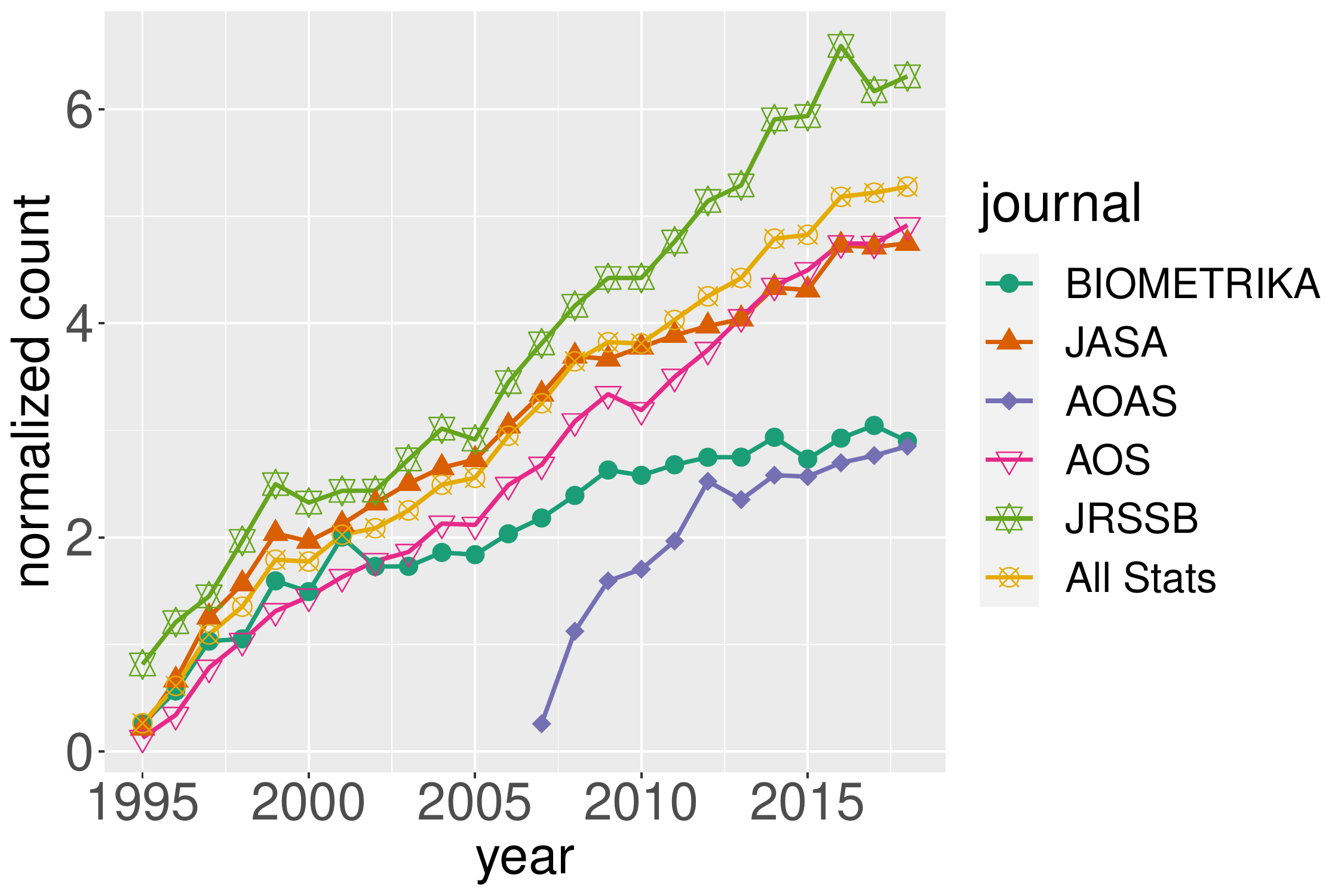}
 \vspace{-0.3cm}
\caption{The normalized citation counts for each journal over the years. Four papers with citations greater than 5000 are removed from JRSSB. }\label{fig:normal_2}
\end{figure}

 \vspace{-0.5cm}

\begin{figure}[h]
\centering
     \begin{subfigure}[b]{0.47\textwidth}
    \centering
    \includegraphics[width=\textwidth]{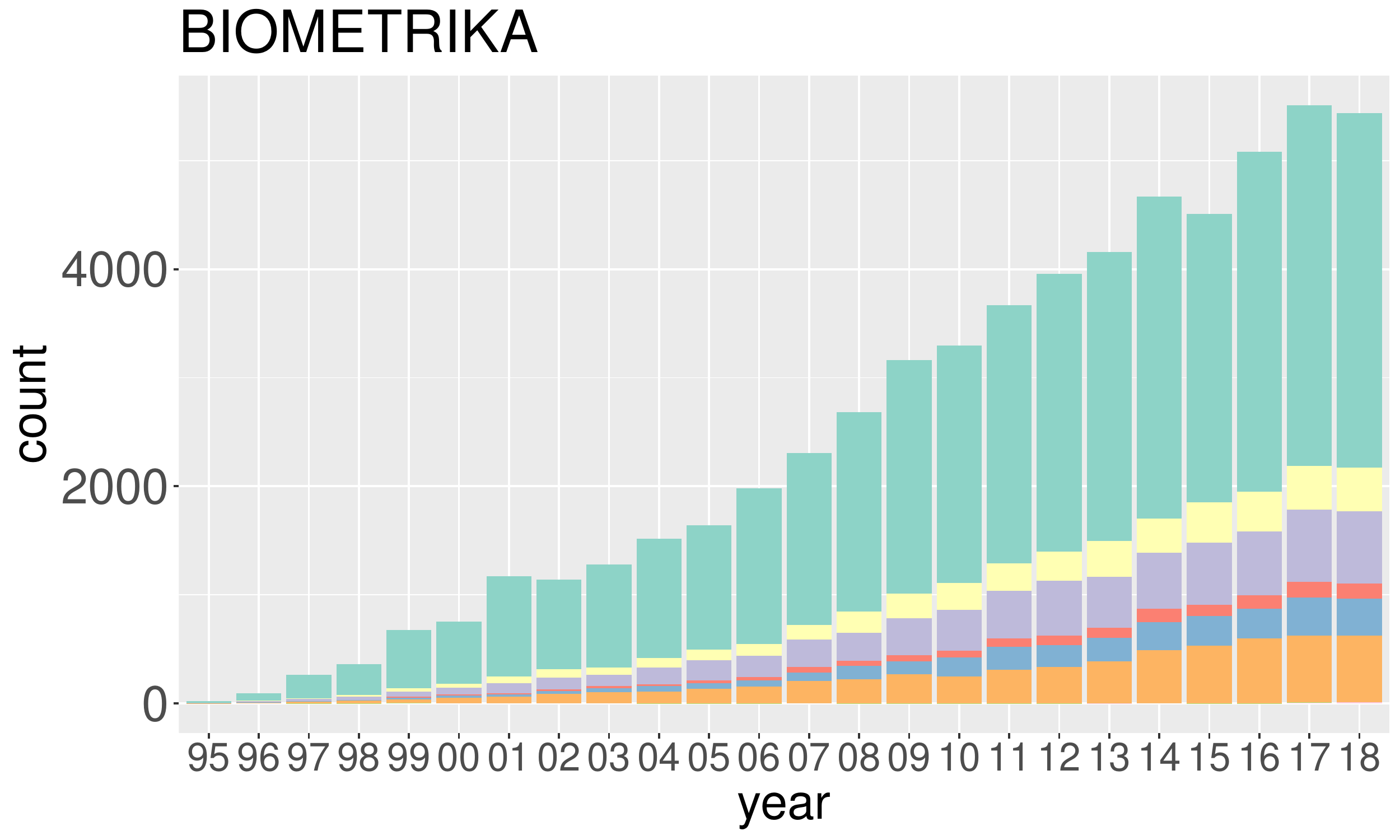}
     \vspace{-0.5cm}
    \caption{}
     \end{subfigure}
     \begin{subfigure}[b]{0.47\textwidth}
    \centering
    \includegraphics[width=\textwidth]{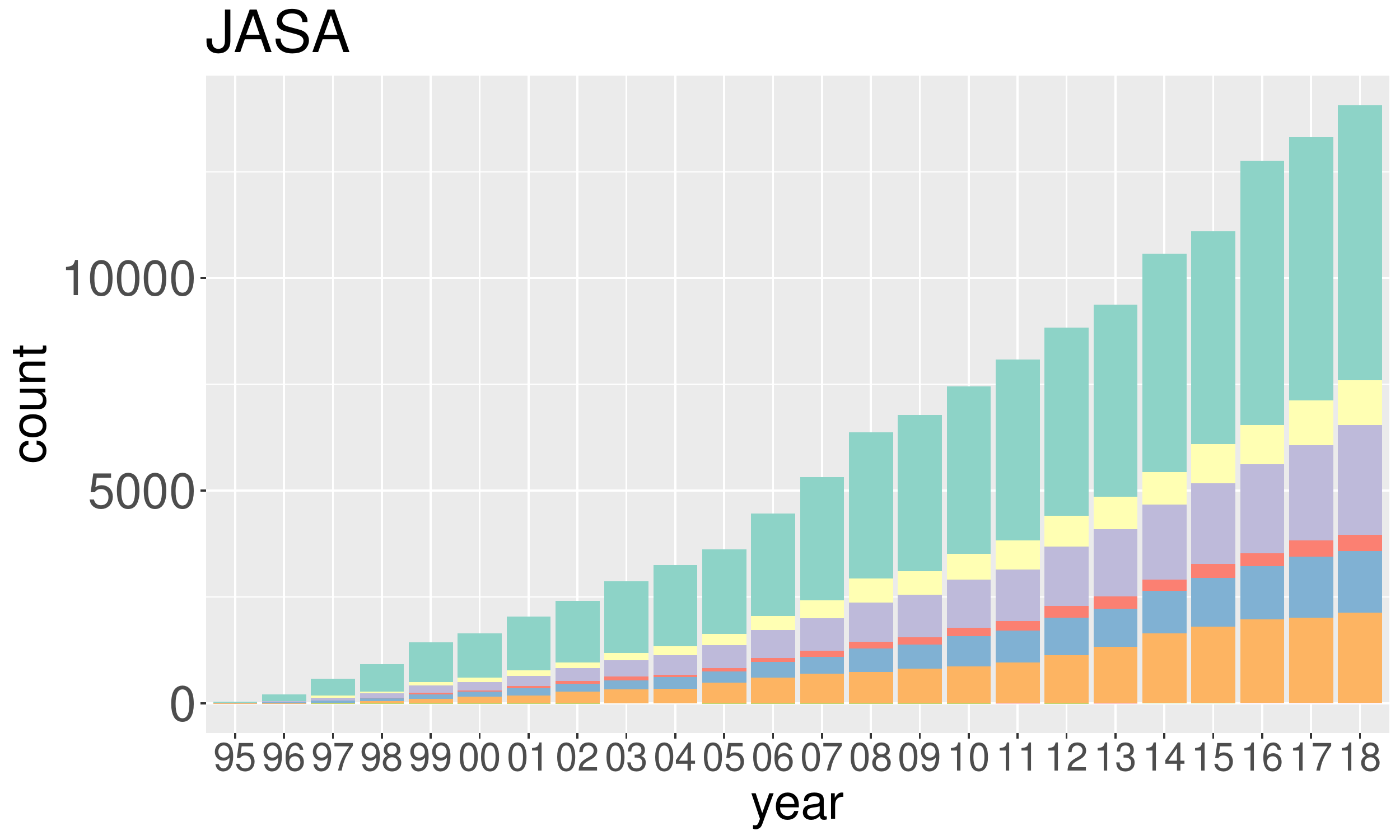}
     \vspace{-0.5cm}
    \caption{}
     \end{subfigure}
      \begin{subfigure}[b]{0.47\textwidth}
    \centering
    \includegraphics[width=\textwidth]{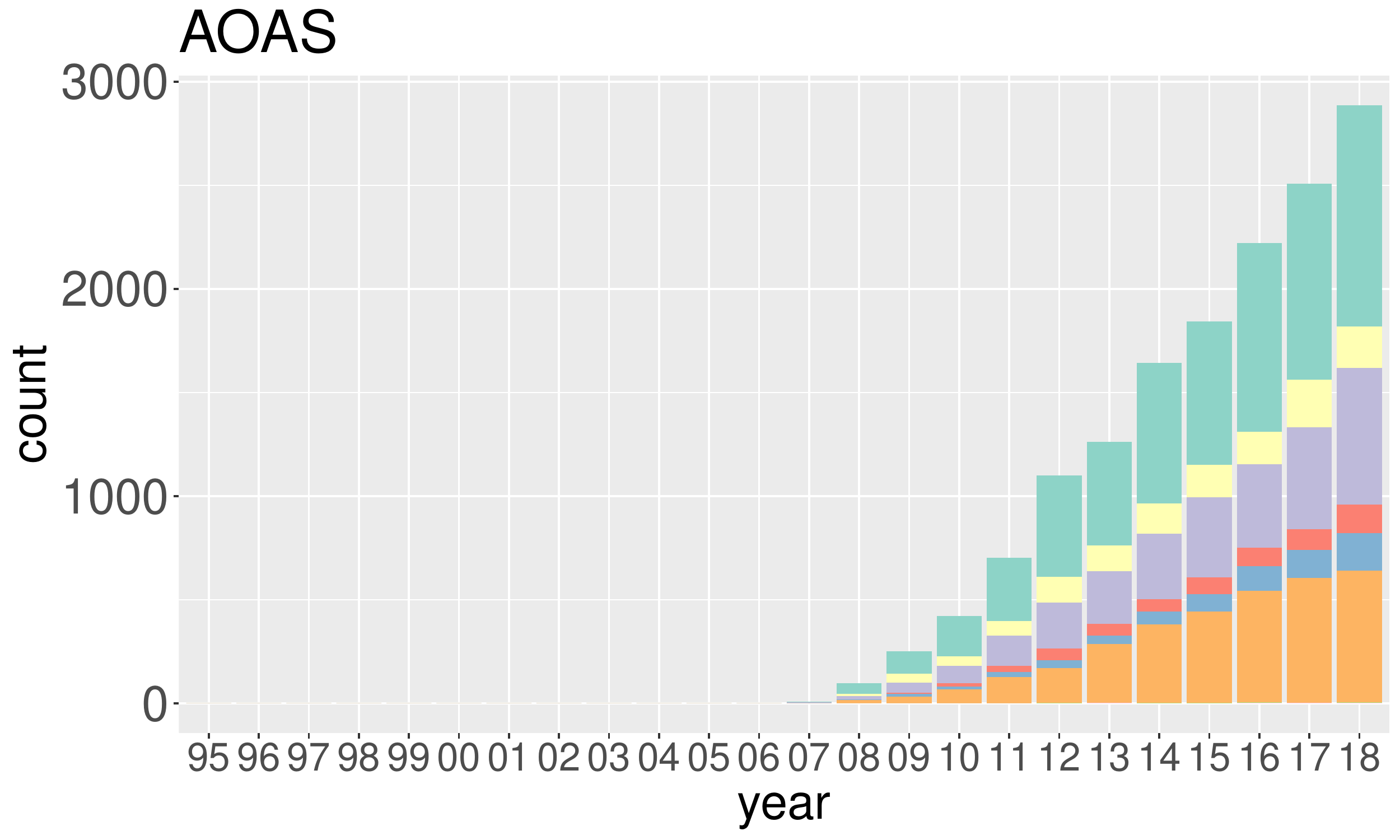}
     \vspace{-0.5cm}
    \caption{}
     \end{subfigure}
      \begin{subfigure}[b]{0.47\textwidth}
    \centering
    \includegraphics[width=\textwidth]{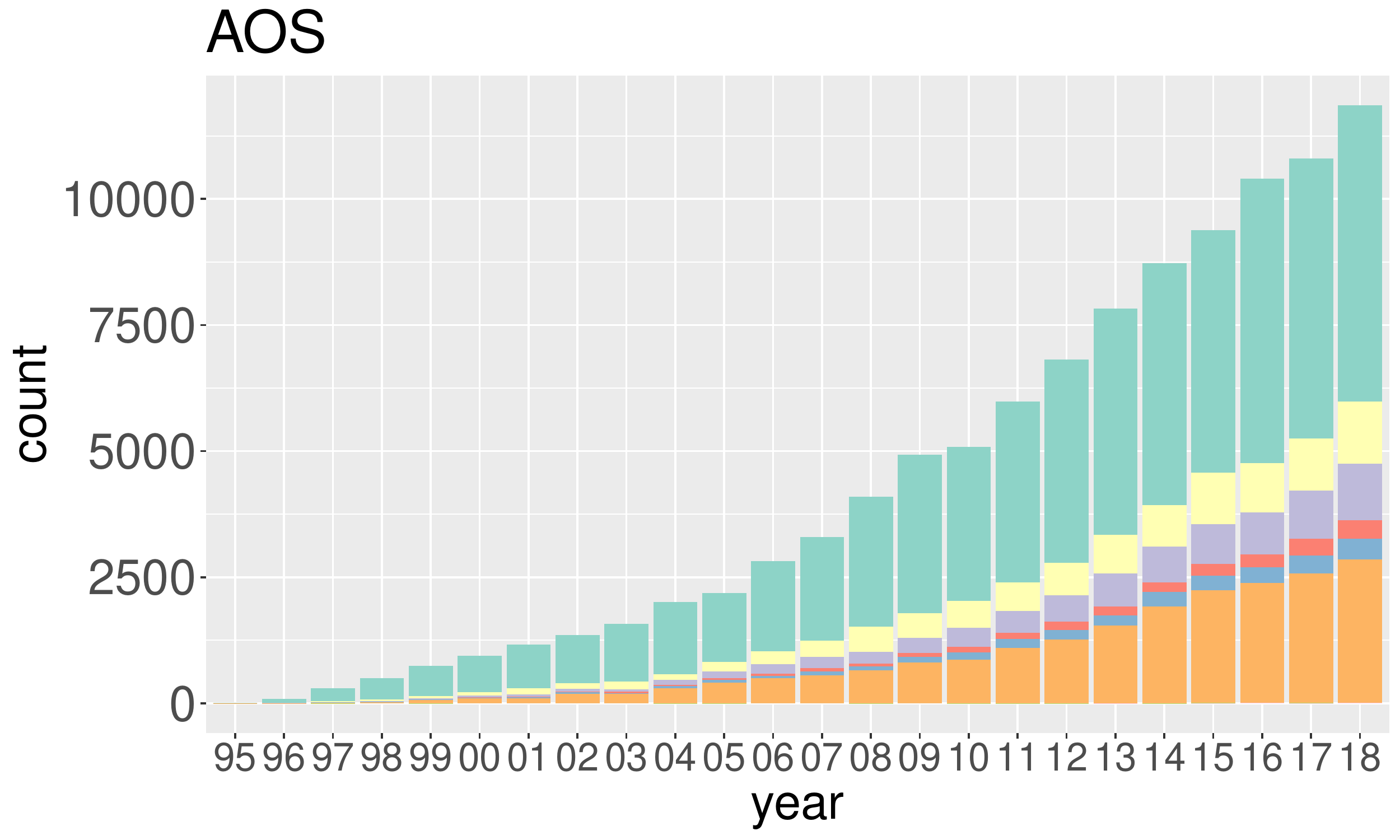}
     \vspace{-0.5cm}
    \caption{}
     \end{subfigure}
      \begin{subfigure}[b]{0.55\textwidth}
    \centering
    \includegraphics[width=\textwidth]{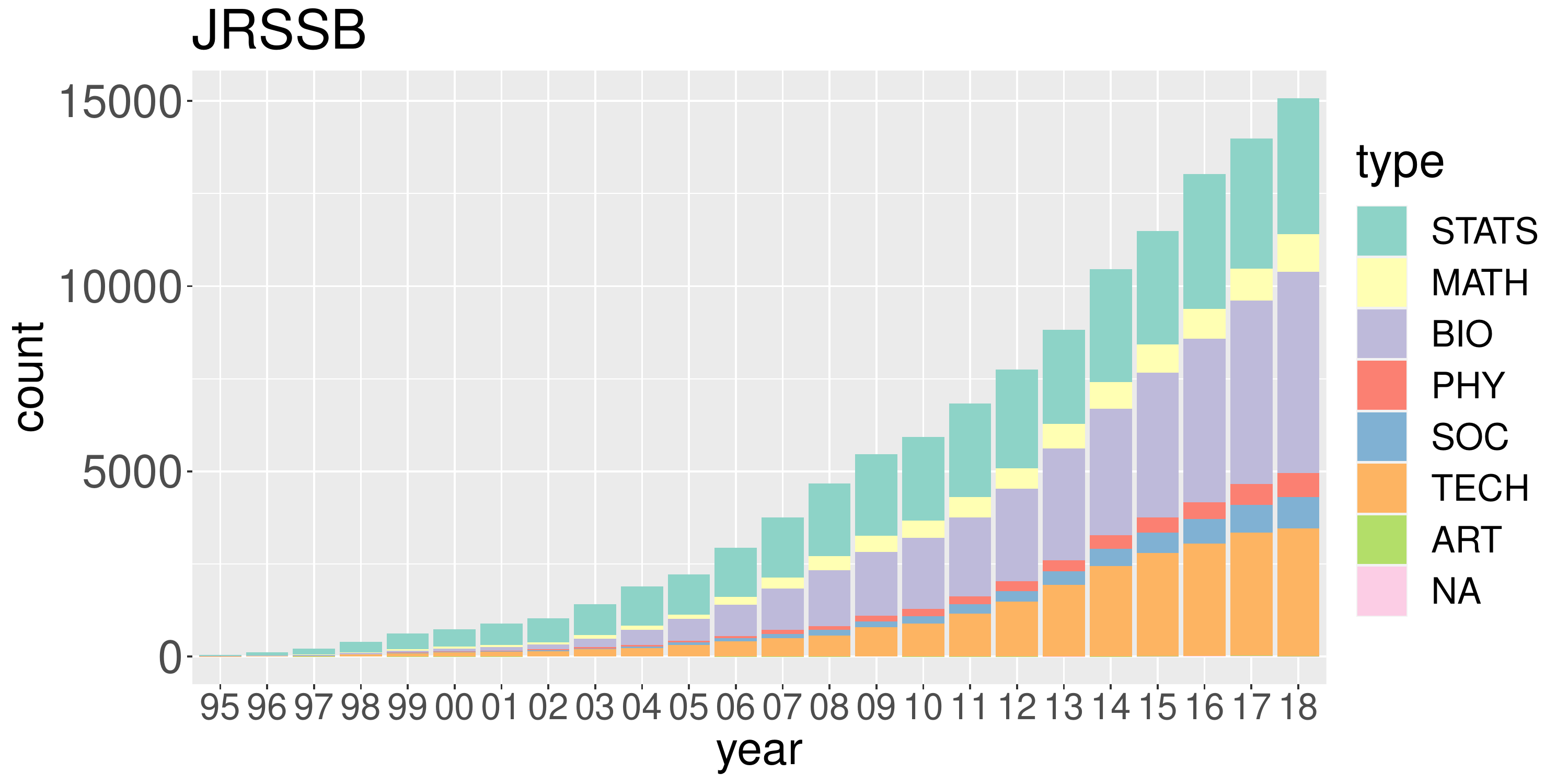}
     \vspace{-0.5cm}
    \caption{}
     \end{subfigure}
     \vspace{-0.3cm}
\caption{Breakdown of citations into research categories for each statistical journal: (a) BIOMETRIKA, (b) JASA, (c) AOAS, (d) AOS, and (e) JRSSB.}\label{fig:diversity_journal}
\end{figure}

\clearpage

\begin{table}[h]
\centering
\footnotesize
\begin{tabular}{p{8cm}|c}
\hline\hline
WoS category & Frequency  \\ \hline
Statistics \& Probability  & 534 \\
Mathematics, Interdisciplinary Applications  & 69 \\ 
Computer Science, Interdisciplinary Applications  & 61 \\ 
Mathematical \& Computational Biology   & 48 \\
Economics &  40 \\ 
Social Sciences, Mathematical Methods  & 35 \\ 
Operations Research \& Management Science  & 30 \\ 
Computer Science, Artificial Intelligence  & 28 \\ 
Engineering, Electrical \& Electronic  & 24 \\ 
 Mathematics, Applied  & 21 \\\hline
\end{tabular}
\caption{Top 10 most frequent WoS categories in all the citations towards \citet{azzalini1996multivariate}.}\label{table:hl_paper_1}
\end{table}

\begin{table}[h]
\centering
\footnotesize
\begin{tabular}{p{8cm}|c}
\hline\hline
WoS category & Frequency  \\ \hline
Computer Science, Artificial Intelligence  & 25 \\
 Engineering, Electrical \& Electronic  & 23 \\
 Operations Research \& Management Science  & 21 \\
  Economics &  17 \\
  Business, Finance  & 12 \\
Computer Science, Interdisciplinary Applications  &   9 \\
 Engineering, Industrial   &  9 \\
   Multidisciplinary Sciences  &  9 \\
 Computer Science, Information Systems   &  8 \\
 Engineering, Civil  &  8\\ \hline
\end{tabular}
\caption{Top 10 most frequent WoS categories in the citations towards \citet{azzalini1996multivariate} after removing the internal papers.}\label{table:hl_paper_2}
\end{table}

\begin{table}[h]
\centering
\footnotesize
\begin{tabular}{p{8cm}|c}
\hline\hline
WoS category & Frequency  \\ \hline
Psychiatry   & 152 \\
Psychology, Multidisciplinary  & 132 \\
Psychology, Clinical &  106 \\
Public, Environmental \& Occupational Health  & 97 \\
Medicine, General \& Internal  & 88\\
Psychology  &  88 \\
Multidisciplinary Sciences  & 58 \\
Neurosciences  & 55 \\
Psychology, Applied &  49 \\
Psychology, Developmental  & 47\\ \hline
\end{tabular}
\caption{Top 10 most frequent WoS categories in all the citations towards \citet{duval2000nonparametric}.}\label{table:lh_paper_1}
\end{table}

\begin{table}[h]
\centering
\footnotesize
\begin{tabular}{p{8cm}|c}
\hline\hline
WoS category & Frequency  \\ \hline
Psychiatry   & 152 \\
Psychology, Multidisciplinary  & 132 \\
Psychology, Clinical &  106 \\
Medicine, General \& Internal  & 88\\
Psychology  &  88 \\
Public, Environmental \& Occupational Health  & 86 \\
Multidisciplinary Sciences  & 58 \\
Neurosciences  & 55 \\
Psychology, Applied &  49 \\
Psychology, Developmental  & 47\\ \hline
\end{tabular}
\caption{Top 10 most frequent WoS categories in the citations towards \citet{duval2000nonparametric} after removing the internal papers.}\label{table:lh_paper_2}
\end{table}

\begin{table}[h]
\centering
\footnotesize
\begin{tabular}{p{8cm}|c}
\hline\hline
WoS category & Frequency  \\ \hline
Psychiatry  & 360 \\
Psychology, Developmental & 342 \\
Psychology, Clinical & 293 \\
Substance Abuse &  171 \\
 Public, Environmental \& Occupational Health & 167 \\
Psychology, Educational &  147 \\
Psychology, Multidisciplinary & 134 \\
Family Studies   & 112\\
Education \& Educational Research & 108 \\      Psychology & 106 \\\hline
\end{tabular}
\caption{Top 10 most frequent WoS categories in all the citations towards \citet{lo2001testing}.}\label{table:lh_paper_3}
\end{table}

\clearpage

\section{Proofs of the main results}
\subsection{Proof of Proposition~\ref{prop:aPPR_clustering}}\label{app.sec:prop_1}

Under the DC-SBM, \citet{CZR20} constructed the ``block-wise” population version of aPPR vector $\bp^*\in \bbR^K$ and proved that when there is a single seed node in block 1,
\begin{equation}\label{eq:CZR_order}
\bp^*_1 > \max\{\bp^*_k \mid k = 2, \ldots, K\}\,. \end{equation}
The separation between block 1 and the other blocks is defined as 
$\Delta_\alpha \in [0, 1]$,
\begin{equation}\label{eq:CZR_delta}
\Delta_\alpha = \frac{\bp^*_1 - \max\{\bp^*_k \mid k = 2, \ldots, K\}}{\bp^*_1}.\,
\end{equation}
Note that $\Delta_\alpha$ is an increasing function function of $\alpha$. This separation together with appropriate concentration analysis allowed them to show in their Corollary 1 that the sample aPPR vector can consistently recover all the nodes in block 1 given the correct size cutoff. 

The following property of $\bp^*$ can be easily derived from the linearity of PPR vectors in general,
\begin{equation}\label{eq:CZR_linearity_1}
\bp^*(\omega_1 \pi_1 + \omega_2 \pi_2) = \omega_1 \bp^*(\pi_1)
+ \omega_2 \bp^*(\pi_2), \mbox{ where } \omega_i \geq 0 \mbox{ and } \omega_1 +\omega_2 = 1 \,.
\end{equation}
This property enables us to extend Corollary 1 in \citet{CZR20} to the setting with multiple seed nodes in a straightforward way.

\begin{proof}[\textbf{Proof of Proposition}~\ref{prop:aPPR_clustering}]

We first check that their assumption $\frac{\max_{i \in \cI} \cd_i}{\min_{i \in \cI} \cd_i} < c_0$ for some constant $c_0$ holds under our assumption~\ref{prop.dcsbm.assp.3}.
We have 
\begin{equation*}
 \cd_i = \sum_{j} \sA_{ij} 
    =  \sum_{j}\theta_i \theta_j B_{g(i)g(j)} = \theta_i \rho_N \sum_{j} \theta_j S_{g(i)g(j)} \,.
\end{equation*}
By assumption~\ref{prop.dcsbm.assp.3},
\begin{equation}
   \frac{\max_{i \in \cI} \cd_i}{\min_{i \in \cI} \cd_i}  \leq 
    \frac{\max_{i \in \cI} \theta_i  \sum_{j} \theta_j S_{g(i)g(j)}}{\min_{i \in \cI} \theta_i  \sum_{j} \theta_j S_{g(i)g(j)}} 
  \leq \frac{U_\theta  \max S_{ij} \sum_{j} \theta_j }{L_\theta \min S_{ij} \sum_{j} \theta_j} = \frac{U_\theta  \max S_{ij}}{L_\theta \min S_{ij}} \,.
\end{equation}
Since $S$ is a fixed matrix, $\frac{\max_{i \in \cI} \cd_i}{\min_{i \in \cI} \cd_i}$ is bounded above.

It remains to show the inequality~\eqref{eq:CZR_order} holds for multiple seed nodes from the same block. Without loss of generality, we consider two seed nodes $v_1 = 1$ and $v_2 = 2$ from block 1, and their corresponding preference vectors are $\pi_1 = (1,0,0,\ldots,0)^\top$ and $\pi_2 = (0,1,0,\ldots,0)^\top$. When $\omega_1 + \omega_2 = 1$ and $\omega_i \geq 0$, $\omega_1 \pi_1 + \omega_2\pi_2$ can be considered as a preference vector containing two seed nodes from the same block. Now Eq~\eqref{eq:CZR_order} applies to $\pi_1$ and $\pi_2$ separately, that is
\begin{multline}\label{eq:CZR_order_2}
 \bp_1^*(\pi_1) >  \max\{\bp^*_k (\pi_1) \mid k = 2, \ldots, K\} \mbox{ and }\\  \bp_1^*(\pi_2) > \max\{\bp^*_k (\pi_2) \mid k = 2, \ldots, K\} \,.
\end{multline}
By Eq~\eqref{eq:CZR_linearity_1} and Eq~\eqref{eq:CZR_order_2}, we have 
\begin{align}
    \bp_1^*(\omega_1 \pi_1 + \omega_2 \pi_2) & = \omega_1 \bp_1^*(\pi_1)   
+ \omega_2 \bp_1^*(\pi_2)   \notag\\
& > \omega_1 \max\{\bp^*_k (\pi_1) \mid k = 2, \ldots, K\}  + \omega_2 \max\{\bp^*_k (\pi_2) \mid k = 2, \ldots, K\} \notag\\
& \geq  \max\{\omega_1\bp^*_k (\pi_1) + \omega_2 \bp^*_k (\pi_2) \mid k = 2, \ldots, K\} \notag\\
& = \max\{\bp^*_k (\omega_1\pi_1 + \omega_2\pi_2) \mid k = 2, \ldots, K\} \,.
\end{align}
The rest of the proof is the same as that of Corollary 1 in \citet{CZR20}.
\end{proof}

\subsection{Properties of the population version of conductance} \label{app.sec:prop_SBM}

In this section, we analyze the optimality properties of the conductance function under the population version before we present the sample version in the next section. Such a technique has been widely used in a number of works (e.g., \citet{BC09} and \citet{zhao2012consistency}); our case mostly differs in the construction of the confusion matrix and analysis of the population version of the objective function.  

A major difference between the previous works and our analysis is that 
they aim to recover all the blocks, whereas we are only concerned about the target block (block 1). For a given cutoff set $\cC_n$ in Eq~\eqref{eq:appr_select} ($n \in [N]$), which essentially partitions all the nodes $\cI$ into two sets, we consider the label assignment function $z=h(\cC_n)$. More concretely, 
\begin{equation}\label{eq:label_equation}
z(u) =\begin{cases}
1\,,& u \in \cC_n\,;\\
2\,,& u \notin \cC_n\,, 
\end{cases}
\end{equation}
for each node $u\in\cI$. In other words, we merge blocks $2,\ldots, K$ into one block and collectively call them block 2. Therefore, the correct assignment $z_0$ (as far as block 1 is concerned) should have $z_0(u) =1$ for $u \in \cI_1$, and $z_0(u) = 2$ for $u \notin \cI_1$. Given the aPPR vector $p^*$ and the corresponding sequence $\{\cC_n\}_{n\in[N]}$, denote $$\zeta = \left\{ z=h(\cC_n) | n \in [N]\right\},$$
which is the set of all possible labels generated from $\{\cC_n\}_{n\in[N]}$. We have $|\zeta| = N$.

Recall that Proposition~\ref{prop:aPPR_clustering} establishes the aPPR vector recovers block 1 with high probability when $n=n_1$, i.e., $\cC_{n_1}=\cI_1$. We will show that under this high probability event, $\phi(\cC_{n_1})$ is a local minimum by analyzing the neighborhood around $n_1$. It is easy to see this event also implies the following property for the set $\cC_n$,
\begin{equation}\label{eq:prop_1_assump}
    \cC_n \left\{\begin{array}{c}
\subsetneq \\
=\\
\supsetneq
\end{array}\right\} \cI_1 \mbox{ when } n 
\left\{\begin{array}{c}
< \\
=\\
>
\end{array}\right\} n_1\,.
\end{equation}
 In other words, all the nodes in $\cC_n$ are from block 1 when the cutoff $n < n_1$;  $\cC_n$ is exactly block 1 when $n = n_1$; and the whole block 1 is contained in $\cC_n$ when $n > n_1$.

For clarity of description, we first study the properties of $\phi$ under the SBM with $K$ blocks. Following the convention, let $\|z - z_0\|_1 = \sum^N_{u = 1}\bone\{z(u) \neq z_0(u)\}$, and for $1\leq a,b \leq 2$,
\begin{equation}
     O_{ab}(z) = \sum_{u \neq v}  \bone\{ z(u) = a, z(v) = b\} A_{uv}\,.
\end{equation}
Define the confusion matrix $R \in [0,1]^{2 \times K} $,
\begin{equation}
\label{eq:confusion_sbm}
    R_{ab}(z, g) = \frac{1}{N}\sum^N_{u = 1}\bone\{z(u) = a, g(u) = b\},\,
\end{equation}
where $g$ denotes the correct labels as introduced in \ref{sec:intro_SBM}. Let $R$ abbreviate $R(z, g)$, and $RSR^\top$ abbreviate $R(z, g) S R(z, g)^\top$. Note that $R^\top \bone = \tau$ where $\tau_k = |\cI_k|/N$ is the proportion of nodes in block $k$. Let $\mu_N = N^2 \rho_N$, then
$$\frac{1}{\mu_N}\mathbb{E}[O(z) \mid  g ] = RSR^\top . \,$$
For convenience, for a general $2 \times 2$ matrix $M$, define
\begin{equation}\label{eq:F_def}
    F(M) = \frac{M_{11}}{ M_{1\cdot}}\,, \mbox{ where }  M_{1\cdot} = M_{11} + M_{12}\,.
\end{equation}
We immediately have 
\begin{equation}\label{eq:conduct_to_F}
    1 - \phi(\cC_n) = F\left(\frac{O(z)}{\mu_N}\right)\,.
\end{equation}
Moreover, we write $$G(z) = F(RSR^\top(z))\,,$$ which is the population version of Eq~\eqref{eq:conduct_to_F} and only depends on $z$. The following lemma shows $z_0$ is a well separated local optimum in a suitable neighborhood defined around $\cC_{n_1}$. Recall that we are working under the event $\cC_{n_1}=\cI_1$ so that Eq~\eqref{eq:prop_1_assump} holds.

\begin{lemma}\label{prop.sbm.differ}
Suppose the assumption~\ref{prop.sbm.assp.2} holds under the SBM. Given a large enough $N$ and a sequence $\{\cC_n\}_{n\in [N]}$, there exists $n'$ satisfying $n' - n_1 = \Omega(N)$ such that 
    $$ G(z) - G(z_0) \leq - \frac{1}{N} \Omega(|n - n_1|)$$
    uniformly for all $z$ in the set $\zeta' = \{z=h(\cC_n) \mid n \in [n']\}$.
    
\end{lemma}

\begin{proof}[\textbf{Proof of Lemma}~\ref{prop.sbm.differ}]
We have
\begin{equation}\label{eq:prop_1_F}
   G(z) = F(RSR^\top(z)) = \frac{[RSR^\top]_{11}}{[RSR^\top]_{1\cdot} } = \frac{\sum_{i,j} R_{1i} S_{ij} R_{1j}}{\sum_{i,j,k} R_{1i} S_{ij} R_{kj}}\,.
\end{equation}

According to~\eqref{eq:prop_1_assump} , we consider the following cases. 

\underline{Case 1}: $n = n_1$. Then, $\cC_n = \cI_1$ (i.e., $\{u \mid z(u) = 1\}  = \{u \mid g(u) = 1\}$), $z=z_0$. Note that $R_{1j} = 0$ for $j \neq 1$, $R_{21} = 0$, and  $R_{11} = \tau_1$.
By Eq \eqref{eq:prop_1_F}, 
\begin{equation*}
    G(z_0) = \frac{S_{11} \tau_1^2}{\tau_1 \sum_{j,k} S_{1j} R_{kj} }= \frac{S_{11} \tau_1^2}{\tau_1 \sum_{j} S_{1j} \tau_j} = \tau_1\frac{S_{11}}{\tS_{1\cdot}}\,.
\end{equation*}

\underline{Case 2}: $n < n_1$.  By Eq~\eqref{eq:prop_1_assump}, $\cC_n \subsetneq \cI_1$, that is, $\{u \mid z(u) = 1\} \subsetneq  \{u \mid g(u) = 1\}$. It follows then  $R_{1j} = 0$ for $j \neq  1$, and
\begin{equation}
    R_{11} = \frac{1}{N}\sum^N_{u = 1}\bone\{z(u) = 1, g(u) = 1\} = \frac{\sum^N_{u = 1}\bone\{z(u) = 1\}}{N} = \frac{|\cC_n|}{N} = \frac{n}{N}\,.
\end{equation}
We have
\begin{equation*}
\tau_1 - R_{11} = \frac{1}{N}   |n - n_1|\,,\mbox{ since } \tau_1 = \frac{n_1}{N}\,,
\end{equation*}
and
\begin{equation*}
    G(z) = \frac{S_{11} R_{11}^2}{R_{11} \sum_{j} S_{1j} \tau_j} = R_{11}\frac{S_{11}}{\tS_{1\cdot}} \,.
\end{equation*}
Therefore,
\begin{equation*}
G(z) - G(z_0) = (R_{11} - \tau_1) \frac{S_{11}}{\tS_{1\cdot}}  \leq  - \frac{1}{N}  \Omega( |n - n_1|)\,.
\end{equation*}

\underline{Case 3}: $ n > n_1 $. By Eq~\eqref{eq:prop_1_assump}, $\cC_n \supsetneq \cI_1$, that is, $\{u \mid z(u) = 1\} \supsetneq  \{u \mid g(u) = 1\}$. Hence $R_{21} = 0$, $R_{11} = \tau_1$,  and 
\begin{align}\label{eq:R1b}
    \sum_{j > 1}R_{1j} &=  \frac{1}{N}\sum^N_{u = 1} \sum_{b > 1}\bone\{z(u) = 1, g(u) = j\} = \frac{1}{N}\sum^N_{u = 1} \bone\{z(u) = 1, g(u) \neq 1\} \notag\\ 
& =\frac{1}{N}( \sum^N_{u = 1}\bone\{z(u) = 1\} - \sum^N_{u = 1}\bone\{g(u) = 1\}) = \frac{1}{N}( |\cC_n| - |\cI_1|)
= \frac{1}{N} |n - n_1|\,.  
\end{align}

We have 
\begin{align}\label{eq:prop_1_n_geq}
     G(z) &= \frac{\tau_1^2 S_{11} + 2 \tau_1 \sum_{j > 1} S_{1j} R_{1j} + \sum_{i,j > 1}R_{1i} S_{ij} R_{1j} }{\sum_{j} R_{1i}S_{ij}\tau_j} \notag\\
    & = \frac{\tau_1^2 S_{11} + 2 \tau_1 \sum_{j > 1} S_{1j} R_{1j} + \sum_{i,j > 1}R_{1i} S_{ij} R_{1j} }{\tau_1 \tS_{1\cdot} + \sum_{i > 1} \tS_{i\cdot} R_{1i}}\,. 
\end{align}

Substituting,
\begin{equation*}
   G(z) -  G(z_0) = \frac{\tau_1^2 S_{11} + 2 \tau_1 \sum_{j > 1} S_{1j} R_{1j} + \sum_{i,j > 1}R_{1i} S_{ij} R_{1j} }{\tau_1 \tS_{1\cdot} + \sum_{i > 1} \tS_{i\cdot} R_{1i}} - \frac{\tau_1 S_{11}}{\tS_{1\cdot}}\,.
\end{equation*}
For a small but fixed $\varepsilon > 0$ (to be specified later), by choosing $n'$ satisfying $n' - n_1 = \lfloor N \varepsilon \rfloor$, we can restrict  $\max_{j > 1} R_{1j} \leq \varepsilon$ for all $N>1/\varepsilon$ and $n \in [n']$ according to Eq~\eqref{eq:R1b}.
Then we have 
\begin{align}\label{eq:prop_1_n_geq_differ}
   &G(z) -  G(z_0)  \notag\\ &\leq  \frac{\tau_1^2 S_{11} + 2\tau_1 \max_{j > 1} S_{1j}\sum_{j > 1}R_{1j} + \varepsilon \max_{i, j > 1} S_{ij}  \sum_{j > 1}R_{1j} }{\tau_1 \tS_{1\cdot} + \min_{i > 1}\tS_{i\cdot} \sum_{j > 1}R_{1j} } - \frac{\tau_1 S_{11}}{\tS_{1\cdot}}\notag\\
&= \frac{\left[\tS_{1\cdot}(2\tau_1  \max_{j > 1} S_{1j} + \varepsilon \max_{i,j > 1} S_{ij}) - \tau_1 S_{11}\min_{i > 1}\tS_{i\cdot} \right] \sum_{j > 1}R_{1j} }{\tS_{1\cdot} (\tau_1 \tS_{1\cdot} + \min_{i > 1}\tS_{i\cdot} \sum_{j > 1}R_{1j})}\notag\\
&= \frac{\left[\tS_{1\cdot}(2  \max_{j > 1} S_{1j} + \varepsilon \tau_1^{-1}\max_{i,j > 1} S_{ij}) -  S_{11}\min_{i > 1}\tS_{i\cdot} \right]  }{ \tS_{1\cdot}(\tS_{1\cdot} + \tau_1^{-1}\min_{i > 1}\tS_{i\cdot} \sum_{j > 1}R_{1j})  } (\frac{1}{N}   |n - n_1|) 
\end{align}
According to assumption~\ref{prop.sbm.assp.2}, there exists a constant $c>0$ such that
\begin{equation}\label{eq:prop_1_assp1_eq}
S_{11}\min_{i > 1}\tS_{i\cdot} - c >  2  \max_{j > 1} S_{1j} \tS_{1\cdot}\,.    
\end{equation}
Let 
\begin{equation*}
    \varepsilon = \frac{c\tau_1}{2\tS_{1\cdot} \max_{i,j > 1} S_{ij}}\,,
\end{equation*}
it follows then 
\begin{equation*}
    \tS_{1\cdot}(2  \max_{j > 1} S_{1j} + \varepsilon \tau_1^{-1} \max_{i,j > 1} S_{ij}) - S_{11}\min_{i > 1}\tS_{i\cdot} < -\frac{c}{2} <0\,,
\end{equation*}
and~\eqref{eq:prop_1_n_geq_differ} is upper bounded by
\begin{equation}\label{eq:lemma_1_upper_bound}
 -\frac{c/2}{ \tS_{1\cdot}(\tS_{1\cdot} + \tau_1^{-1}\min_{i > 1}\tS_{i\cdot} K\varepsilon)  } (\frac{1}{N}   |n - n_1|) =     -\frac{1}{ \frac{2 \tS_{1\cdot}^2}{c} + \frac{K \min_{i > 1}\tS_{i\cdot} } {\max_{i,j > 1} S_{ij}} } (\frac{1}{N}   |n - n_1|) 
 \end{equation}

Therefore, 
$$ G(z) -  G(z_0) < - \frac{1}{N}  \Omega( |n - n_1|) \mbox{ for all } n \in\{n_1 + 1, \ldots, n'\} \,.$$
\end{proof}

The same result can be shown for a DC-SBM with $K$ blocks by defining a confusion tensor, adding a dimensionality for the nodes; all the other notations remain the same unless otherwise specified. Define the confusion tensor $T \in \{0,\frac{1}{N}\}^{2\times K \times N}$ as
\begin{equation}
    T_{abu} (z, g) = \frac{1}{N} \bone\{z(u) = a, g(u) = b\}\,.
\end{equation}
 Then, define the degree-corrected confusion matrix $\tR \in [0,1]^{2 \times K}$,
\begin{equation}
\tR_{a,b}(z, g) =    \sum^{N}_{u = 1} \theta_u T_{abu} (z,g) \,.
\end{equation}
Let $T$ abbreviate $T(z,g)$, and $\tR$ abbreviate $\tR(z,g)$. 
Now we have
\begin{equation}
    [\tR^\top\bone]_b = \sum_{a,u} \theta_u T_{abu} = \frac{n_b}{N} = \tau_b\,.
\end{equation}
Also, 
\begin{equation}
    \frac{1}{\mu_N} \mathbb{E} [O_{ab}\mid g, \theta] = \sum^{K}_{i,j = 1} \sum^{N}_{u,v = 1} T_{a,i,u} \theta_{u} S_{ij}  \theta_{v} T_{b,j,v} = [\tR S \tR^\top]_{ab}\,.
\end{equation}

The following lemma is similar to Lemma~\ref{prop.sbm.differ} but extends the result to the DC-SBM.

\begin{lemma}\label{prop.dcsbm.differ}
Suppose the assumptions~\ref{prop.dcsbm.assp.3} and~\ref{prop.sbm.assp.2} hold under the DC-SBM. Define $\tG(z) = F(\tR S \tR^\top(z))$. Given a large enough $N$ and a sequence $\{\cC_n\}_{n\in [N]}$,
there exists $n'$ satisfying $n' - n_1 = \Omega(N)$ such that 
    $$\tG(z) -\tG(z_0) \leq - \frac{1}{N} \Omega(|n - n_1|)$$
     uniformly for all $z$ in the set $\zeta' = \{z=h(\cC_n) \mid n \in [n']\}$.
\end{lemma}

\begin{proof}[\textbf{Proof of Lemma}~\ref{prop.dcsbm.differ}]

We have 
\begin{equation}\label{eq:prop_2_F}
   \tG(z) = F(\tR S \tR^\top(z)) = \frac{[\tR S \tR^\top]_{11}}{[\tR S \tR^\top]_{1\cdot}} = \frac{\sum_{i,j} \tR_{1i} S_{ij} \tR_{1j}}{\sum_{i,j,k} \tR_{1i} S_{ij} \tR_{kj}}\,.
\end{equation}
Similar to the proof for Lemma~\ref{prop.sbm.differ}, we consider the three cases in Eq~\eqref{eq:prop_1_assump}.

\underline{Case 1}: $n = n_1$. Again, we have $\cC_n = \cI_1$ and $z=z_0$, so $\tR_{1j} = 0$ for $j\neq 1$, $\tR_{21} = 0$, and $\tR_{11} = \tau_1$.
By Eq~\eqref{eq:prop_2_F}, we have
\begin{equation}\label{eq:prop_2_n_eq}
    \tG(z_0)  = \tau_1\frac{S_{11}}{\tS_{1\cdot}}\,.
\end{equation}

\underline{Case 2}: $n < n_1$.  By Eq~\eqref{eq:prop_1_assump}, $\cC_n \subsetneq \cI_1$ (i.e., $\{u \mid z(u) = 1\} \subsetneq  \{u \mid g(u) = 1\}$), so  $\tR_{1j} = 0$ for $j \neq  1$, and
\begin{multline}\label{eq:prop_2_distant_1}
   \tau_1 -   \tR_{11}   =  \frac{n_1}{N} - \sum^{N}_{u = 1} \frac{\theta_u}{N} \bone\{z(u) = 1, g(u) = 1\} \\
   = \sum_{u \in \cI_1} \frac{\theta_u}{N} - \sum_{u \in \cC_n} \frac{\theta_u}{N} = \frac{
   1}{N} \sum_{u \in \cI_1 \setminus \cC_n} \theta_{u} \\ 
   \geq \frac{L_\theta }{N} \sum_{u \in \cI_1 \setminus \cC_n} 1 = \frac{L_\theta }{N}|n - n_1| \,.
\end{multline}
The last inequality holds by assumption~\ref{prop.dcsbm.assp.3}. Also, we have 
\begin{equation*}
    \tG(z) =  \tR_{11}\frac{S_{11}}{\tS_{1\cdot}} \,.
\end{equation*}
Then,
\begin{equation*}
\tG(z) - \tG(z_0) = (\tR_{11} - \tau_1) \frac{S_{11}}{\tS_{1\cdot}}  \leq  - \frac{1}{N}  \Omega( |n - n_1|)\,.
\end{equation*}

\underline{Case 3}: $ n > n_1 $. By Eq~\eqref{eq:prop_1_assump}, $\cC_n \supsetneq \cI_1$ (i.e., $\{u \mid z(u) = 1\} \supsetneq  \{u \mid g(u) = 1\}$), so  $\tR_{21} = 0$, $\tR_{11} = \tau_1$, and 
\begin{equation*}
\sum_{j > 1}\tR_{1j} =  \frac{1}{N} \sum^N_{u = 1} \theta_u \bone\{z(u) = 1, g(u) \neq 1\} = \frac{1}{N} \sum_{u \in \cC_n\setminus \cI_1} \theta_u \geq \frac{L_\theta}{N}|n - n_1| \,.
\end{equation*}
Similar to Eq~\eqref{eq:prop_1_n_geq},
\begin{equation*}
   \tG(z) = \frac{\tau_1^2 S_{11} + 2 \tau_1 \sum_{j > 1} S_{1j} \tR_{1j} + \sum_{i,j > 1}\tR_{1i} S_{ij} \tR_{1j} }{\tau_1 \tS_{1\cdot} + \sum_{i > 1} \tS_{i\cdot} \tR_{1i}}\,.
\end{equation*}
Now, we have 
\begin{equation*}
   \tG(z) -  \tG(z_0) = \frac{\tau_1^2 S_{11} + 2 \tau_1 \sum_{j > 1} S_{1j} \tR_{1j} + \sum_{i,j > 1}\tR_{1i} S_{ij} \tR_{1j} }{\tau_1 \tS_{1\cdot} + \sum_{i > 1} \tS_{i\cdot} \tR_{1i}} - \frac{\tau_1 S_{11}}{\tS_{1\cdot}}\,.
\end{equation*}
By assumption~\ref{prop.dcsbm.assp.3},
\begin{equation*}
\tR_{ij} = \frac{1}{N}\sum_{u \in [N]} \theta_u \bone\{z(u) = i, g(u) = j\} \\ \leq \frac{U_\theta}{N}\sum_{u \in [N]} \bone\{z(u) = i, g(u) = j\}  =  U_\theta R_{ij}\,.    
\end{equation*} 
Here $R$ is the original confusion matrix defined in Eq~\eqref{eq:confusion_sbm}. By the same argument as in Lemma~\ref{prop.sbm.differ}, we can find a fixed $\varepsilon > 0$, such that by choosing $n'-n_1=\lfloor N\varepsilon\rfloor$, we can restrict  $\max_{j > 1} R_{1j} \leq \varepsilon$ for all $N > 1/\varepsilon$ and $n \in [n']$ according to Eq~\eqref{eq:R1b}.  Therefore, $$\max_{j > 1} \tR_{1j} \leq U_\theta \varepsilon\,.$$ 
Similar to Eq~\eqref{eq:prop_1_n_geq_differ},
\begin{equation}\label{eq:prop_2_n_geq_differ}
\tG(z) -  \tG(z_0) 
\leq \frac{[\tS_{1\cdot}(2  \max_{j > 1} S_{1j} + \varepsilon  \tau_1^{-1} U_\theta \max_{i,j > 1} S_{ij}) -  S_{11}\min_{i > 1}\tS_{i\cdot} ] \sum_{j > 1}\tR_{1j} }{\tS_{1\cdot}(\tS_{1\cdot} + \tau_1^{-1}\min_{i > 1}\tS_{i\cdot} \sum_{j > 1}\tR_{1j}) }.     
\end{equation}
Let 
\begin{equation*}
    \varepsilon = \frac{c \tau_1}{2 U_\theta \tS_{1\cdot} \max_{i,j > 1} S_{ij}}\,.
\end{equation*}
According to Eq~\eqref{eq:prop_1_assp1_eq}, we have 
\begin{equation*}
    \tS_{1\cdot}(2  \max_{j > 1} S_{1j} + \varepsilon  \tau_1^{-1} U_\theta \max_{i,j > 1} S_{ij}) -  S_{11}\min_{i > 1}\tS_{i\cdot} < -\frac{c}{2} <0\,.
\end{equation*}
Eq~\eqref{eq:prop_2_n_geq_differ} has an upper boundary similar to~\eqref{eq:lemma_1_upper_bound} that is
\begin{equation}
   -\frac{1}{ \frac{2 \tS_{1\cdot}^2}{c} + \frac{K \min_{i > 1}\tS_{i\cdot} } {\max_{i,j > 1} S_{ij}} } (\frac{L_\theta}{N}   |n - n_1|) \,.
   \end{equation}
Therefore, 
$$ \tG(z) -  \tG(z_0) < - \frac{1}{N}  \Omega( |n - n_1|) \mbox{ for } n \in\{n_1 + 1, \ldots, n'\} \,.$$
\end{proof}

\subsection{Proof of Theorem~\ref{thm:min}}
\label{app.sec:thm_min}

 The proof of the main theorem relies on the optimality  properties of the population version we derived in the previous section and concentration inequalities in following lemma.

\begin{lemma}\label{lemma:inequalities}
Let $X(z) = O(z)/\mu_N - \tR S \tR^\top (z)$,  $\| X\|_\infty$ denotes $\max_{ij} |X_{ij}|$. Define the constant $C_S= U_\theta^2\max_{ab} S_{ab}$. Then,
\begin{equation}\label{eq:x_con}
    \mathbb{P}(\max_{z \in \zeta}\| X(z)\|_{\infty} \geq \varepsilon)  \leq 8 N \exp \left(- \frac{ \varepsilon^2 \mu_N}{8  C_S}\right) \mbox{ for } \varepsilon \leq 6  C_S\,,
\end{equation}
\begin{equation}\label{eq:x_con_2}
    \mathbb{P}(\max_{z \in \zeta : \|z - z_0\|_1 = m}\| X(z)\|_{\infty} \geq \varepsilon)  \leq 16 \exp \left(- \frac{ \varepsilon^2 \mu_N}{8 C_S}\right) \mbox{ for } \varepsilon \leq 6  C_S\,,
\end{equation}
and
\begin{equation}\label{eq:x_differ_con}
    \mathbb{P}(\max_{z \in \zeta: \|z - z_0\|_1 = m}\| X(z) - X (z_0) \|_{\infty} \geq \varepsilon) \leq 16 \exp \left( - \frac{ N }{16 m C_S} \varepsilon^2 \mu_N  \right) \mbox{ for } \varepsilon \leq 12 m C_S/ N \,.
\end{equation}

\end{lemma}

\begin{proof}
These are well-known inequalities that can be proved by Bernstein's inequality. For the sake of completeness, we present the details here. We have
\begin{equation*}
    \frac{1}{\mu_N} \mathbb{E} [O_{ab} | g,\theta]  = [\tR S \tR^\top]_{ab}\,,
\end{equation*}
thus
\begin{equation*}
    \mu_N X_{ab} = O_{ab} - \mathbb{E} [O_{ab} | g, \theta]\,.
\end{equation*}
Also, we have
\begin{equation*}
    O_{ab} = \sum_{i \neq j}  \bone\{ z(i) = a, z(j) = b\} A_{ij} = 2\sum_{i < j}  \bone\{ z(i) = a, z(j) = b\} A_{ij}\,.
\end{equation*}
$\mu_N X_{ab}$ is a sum of independent zero mean random variables bounded by 1. By Bernstein's inequality, 
\begin{align*}
 \mathbb{P}(|\mu_N X_{ab}| \geq \varepsilon \mu_N ) & \leq 2 \exp \left( \frac{- \varepsilon^2 \mu_N^2}{2(\mbox{Var}(\mu_N X_{ab}) + \varepsilon \mu_N/3 )}\right) \\ & = 2 \exp \left( \frac{- \varepsilon^2 \mu_N^2}{2(\mbox{Var}(O_{ab}) + \varepsilon \mu_N/3 )}\right) \,.   
\end{align*}
Note that $A_{ij}  \overset{\mbox{ind.}}{\sim} \mbox{Bernoulli}\left( \theta_i \theta_j B_{g(i)g(j)} \right)$ and $B_{ij} = \rho_N S_{ij}$, it follows
\begin{align*}
    & \mbox{Var}(A_{ij}) =  \theta_i \theta_j \rho_N S_{g(i)g(j)} -  (\theta_i \theta_j \rho_N S_{g(i)g(j)})^2 \leq \rho_N U_\theta^2 \max S_{ij} =  \rho_N  C_S \,;\\
    & \mbox{Var}(O_{ab}) \leq  4 \frac{N(N  - 1)}{2} \rho_N   C_S \leq 2N^2 \rho_N   C_S = 2 \mu_N   C_S\,.
\end{align*}

Since $\varepsilon \leq 6  C_S$, for fixed $a, b, z$,
\begin{equation*}
  \mathbb{P}(|X_{ab}| \geq \varepsilon  )\leq 2 \exp \left(- \frac{ \varepsilon^2 \mu_N}{8 C_S}\right) \,.  
\end{equation*}
Therefore, 
\begin{equation}\label{eq:x_con_3}
    \mathbb{P}(\|X(z) \|_\infty \geq \varepsilon) \leq  8 \exp \left(- \frac{ \varepsilon^2 \mu_N}{8  C_S}\right) \mbox{ for a fixed $z$.}
\end{equation}
We have $ |\zeta| = N$, which establishes Eq~\eqref{eq:x_con}. Moreover, according to Eq~\eqref{eq:label_equation}, we have $| \{z \in \zeta : \|z - z_0\|_1 = m\}| \leq 2$, which establishes Eq~\eqref{eq:x_con_2}.

Now, we assume $z(m + 1) = z_0(m+1), \ldots, z(N) = z_0(N)$. Then, 
\begin{align*}
  O_{ab}(z)  -  O_{ab}(z_0)   &= 2\sum^{m}_{i < j}  (\bone\{ z(i) = a, z(j) = b\} - \bone\{ z_0(i) = a, z_0(j) = b\}) A_{ij} \\
& + 2\sum^{m}_{i = 1}\sum^{N}_{j = m + 1} (\bone\{ z(i) = a, z(j) = b\} - \bone\{ z_0(i) = a, z_0(j) = b\}) A_{ij} \,;\\
 \mbox{Var}( O_{ab}(z)  -  O_{ab}(z_0)) &\leq 4\left[\frac{m (m - 1)}{2} + m (N - m)\right] \rho_N C_S \leq 4mN\rho_N C_S = 4\frac{m}{N} \mu_N C_S \,.   
\end{align*}
Based on the Bernstein inequality, we have 
\begin{align*}
 \mathbb{P}(\mu_N |X_{ab}(z) - X_{ab}(z_0)| \geq \varepsilon \mu_N ) & \leq 2 \exp \left( \frac{- \varepsilon^2 \mu_N^2}{2(\mbox{Var}(\mu_N (X_{ab}(z) - X_{ab}(z_0)) + \varepsilon \mu_N/3 )}\right) \notag\\
&\leq 2 \exp \left( \frac{- \varepsilon^2 \mu_N^2}{2(\mbox{Var}(\mu_N (O_{ab}(z) - O_{ab}(z_0)) + \varepsilon \mu_N/3 )}\right) \notag\\
&\leq 2 \exp \left( \frac{- \varepsilon^2 N \mu_N}{2(4 m C_S + \varepsilon N/3 )}\right)\,.   
\end{align*}
For $\varepsilon \leq 12 m C_S/ N$ and fixed $a,b, z, z_0$,
\begin{equation*}
  \mathbb{P}(|X_{ab}(z) - X_{ab}(z_0)| \geq \varepsilon )    \leq 2 \exp \left( - \frac{ N }{16 m C_S} \varepsilon^2 \mu_N  \right)  \,.
\end{equation*}
Then,
\begin{equation*}
  \mathbb{P}(\|X(z) - X(z_0)\|_\infty \geq \varepsilon )    \leq 8 \exp \left( - \frac{ N }{16 m C_S} \varepsilon^2 \mu_N  \right)  \,,
\end{equation*}
which establishes Eq~\eqref{eq:x_differ_con}.
\end{proof}

The proof of the main theorem combines the population version result in Lemma~\ref{prop.dcsbm.differ}, which holds under the high probability event established in Proposition~\ref{prop:aPPR_clustering}, and Lemma~\ref{lemma:inequalities}, which controls noise through concentration. 


\begin{proof}[\textbf{Proof of Theorem}~\ref{thm:min}]

According to Eq~\eqref{eq:conduct_to_F}, our goal is same as showing that there exists $n' - n_1 = \Omega(N)$ such that 
\begin{equation}\label{eq:thm_result}
    F\left( \frac{O(z)}{\mu_N}\right) - F\left( \frac{O(z_0)}{\mu_N}\right) \leq - \frac{1}{N} \Omega_P(\|z - z_0\|_1),
\end{equation}
where $z = \zeta'$ and $\zeta' =  \{ z = h(\cC_n)\mid n \in [n']\}$ for $\lambda_n$ satisfying~\ref{prop.dcsbm.assp.4}. Note that $\|z - z_0\|_1 = |n - n_1|$ according to the definition in Eq~\eqref{eq:label_equation}.

The proof technique is similar to \cite{bickel2015correction}. By Taylor expansion, 
\begin{equation}
   F\left( \frac{O(z)}{\mu_N}\right) - F(\tR S \tR^\top (z) ) 
    = \left. \frac{\partial F}{\partial M} \right|_{M = \tR S \tR^\top (e)} \mbox{vec}(X(z)) + \mathcal{O}(\|X(z)\|_\infty^2)\,;  
\end{equation}
and     
\begin{equation*}
    F\left( \frac{O(z_0)}{\mu_N}\right) - F(\tR S \tR^\top (z_0) ) 
    = \left. \frac{\partial F}{\partial M} \right|_{M = \tR S \tR^\top (z)} \mbox{vec}(X(z_0)) + \mathcal{O}(\|X(z_0)\|_\infty^2)\,,
\end{equation*}
where $\frac{\partial F}{\partial M}$ is the partial derivative with respect to the vectorized $M$.

$\frac{\partial F}{\partial M}$ is continuous with respect to $M$, so
\begin{equation*}
  \left. \frac{\partial F}{\partial M} \right|_{M = \tR S \tR^\top (z)} =  \left. \frac{\partial F}{\partial M} \right|_{M = \tR S \tR^\top (z_0)}  + \mathcal{O} ( \|\tR S \tR^\top (z) - \tR S \tR^\top (z_0)\|_\infty)\,.
\end{equation*}

Note that $ \|\tR S \tR^\top (z) - \tR S \tR^\top (z_0)\|_\infty = \mathcal{O} (\|z - z_0\|_1/ N )$. Then, we have
\begin{multline*}
   F\left( \frac{O(z)}{\mu_N}\right) - F(\tR S \tR^\top (z) )  -     F\left( \frac{O(z_0)}{\mu_N}\right) + F(\tR S \tR^\top (z_0) ) \\
   =  \left. \frac{\partial F}{\partial M} \right|_{M = \tR S \tR^\top (z_0)} \mbox{vec}(X(z) - X(z_0)) + \mathcal{O} \left(\frac{\|z - z_0\|_1}{N} \|X(z)\|_\infty \right) + \mathcal{O}(\|X(z)\|_\infty^2) + \mathcal{O}(\|X(z_0)\|_\infty^2)\,.
\end{multline*}

Therefore, there exists positive constants $C_1$, $C_2$, $C_3$, $C_4$ such that
\begin{multline*}
F\left( \frac{O(z)}{\mu_N}\right) -  F\left( \frac{O(z_0)}{\mu_N}\right) \\
\leq  F(\tR S \tR^\top (z) ) - F(\tR S \tR^\top (z_0) )  
+ C_1\|X(z) - X(z_0)\|_\infty \\ + C_2 \|X(z)\|_\infty \frac{\|z - z_0\|_1}{N} + C_3\|X(z)\|^2_\infty +C_4\|X(z_0)\|^2_\infty
\end{multline*}

Under the highly probability event described in Eq~\eqref{eq:prop_1_assump}, by Lemma~\ref{prop.dcsbm.differ}, there exist $n' - n_1 = \Omega(N)$ and a positive constant $C_0$ satisfying
\begin{equation*}
F(\tR S \tR^\top (z) ) - F(\tR S \tR^\top (z_0) )    \leq  -C_0 \frac{ \|z - z_0\|_1}{N} \,, \mbox{ for $z \in \zeta'$. }
\end{equation*} The other terms can be bounded by concentration, noting that assumption~\ref{prop.dcsbm.assp.4} implies $\lambda_N\to \infty$.
We write 
\begin{align}\label{eq:f_close}
 &\mathbb{P}\left( \max_{z \in \zeta': z \neq z_0} \left\{F\left( \frac{O(z)}{\mu_N}\right) - F\left( \frac{O(z_0)}{\mu_N} \right) \right\} \geq - C_0  \frac{\|z - z_0\|_1}{5N}  \right) \notag  \\
      & \leq  \mathbb{P}\left( \max_{z \in \zeta' : z \neq z_0} \{  - C_0 \frac{ \|z - z_0\|_1}{N}   + C_1\|X(z) - X(z_0)\|_\infty +   C_2 \|X(z)\|_\infty \frac{\|z - z_0\|_1}{N} \right. \\
    & \quad \left. + C_3\|X(z)\|^2_\infty +C_4\|X(z_0)\|^2_\infty \} \geq - C_0  \frac{\|z - z_0\|_1}{5N} \right) \notag  \\
  & \leq \mathbb{P}\left(\max_{z \in \zeta : z \neq z_0} \frac{\| X(z) - X (z_0) \|_{\infty}}{\|z - z_0\|_1 /N} \geq \frac{C_0}{ 5 C_1} \right) + \mathbb{P}\left(\max_{z \in \zeta}\| X(z)\|_{\infty}\geq \frac{C_0}{ 5 C_2} \right) \notag  \\ 
  &  \quad +  \mathbb{P}\left(\max_{z \in \zeta: z \neq z_0} \frac{\| X(z)\|^2_{\infty}}{\|z - z_0\|_1 /N}\geq \frac{C_0}{ 5 C_3} \right) +  \mathbb{P}\left( \max_{z \in \zeta : z \neq z_0 }\frac{\| X(z_0)\|^2_{\infty}}{\|z - z_0\|_1 /N}\geq \frac{C_0}{ 5 C_4} \right)\,.   
\end{align}

By Eq~\eqref{eq:x_con}, we have
\begin{equation}\label{eq:x_itself}
    \mathbb{P}\left(\max_{z \in \zeta}\| X(z)\|_{\infty}\geq \varepsilon \right) \longrightarrow 0 \mbox{ as } \lambda_N \longrightarrow \infty\,.
\end{equation}

According to Eq~\eqref{eq:x_differ_con}, we have
\begin{equation*}
    \mathbb{P}\left(\max_{z \in \zeta :\|z - z_0\|_1 = m}\| X(z) - X (z_0) \|_{\infty} \geq \varepsilon \frac{m}{N}\right) \leq 16 \exp \left( - \frac{ m }{16 N C_S} \varepsilon^2 \mu_N  \right)\,.
\end{equation*}
Then, 
\begin{align}\label{eq:x_differ}
 \mathbb{P}\left(\max_{z \in \zeta : z \neq z_0} \frac{\| X(z) - X (z_0) \|_{\infty}}{\|z - z_0\|_1 /N}   \geq \varepsilon \right) & \leq \sum^N_{m = 1}  \mathbb{P} \left( \max_{z \in \zeta : \|z - z_0\|_1 = m}  \frac{\| X(z) - X (z_0) \|_{\infty}}{m /N} \geq \varepsilon \right) \notag  \\  
   & \leq \sum^N_{m = 1} 16 \exp \left( - m \frac{ \varepsilon^2 \lambda_N }{16 C_S}   \right)
    \longrightarrow 0 \mbox{ as } \lambda_N \longrightarrow \infty\,.   
\end{align}

Also, by Eq~\eqref{eq:x_con_2}, we have
\begin{equation*}
    \mathbb{P}\left(\max_{z \in \zeta : \|z - z_0\|_1 = m} \| X(z)\|^2_{\infty} \geq  \varepsilon \frac{m}{N}\right) \leq 16 \exp \left(- \frac{  m }{8 C_S N} \varepsilon \mu_N \right)\,.
\end{equation*}
Similar to Eq~\eqref{eq:x_differ},
\begin{align} \label{eq:x_sqr}
\mathbb{P}\left(\max_{z \in \zeta : z \neq z_0} \frac{\| X(z)\|^2_{\infty} }{\|z - z_0\|_1 /N}   \geq \varepsilon \right) & \leq \sum^N_{m = 1}  \mathbb{P} \left( \max_{z \in \zeta : \|z - z_0\|_1 = m}  \frac{\| X(z)\|^2_{\infty} }{m /N} \geq \varepsilon \right)  \notag \\  
  & \leq \sum^N_{m = 1} 16 \exp \left(-  m \frac{\varepsilon \lambda_N  }{8 C_S}  \right)
    \longrightarrow 0 \mbox{ as } \lambda_N \longrightarrow \infty\,.    
\end{align}

Again, by Eq~\eqref{eq:x_con_3}, 
\begin{equation*}
    \mathbb{P}(\|X(z_0) \|^2_\infty \geq \varepsilon \frac{m}{N}) \leq 8 \exp \left(- \frac{ m }{8  C_S N} \varepsilon  \mu_N \right)  \mbox{  for all } m \in [N]\,.
\end{equation*}
Then,
\begin{align} \label{eq:xz_sqr}
\mathbb{P}\left(\max_{z \in \zeta : z \neq z_0} \frac{\| X(z_0)\|^2_{\infty} }{\|z - z_0\|_1 /N}   \geq \varepsilon \right) & \leq \sum^N_{m = 1}  \mathbb{P} \left(  \frac{\| X(z_0)\|^2_{\infty} }{m /N} \geq \varepsilon \right)   \notag\\  
 &  \leq \sum^N_{m = 1}  8 \exp \left(- m \frac{  \varepsilon  \lambda_N }{8  C_S} \right)
    \longrightarrow 0 \mbox{ as } \lambda_N \longrightarrow \infty\,.    
\end{align}

Combining Eq~\eqref{eq:f_close}, \eqref{eq:x_itself}, \eqref{eq:x_differ}, \eqref{eq:x_sqr}, and \eqref{eq:xz_sqr}, we have
\begin{equation*}
  F\left( \frac{O(z)}{\mu_N}\right) - F\left( \frac{O(z_0)}{\mu_N} \right) \leq - \frac{1}{N}\Omega_P(\|z - z_0\|_1)\, \mbox{ for } z \in \zeta'\,.
\end{equation*}

\end{proof}

\clearpage

\section{Additional results for simulations and real data}

\subsection{Simulation studies in Section~\ref{sec:simulation}}~\label{app.sec:simulation}
\spacingset{1}
\vspace{-1cm}
\begin{figure}[h!]
\begin{subfigure}[b]{0.5\textwidth}
    \centering
   \includegraphics[width=\textwidth]{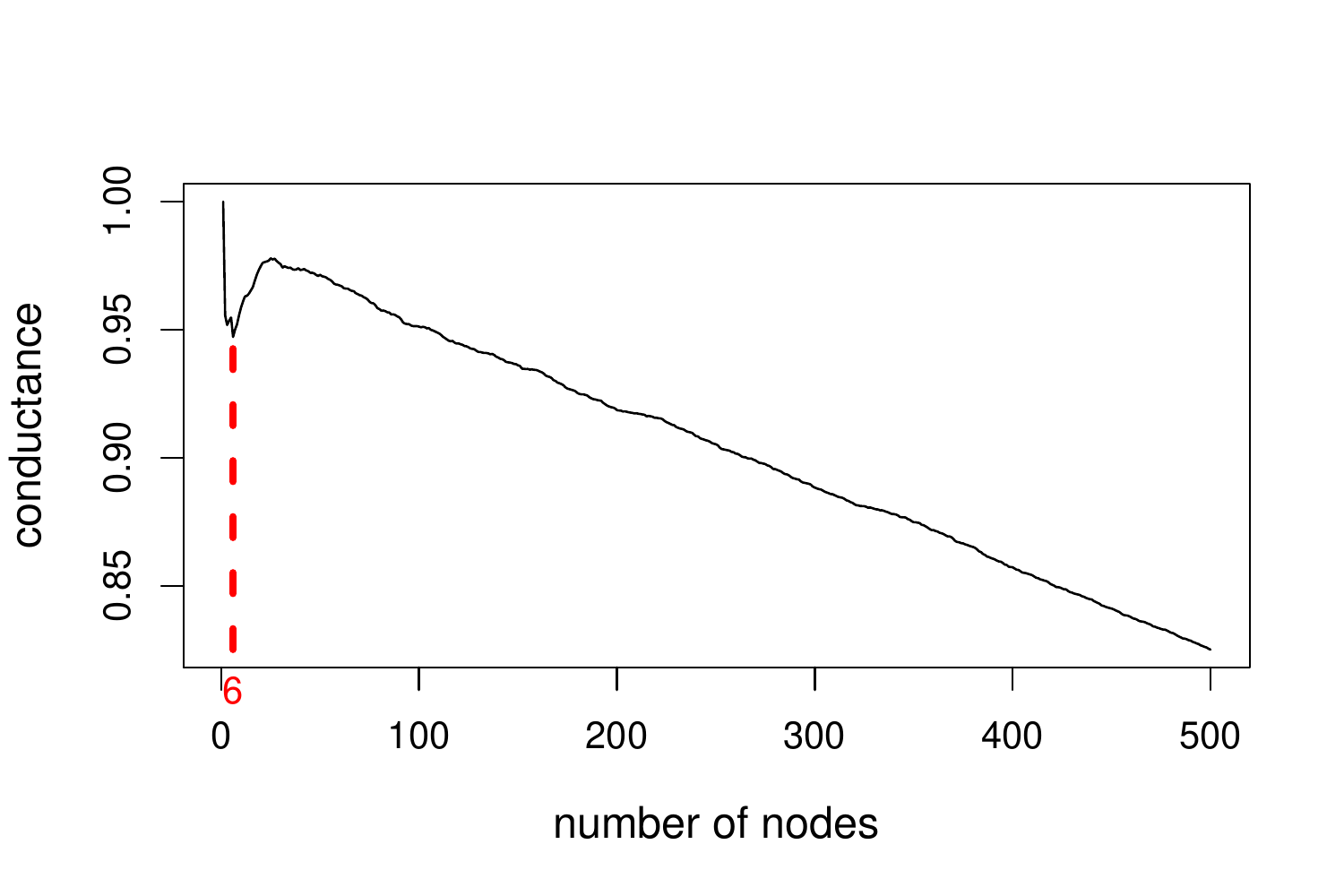}
     \caption{1 seed} 
\end{subfigure}
\begin{subfigure}[b]{0.5\textwidth}
    \centering
   \includegraphics[width=\textwidth]{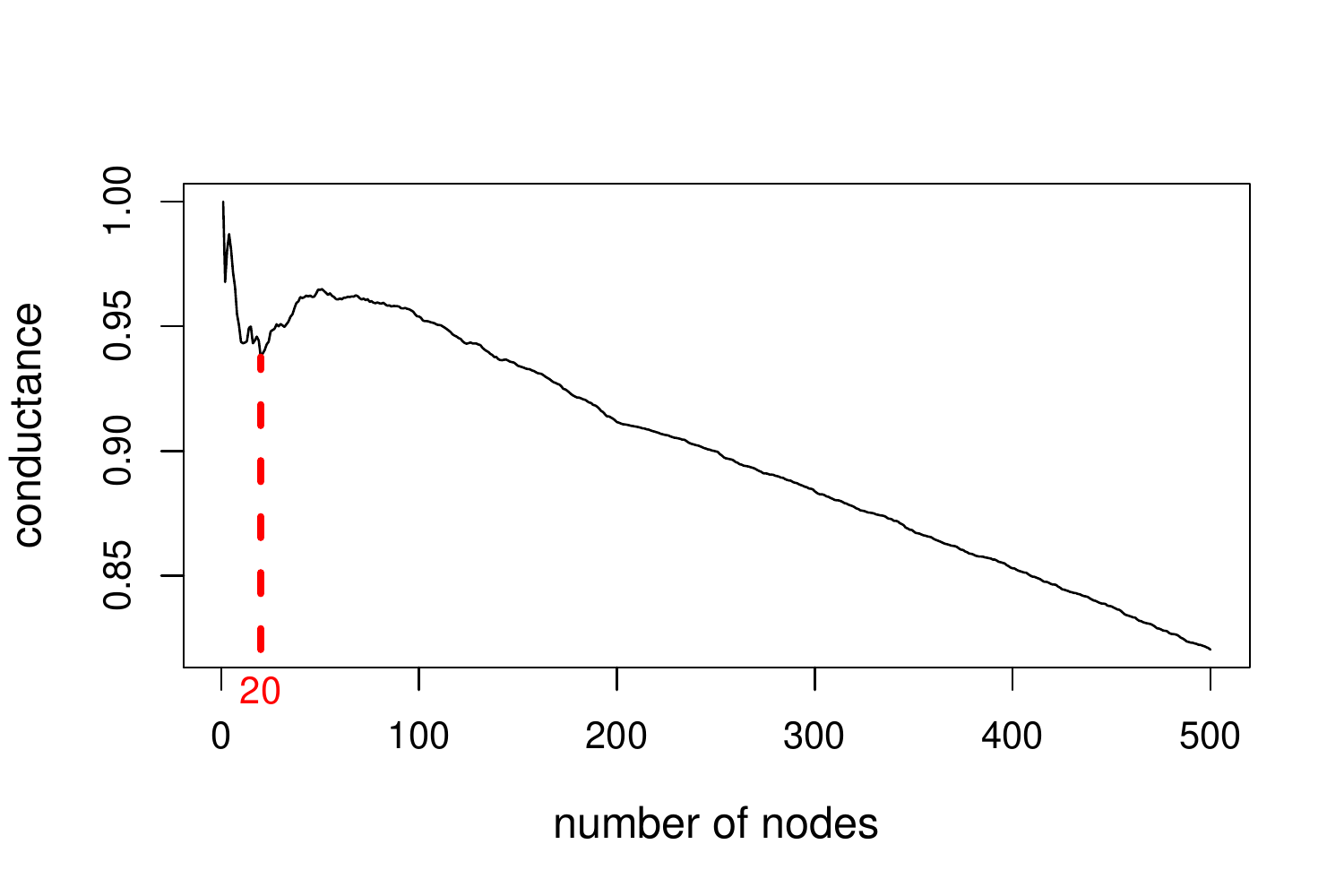}
      \caption{5 seeds} 
\end{subfigure}
\begin{subfigure}[b]{0.5\textwidth}
    \centering
   \includegraphics[width=\textwidth]{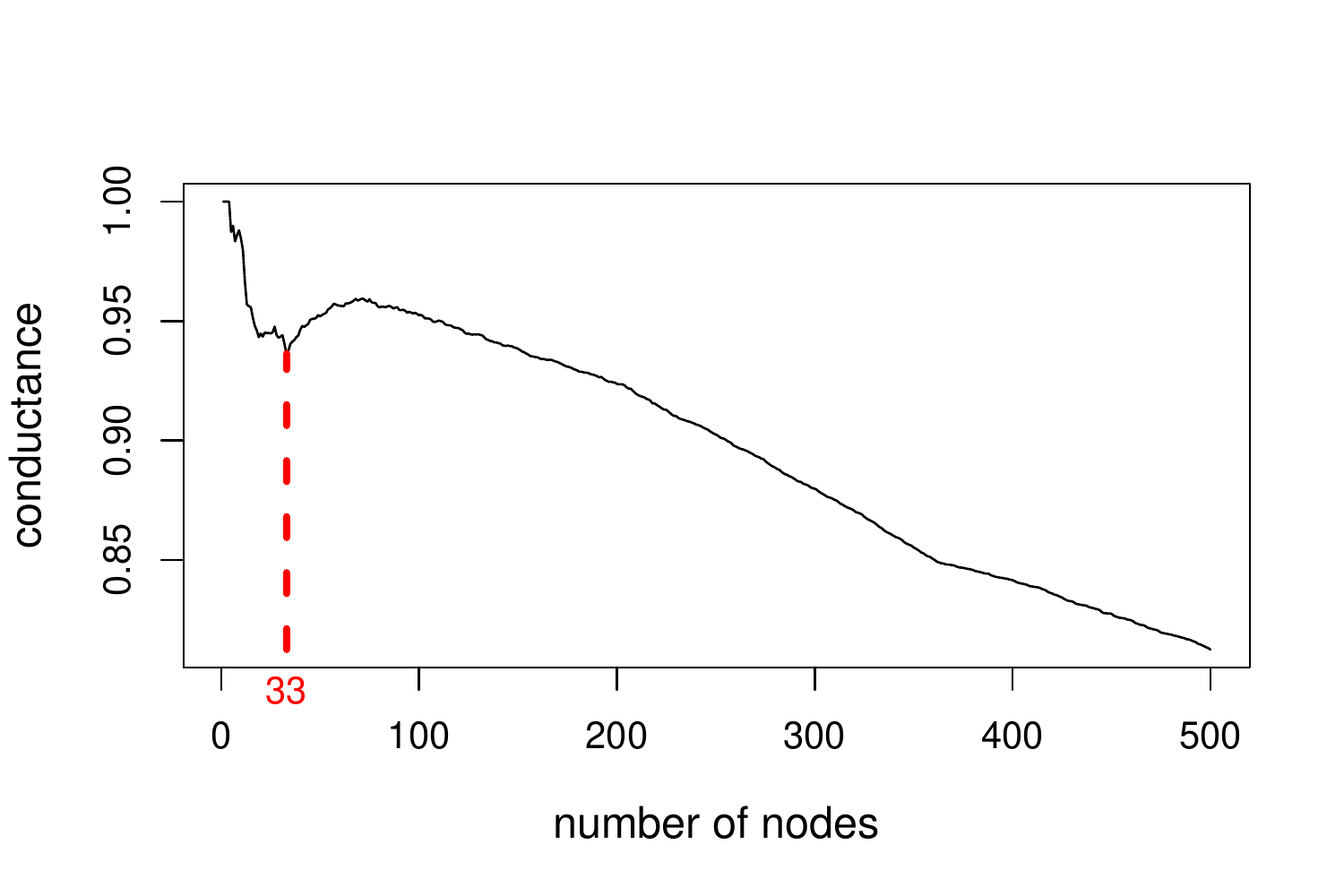}
      \caption{10 seeds} 
\end{subfigure}
\begin{subfigure}[b]{0.5\textwidth}
    \centering
   \includegraphics[width=\textwidth]{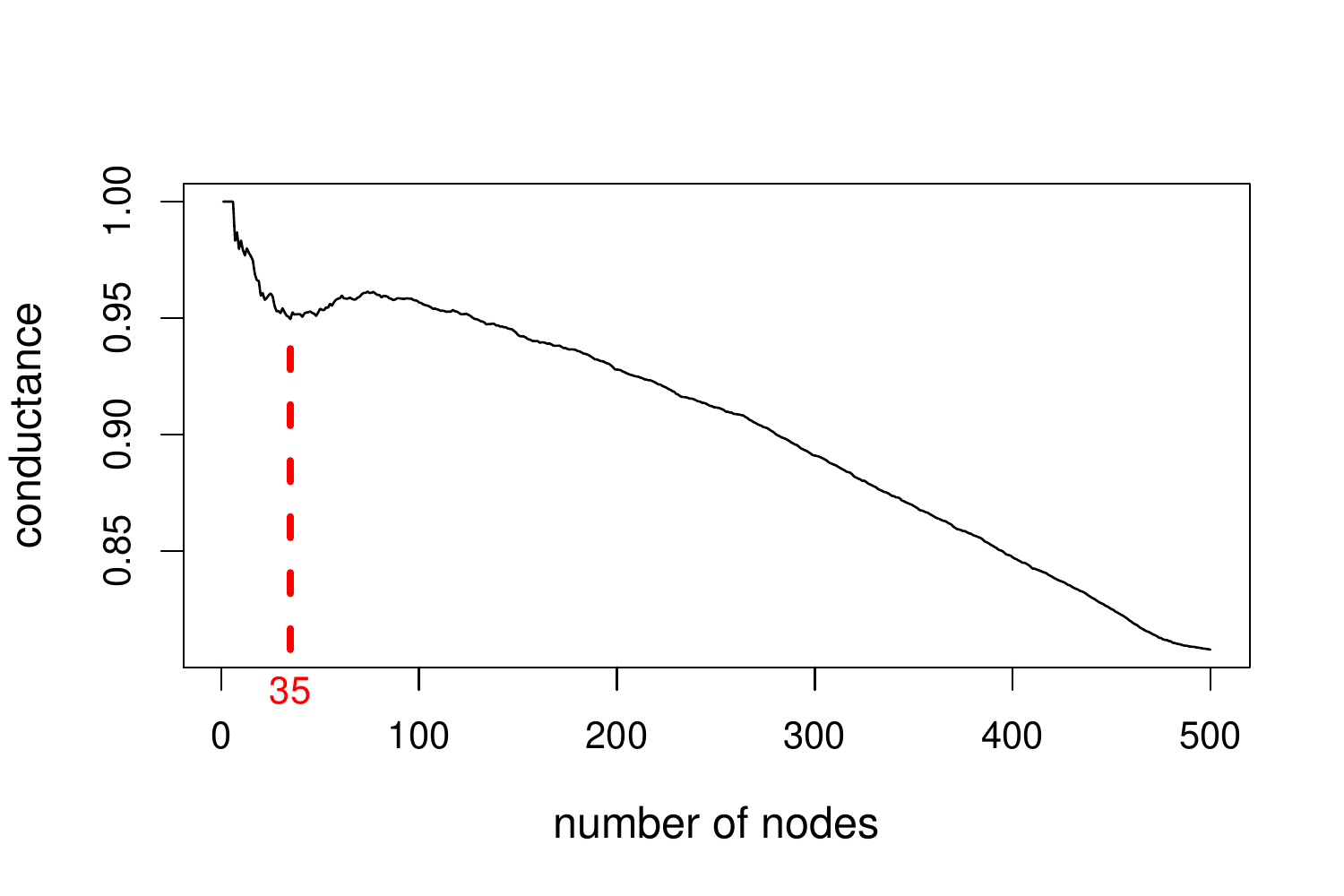}
      \caption{15 seeds} 
\end{subfigure}
\begin{subfigure}[b]{0.5\textwidth}
    \centering
   \includegraphics[width=\textwidth]{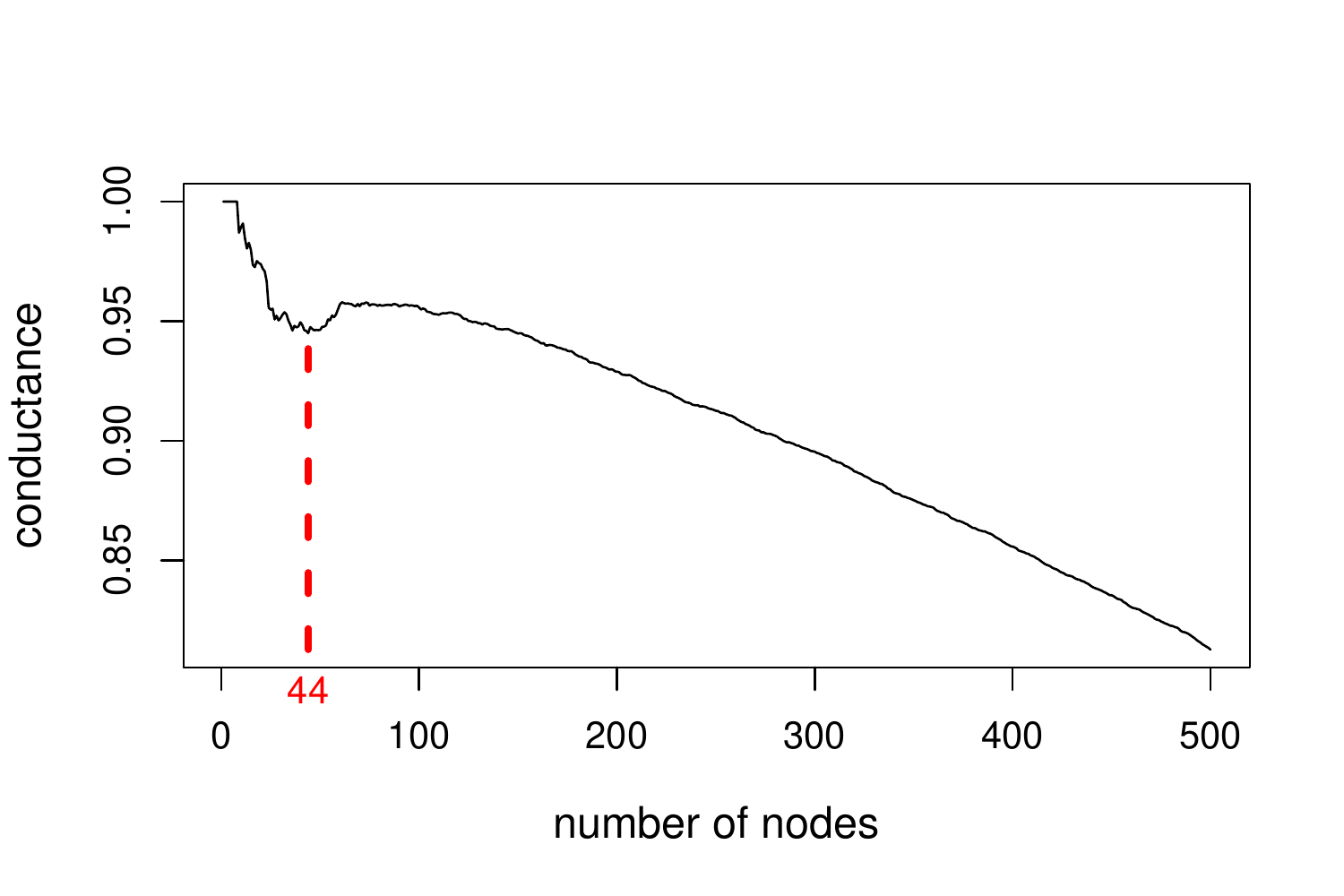}
     \caption{20 seeds} 
\end{subfigure}
\caption{Examples of conductance plots for different seed numbers when $\alpha = 0.15$, under the setting of Table~\ref{table:sim_results_1}. The dotted line indicates the local minimum selected.} \label{fig:conduct_exp}
\end{figure}

\begin{table}[h]
\centering
\footnotesize
\begin{tabular}{l l|llll}
\hline\hline
&  & \multicolumn{2}{c}{Precision}  & \multicolumn{2}{c}{Recall}  \\ \hline
 $\alpha$ & \# of seeds & Mean(\%) & SD(\%) & Mean(\%) & SD(\%) \\ \hline
\multirow{4}{*}{0.05}   & 1   & 97.00 & 0.14 & 98.77 & 1.54 \\
                        & 5   & 97.49 & 0.22 & 99.23 & 0.94\\
                        & 10  & 97.93 & 0.21 & 99.43 & 0.63 \\
                        & 15  & 98.30 & 0.30 & 99.51 & 0.45 \\ 
                        & 20  & 98.59 & 0.40 & 99.58 & 0.32 \\\hline \hline

\multirow{4}{*}{0.25}   & 1   & 97.00 & 0.15 & 98.76 & 1.54 \\
                        & 5   & 97.48 & 0.21 & 99.24 & 0.94\\
                        & 10  & 97.93 & 0.22 & 99.41 & 0.63 \\
                        & 15  & 98.30 & 0.29 & 99.50 & 0.45 \\ 
                        & 20  & 98.60 & 0.40 & 99.55 & 0.36 \\\hline \hline
\end{tabular}
\caption{The means and standard deviations of precision and recall rates for local clustering under different $\alpha$ values and numbers of seeds. Each setting is simulated 50 times.}\label{table:sim_results_2}
\end{table}

\clearpage

\subsection{ More simulation results}~\label{app.sec:more_simulation}

We compare the local clustering procedure with common global clustering techniques including the usual spectral clustering, SCORE \citep{J15} and the more recently proposed SCORE+ \citep{JZS18}. We find that SCORE+ performs better than SCORE in most of our experimental settings, so we only present the results of SCORE+ in what follows. We consider the three settings below.

\spacingset{1}
\begin{table}[h]
\centering
\footnotesize
\begin{tabular}{l l|llll}
\hline\hline
&  & \multicolumn{2}{c}{Precision}  & \multicolumn{2}{c}{Recall}  \\ \hline
 & Method & Mean(\%) & SD(\%) & Mean(\%) & SD(\%) \\ \hline
\multirow{3}{*}{Setting 1}  & Local & 97.49 & 0.50 & 99.11 &  0.45 \\ 
     & Spectral  &  97.18 &  2.15 & 94.19 &  15.56\\ 
     & SCORE+    &  98.30 & 0.32 & 99.87 & 0.07 \\ \hline
\multirow{3}{*}{Setting 2}  & Local & 97.53 & 0.47 & 99.35 & 0.31 \\ 
     & Spectral  &  97.34 &  1.41 & 91.65 &  17.63 \\
     & SCORE+    & 94.71  &  1.38 & 56.28 &  14.70 \\ \hline
\multirow{3}{*}{Setting 3}  & Local & 98.60 & 0.40   &  99.56  & 0.37 \\ 
     & Spectral  &  96.97 &  0.34 & 52.47 &  7.69 \\ 
     & SCORE+    &   96.77 & 0.01 &  50.03 & 0.06 \\ \hline \hline
\end{tabular}
\caption{Means and standard deviations of precision and recall for local clustering, spectral clustering and SCORE+.
Each setting is repeated in 50 simulations.}\label{table:sim_results}
\end{table}

\noindent\textbf{Setting 1:}  $n_1 = 150$. The preference vector has $\pi_i = 1/50$ for $i = 1, \ldots, 50$, $\pi_i = 0$ for others. For the local clustering method, we observe that a local minimum usually occurs for $n<200$. Thus we search for the minimum point in the range $1-200$. Figure~\ref{fig:conduct_exp_2_1} is an example of a conductance plot under this setting. The local minimum is obvious at the point $n=129$ in this example. 

\noindent\textbf{Setting 2:}  $n_1 = 100$.   The preference vector has $\pi_i = 1/30$ for $i = 1, \ldots, 30$, $\pi_i = 0$ for others. In this case, we search for the local minimum within the range $1-150$. Figure~\ref{fig:conduct_exp_2_2} gives an example of a conductance plot under setting 2. Again the local minimum is clear in the plot.

\noindent\textbf{Setting 3:}  $n_1 = 50$.   The preference vector has $\pi_i = 1/20$ for $i = 1, \ldots, 20$, $\pi_i = 0$ for others. Here, we search for the local minimum in the range $1 - 55$, as mentioned in Section~\ref{sec:simulation}.

All the other parameters (e.g., $K$, $B$ and $\theta_i$) are the same as in Section~\ref{sec:simulation}.
The teleportation constant $\alpha$ is set to 0.15 as before. 
Note that the number of seeds in $\pi$ decreases when the size of block 1 decreases, as we expect fewer seeds to be available for smaller community sizes. 

For each setting, we calculate the precision and recall from 50 simulations and record their means and standard deviations in Table~\ref{table:sim_results}.
All the methods have high precision and recall rates in Setting 1. However, as $n_1$ decreases and the two block sizes become more imbalanced in Settings 2 and 3, the performance of spectral clustering and SCORE+ become worse, whereas local clustering remains stable with high averages and small standard deviations. As shown in Table~\ref{table:sim_results_1} and Table~\ref{table:sim_results_2}, which examine the effect of $\alpha$ and seed number under Setting 3, local clustering has slightly higher average precision and substantially higher average recall than the other two methods even when only a single seed is used. We also note that the standard deviations of local clustering are smaller than those of spectral clustering in all the settings. More detailed distributions of these precision and recall rates under the three settings can be found in the violin plots of Figure~\ref{fig:label_violin}.


\spacingset{1}

\begin{figure}[h!]
\begin{subfigure}[b]{0.5\textwidth}
    \centering
   \includegraphics[width=\textwidth]{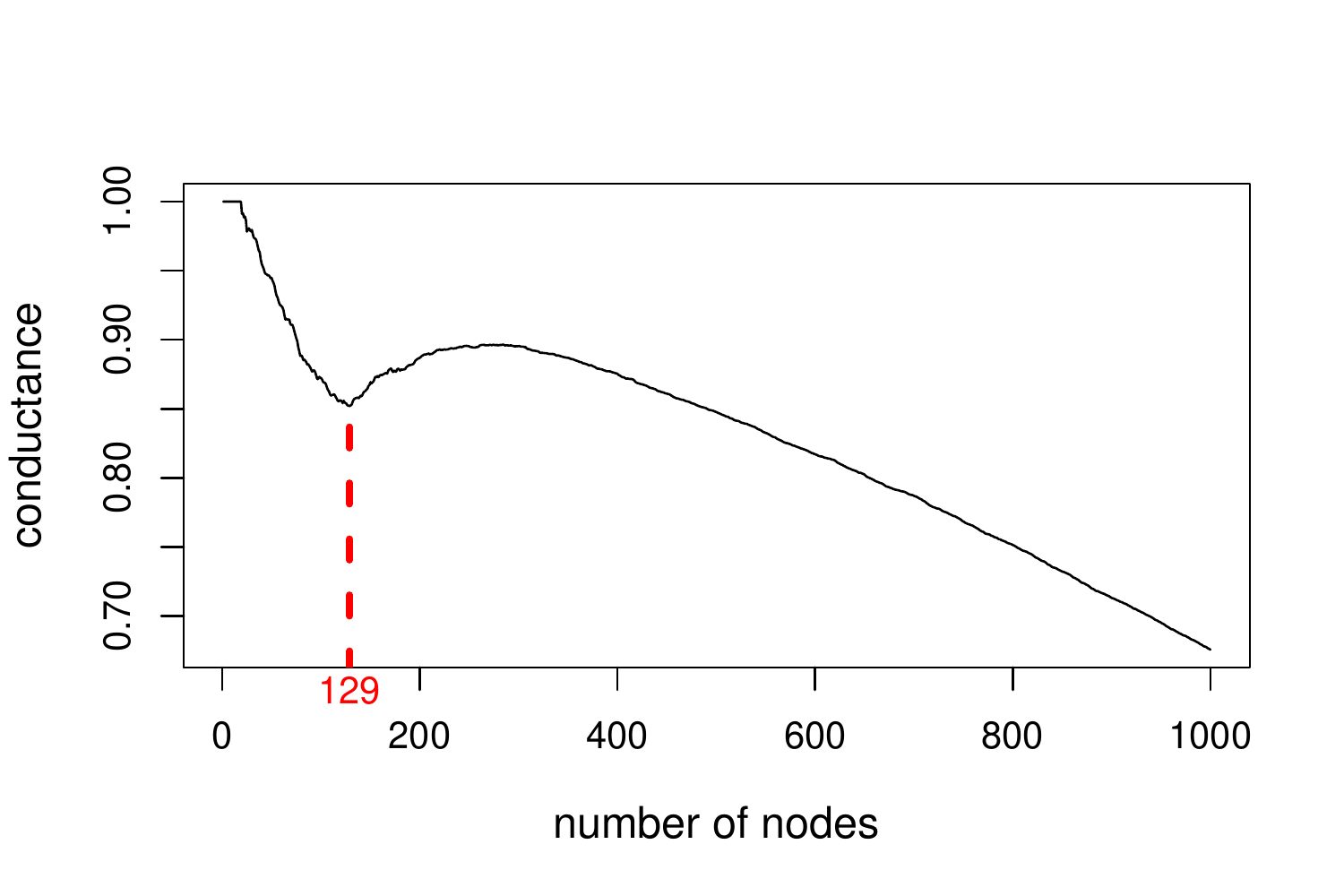}
     \caption{Setting 1}\label{fig:conduct_exp_2_1}
\end{subfigure}
\begin{subfigure}[b]{0.5\textwidth}
    \centering
   \includegraphics[width=\textwidth]{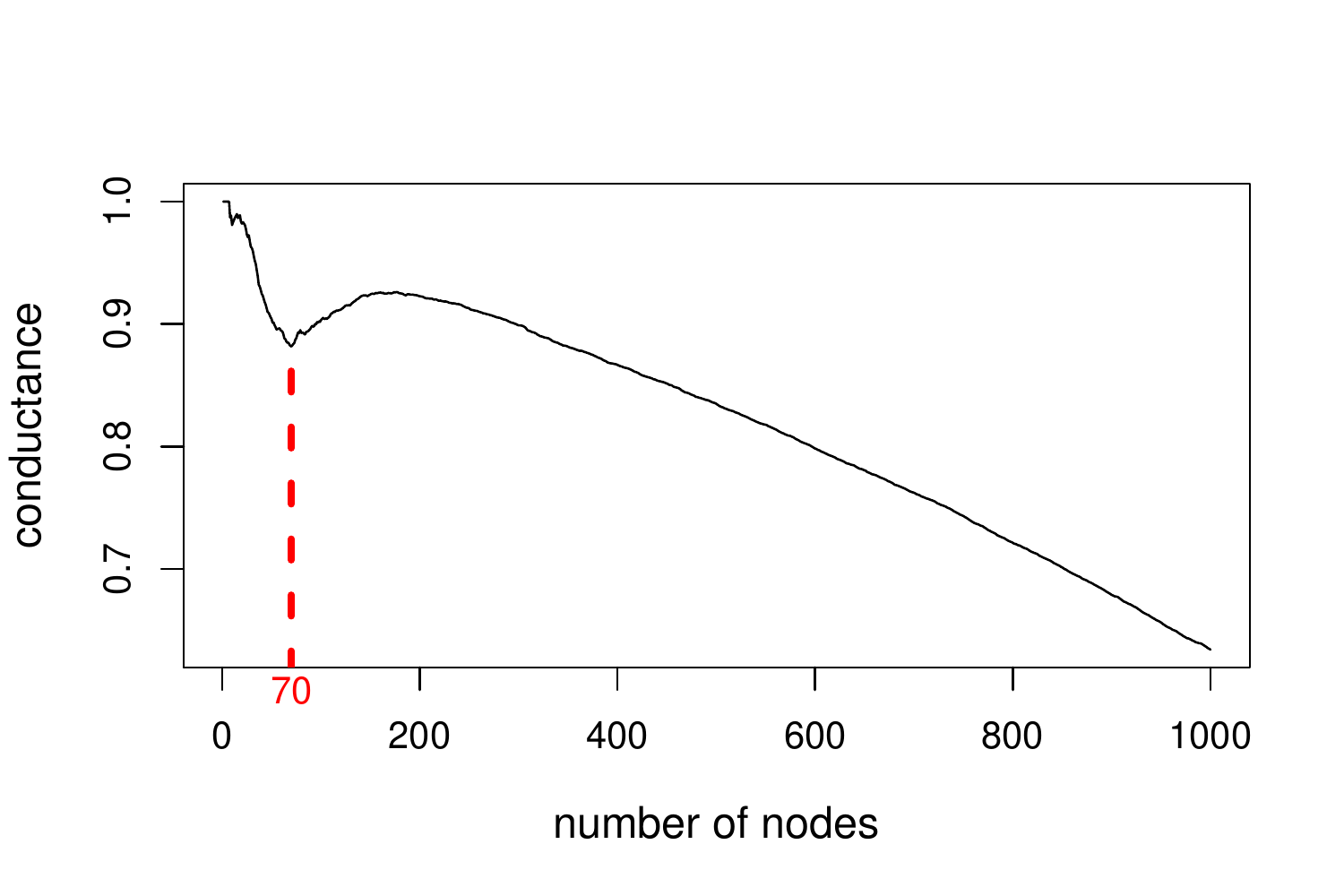}
      \caption{Setting 2}\label{fig:conduct_exp_2_2}
\end{subfigure}
\caption{Examples of conductance plots under Settings 1 and 2. More examples under Setting 3 can be found in Figure~\ref{fig:conduct_exp}.} \label{fig:conduct_exp_2}
\end{figure}

\begin{figure}[h!]
\begin{subfigure}[b]{1\textwidth}
 \caption{Setting 1} 
     \scalebox{0.35}{\includegraphics{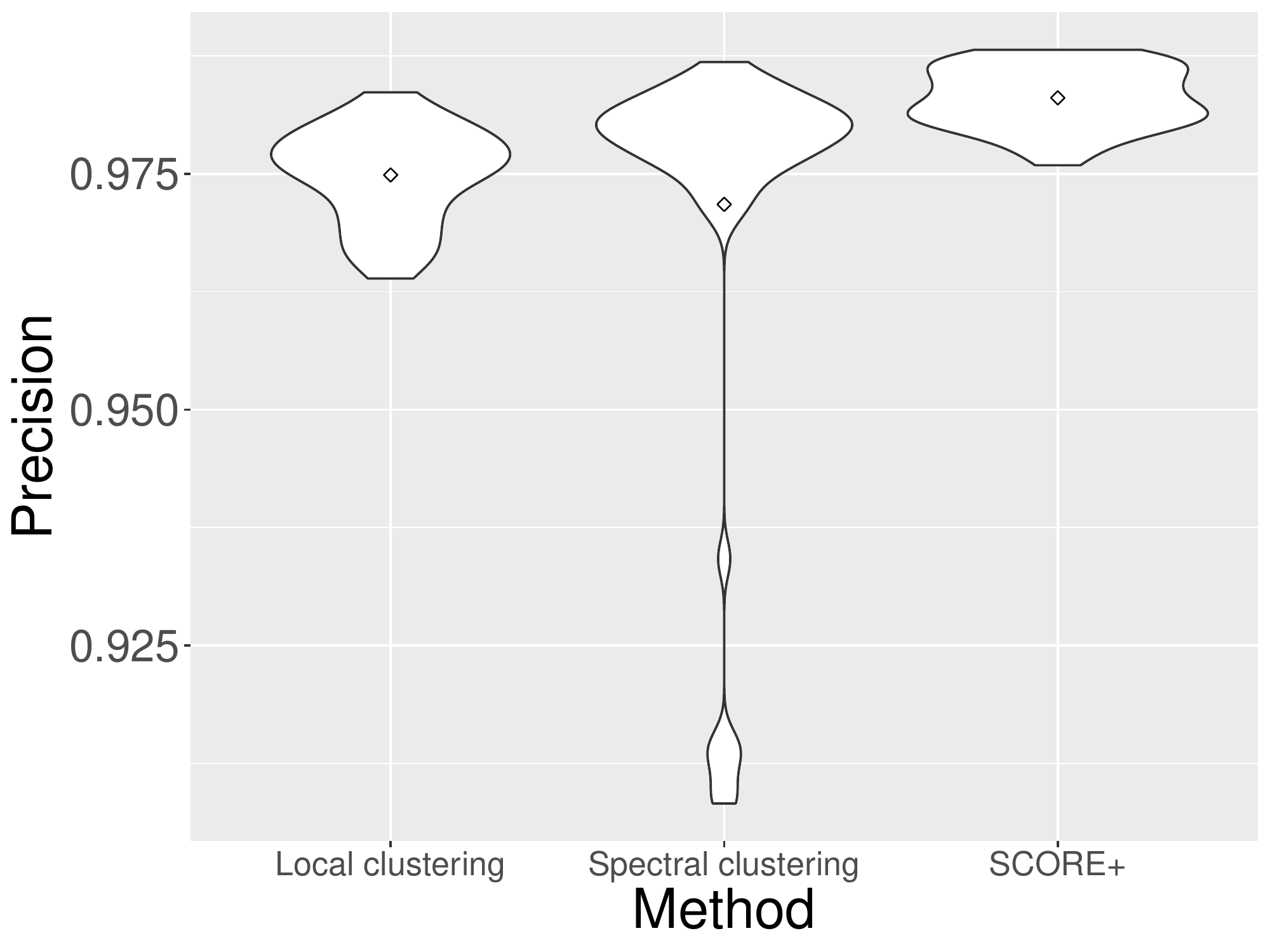}}
    \scalebox{0.35}{\includegraphics{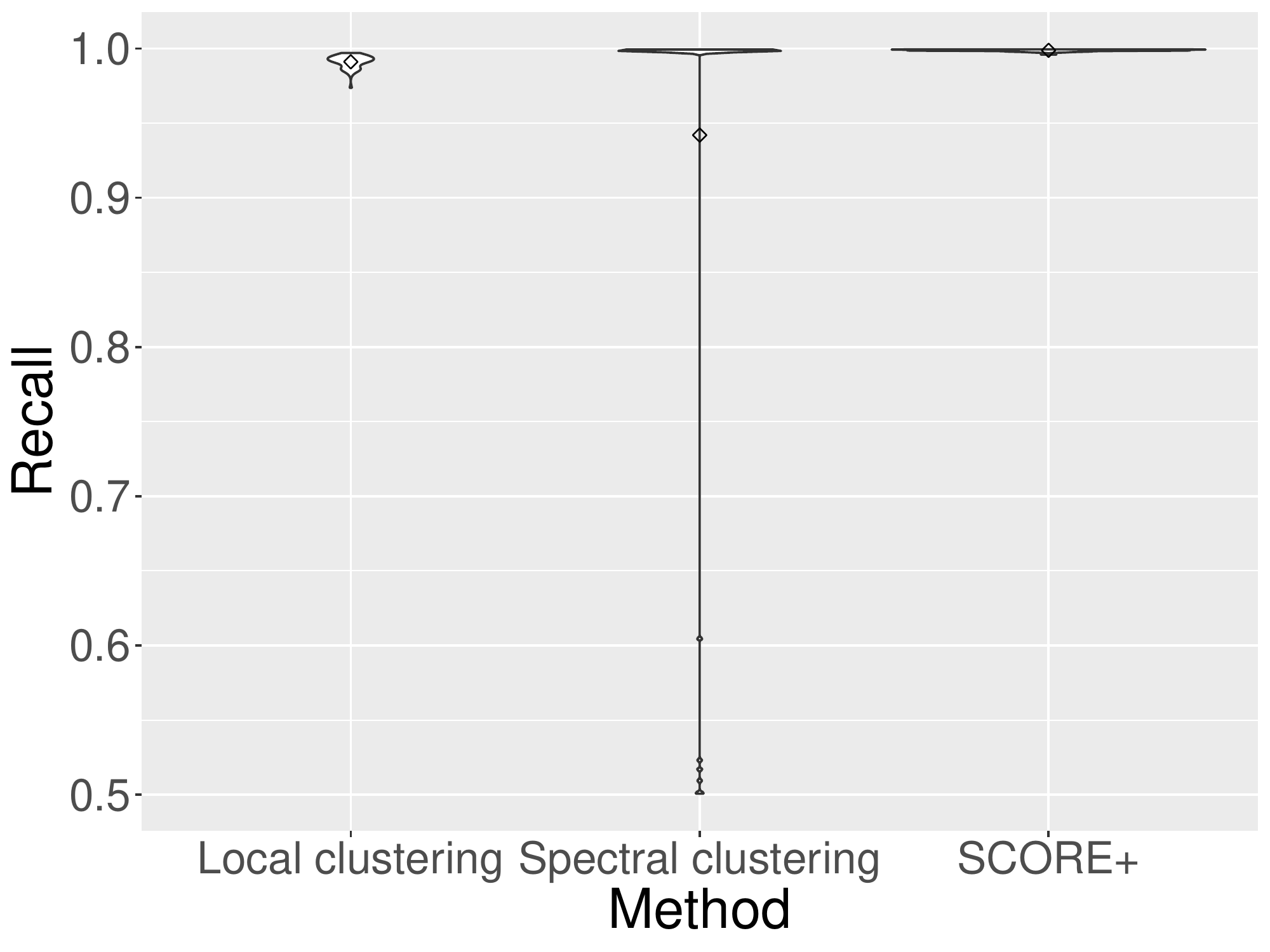}} 
\end{subfigure}
 \begin{subfigure}[b]{1\textwidth}
 \caption{Setting 2} 
     \scalebox{0.35}{\includegraphics{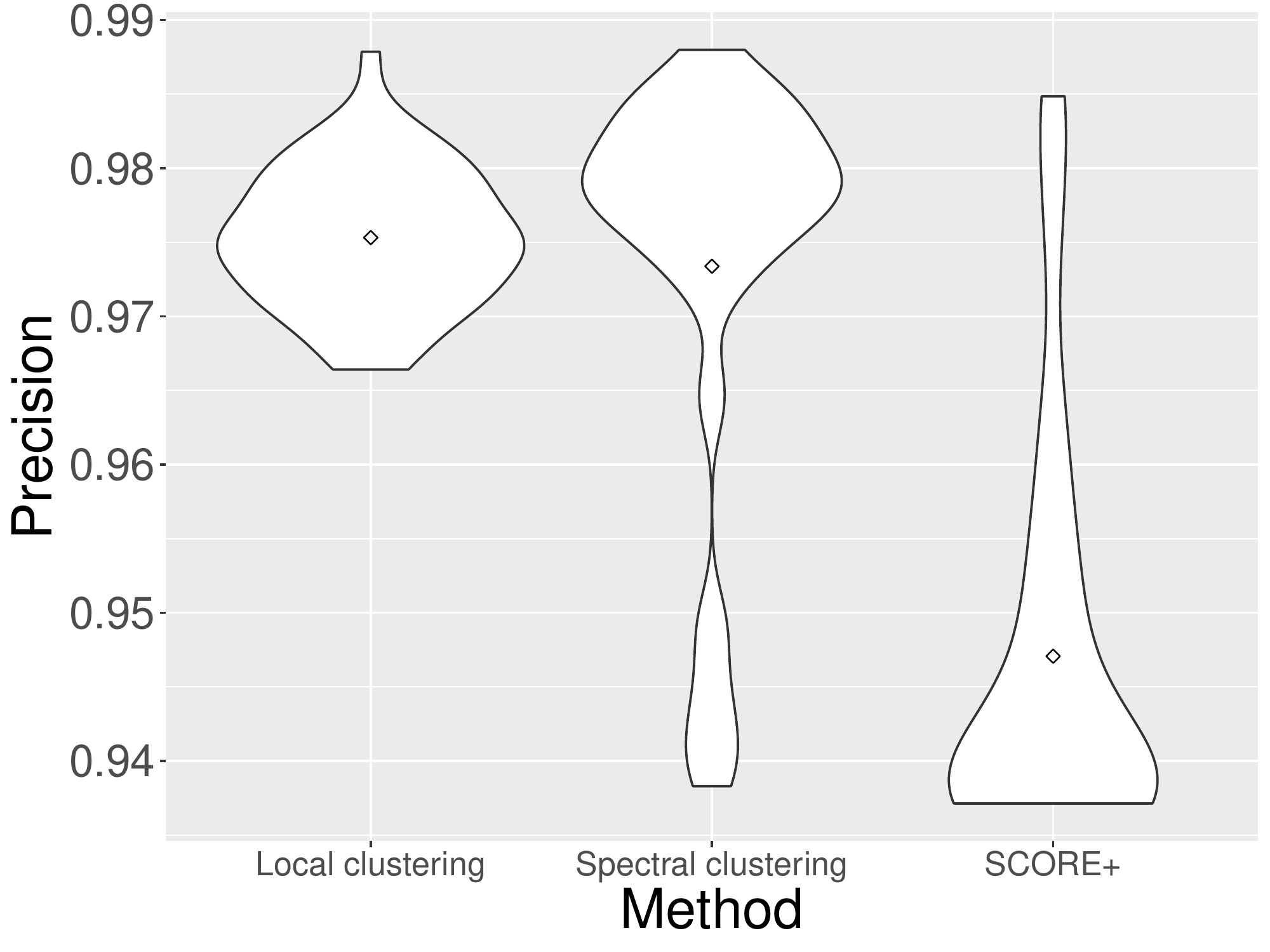}}
    \scalebox{0.35}{\includegraphics{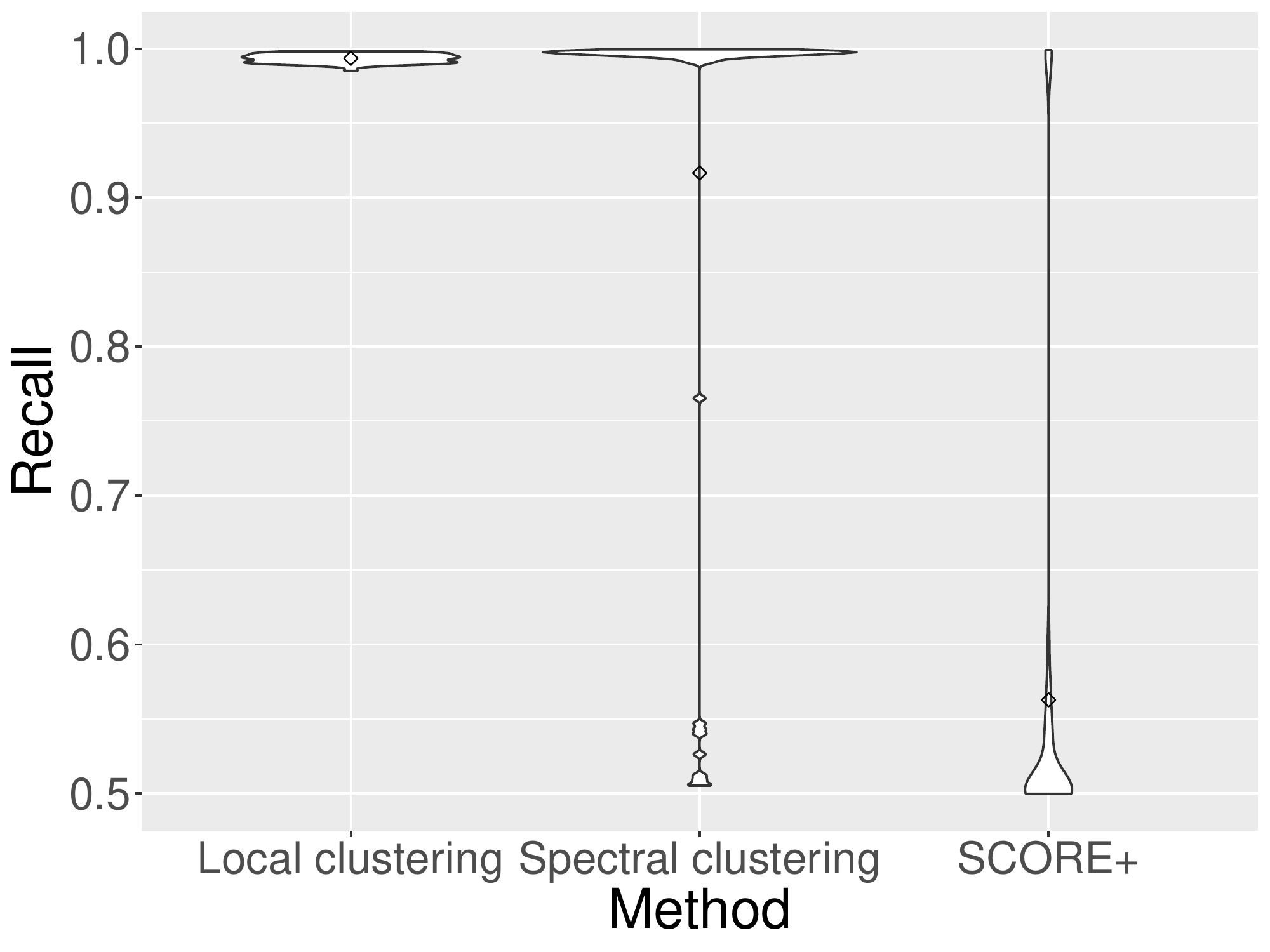}} 
\end{subfigure}
\begin{subfigure}[b]{1\textwidth}
    \caption{Setting 3} 
     \scalebox{0.35}{\includegraphics{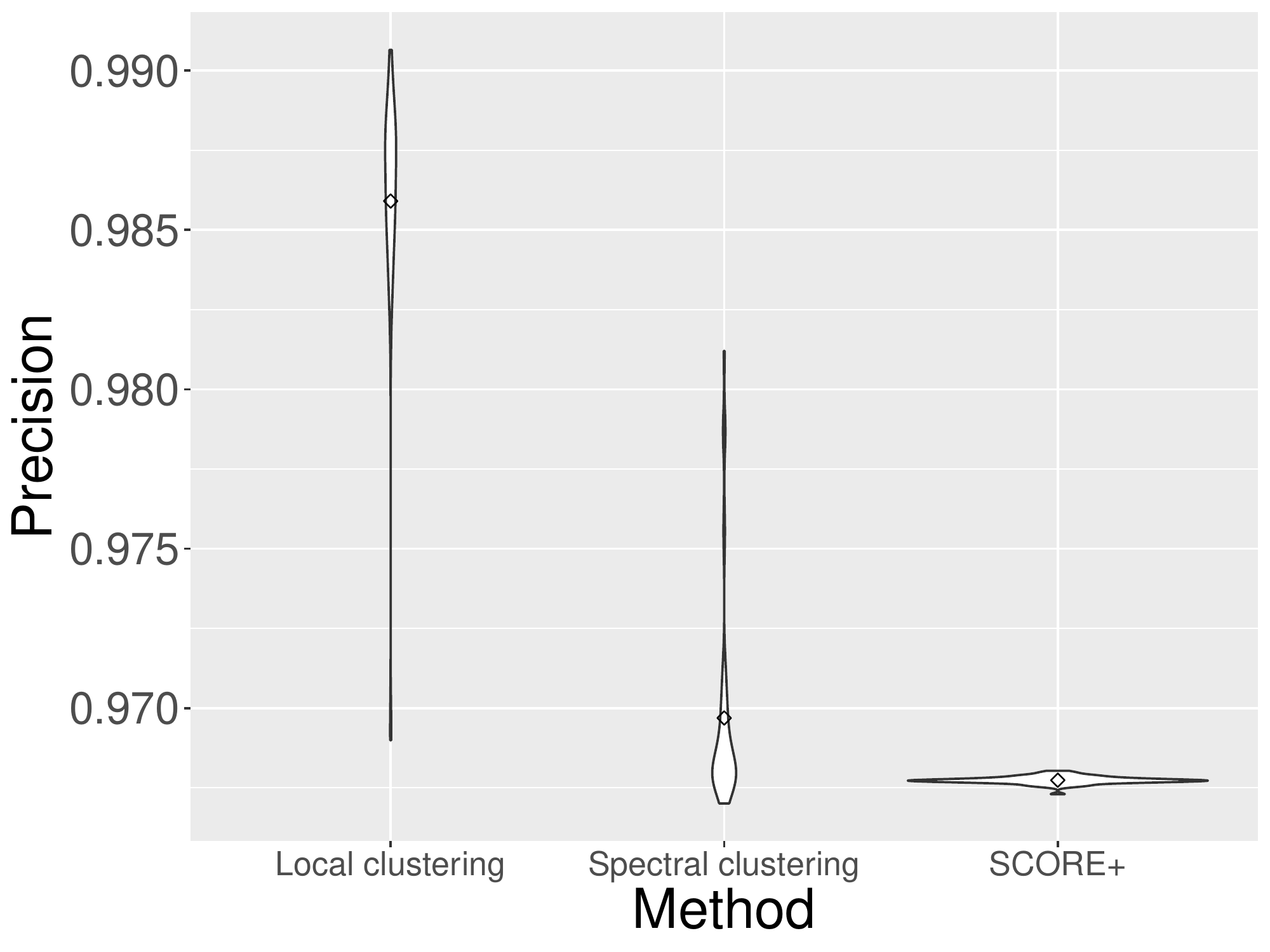}}
    \scalebox{0.35}{\includegraphics{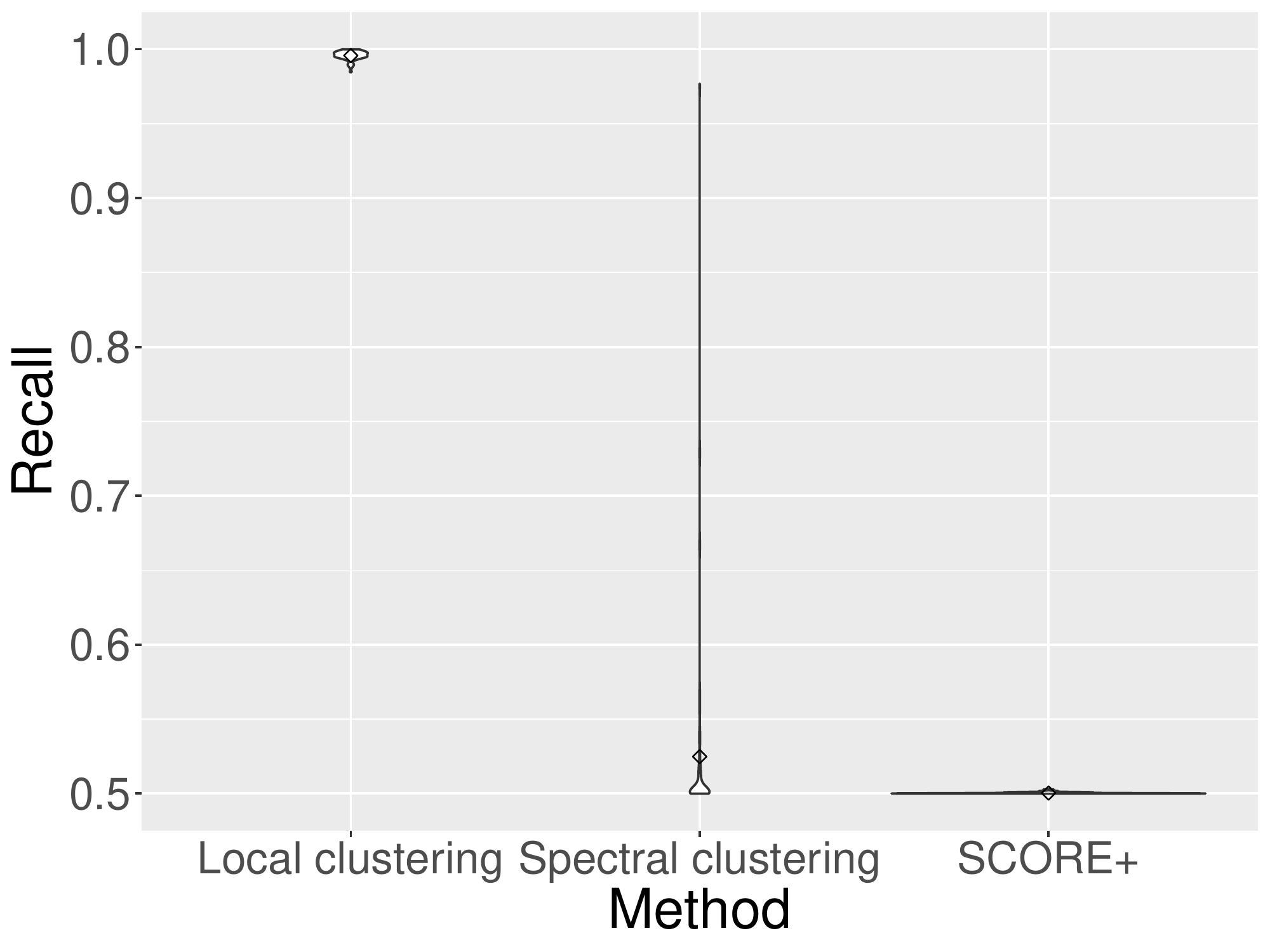}}  
\end{subfigure}
\caption{The violin plots of precision and recall for Setting 1-3.}\label{fig:label_violin}
\end{figure}

\clearpage

\subsection{Threshold selection and conductance plots for case studies in Section~\ref{sec:case_studies}}\label{app.sec:case_study}

\begin{figure}[h]
     \centering
     \begin{subfigure}[b]{0.35\textwidth}
    \centering
    \includegraphics[width=\textwidth]{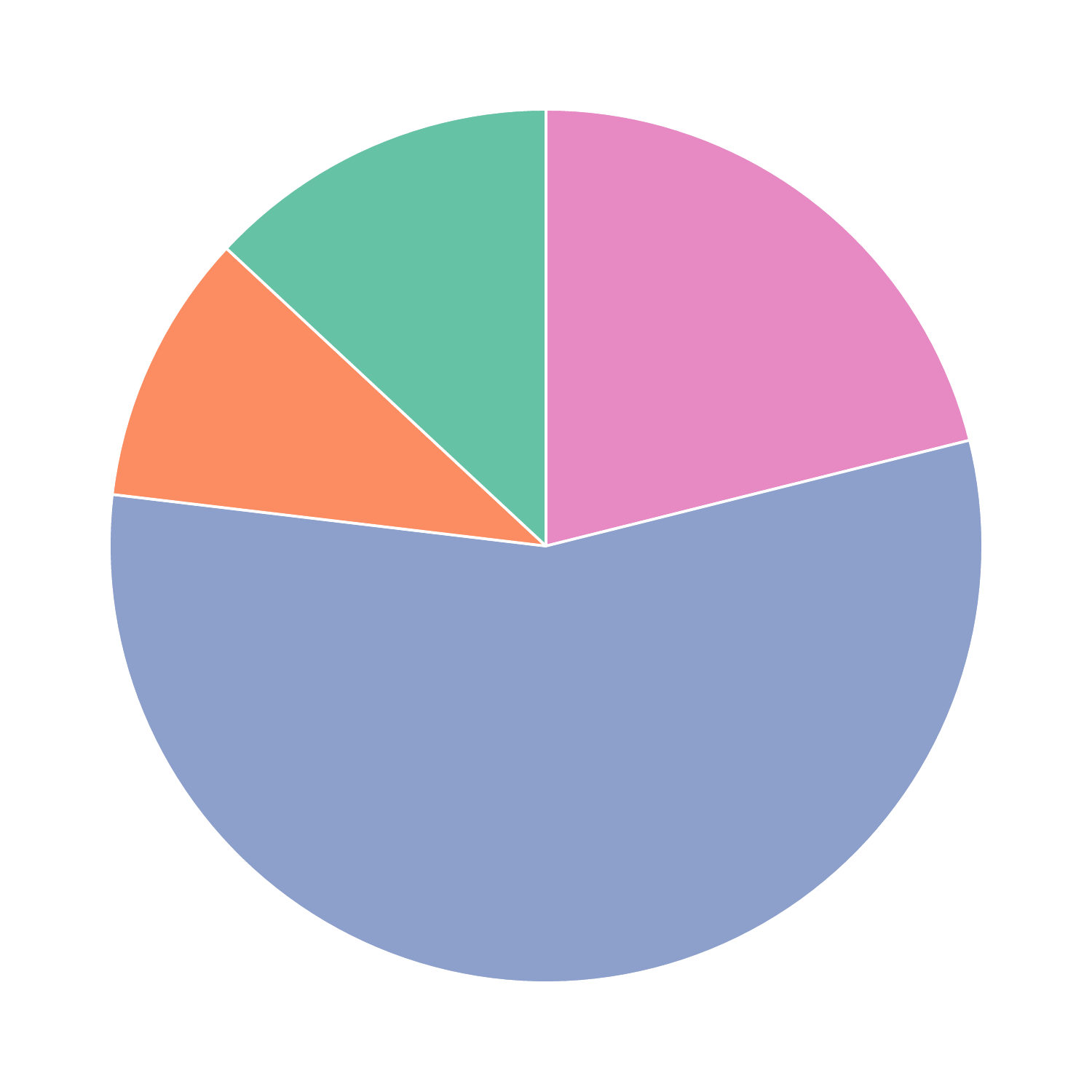}
    \caption{}\label{fig:pie_single_cell}
     \end{subfigure}
      \begin{subfigure}[b]{0.35\textwidth}
    \centering
    \includegraphics[width=\textwidth]{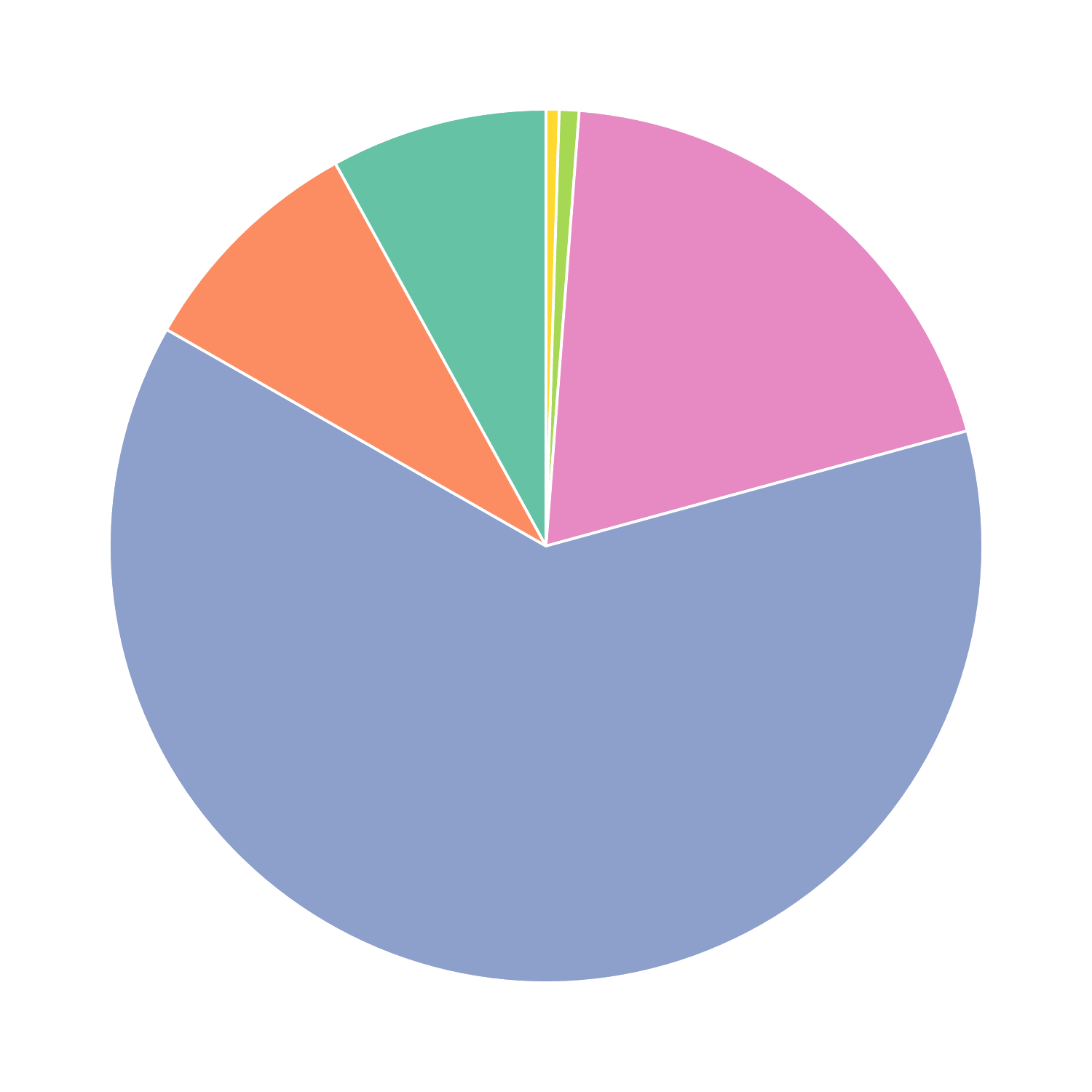}
    \caption{}\label{fig:pie_flu}
     \end{subfigure}
    \begin{subfigure}[b]{0.12\textwidth}
    \centering
    \includegraphics[width=\textwidth]{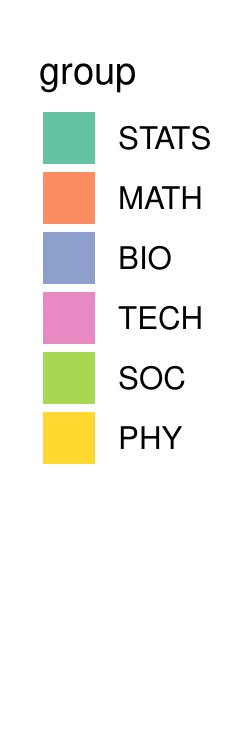}
     \end{subfigure}
 \caption{ The categories of citing papers related to: (a) single-cell transcriptomics; (b) flu.}
\end{figure}

\begin{figure}[h]
\centering
\begin{subfigure}[b]{0.45\textwidth}
    \centering
   \includegraphics[width=\textwidth]{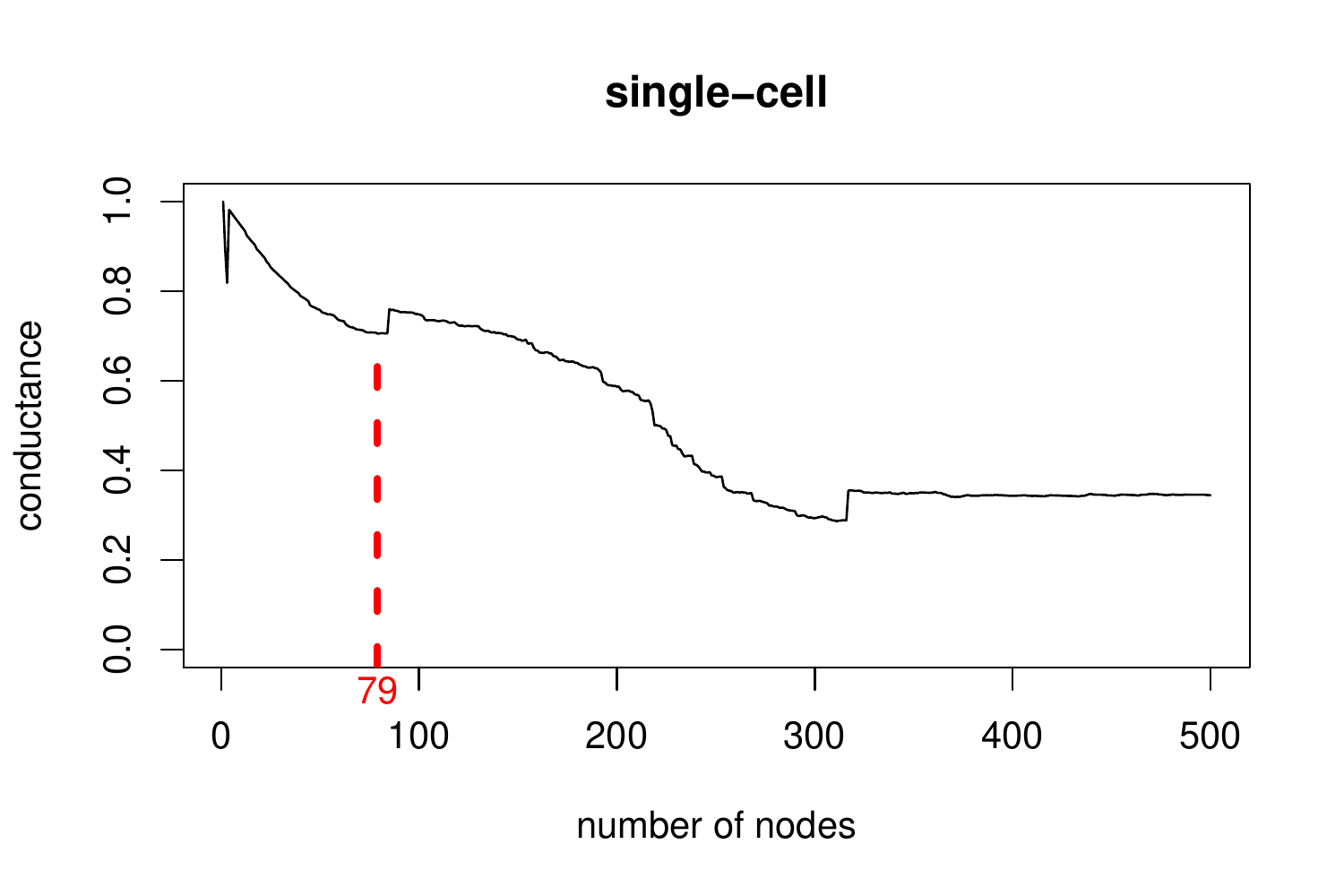}
     \caption{} 
\end{subfigure}
\begin{subfigure}[b]{0.45\textwidth}
    \centering
   \includegraphics[width=\textwidth]{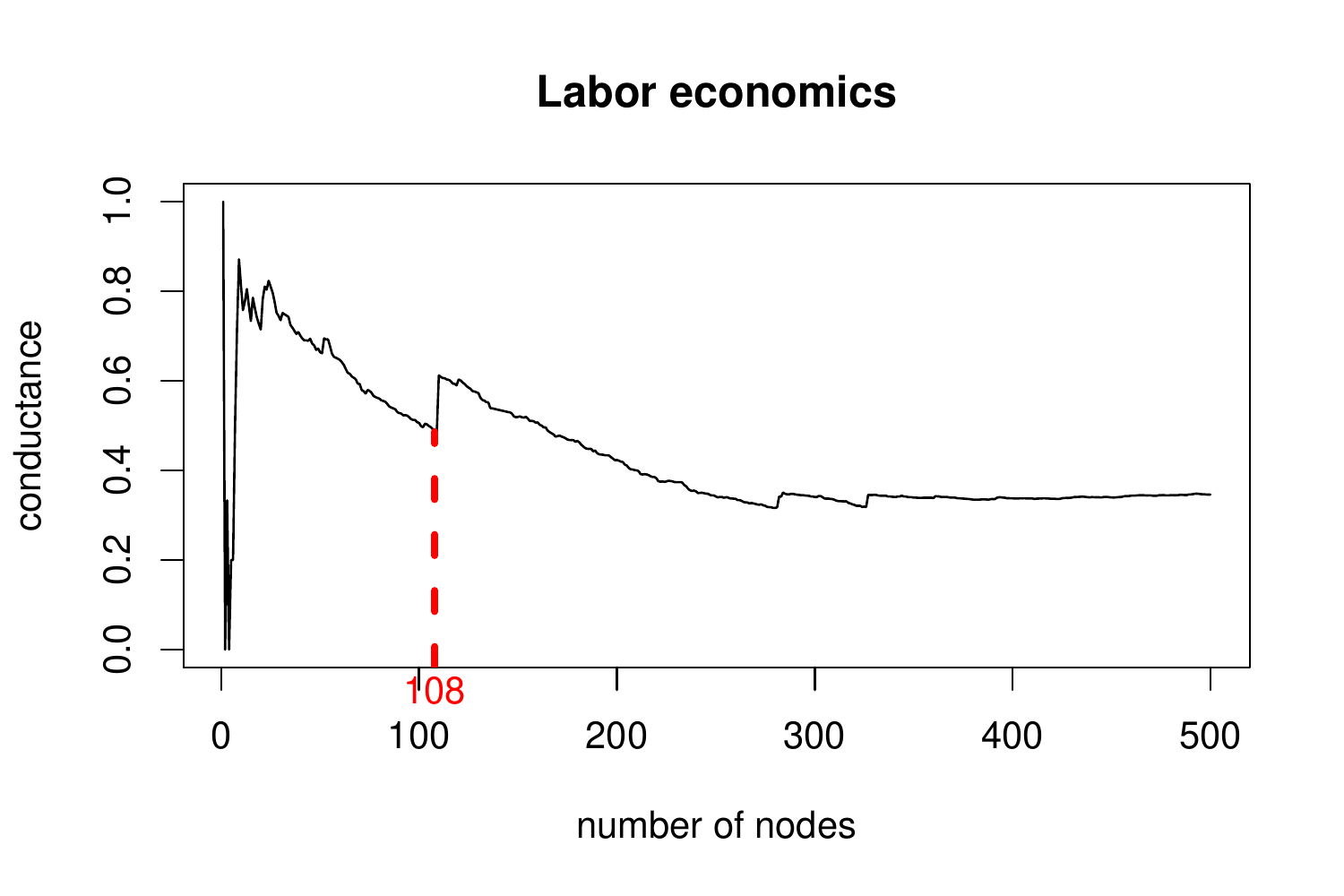}
     \caption{} 
\end{subfigure}
\begin{subfigure}[b]{0.45\textwidth}
    \centering
   \includegraphics[width=\textwidth]{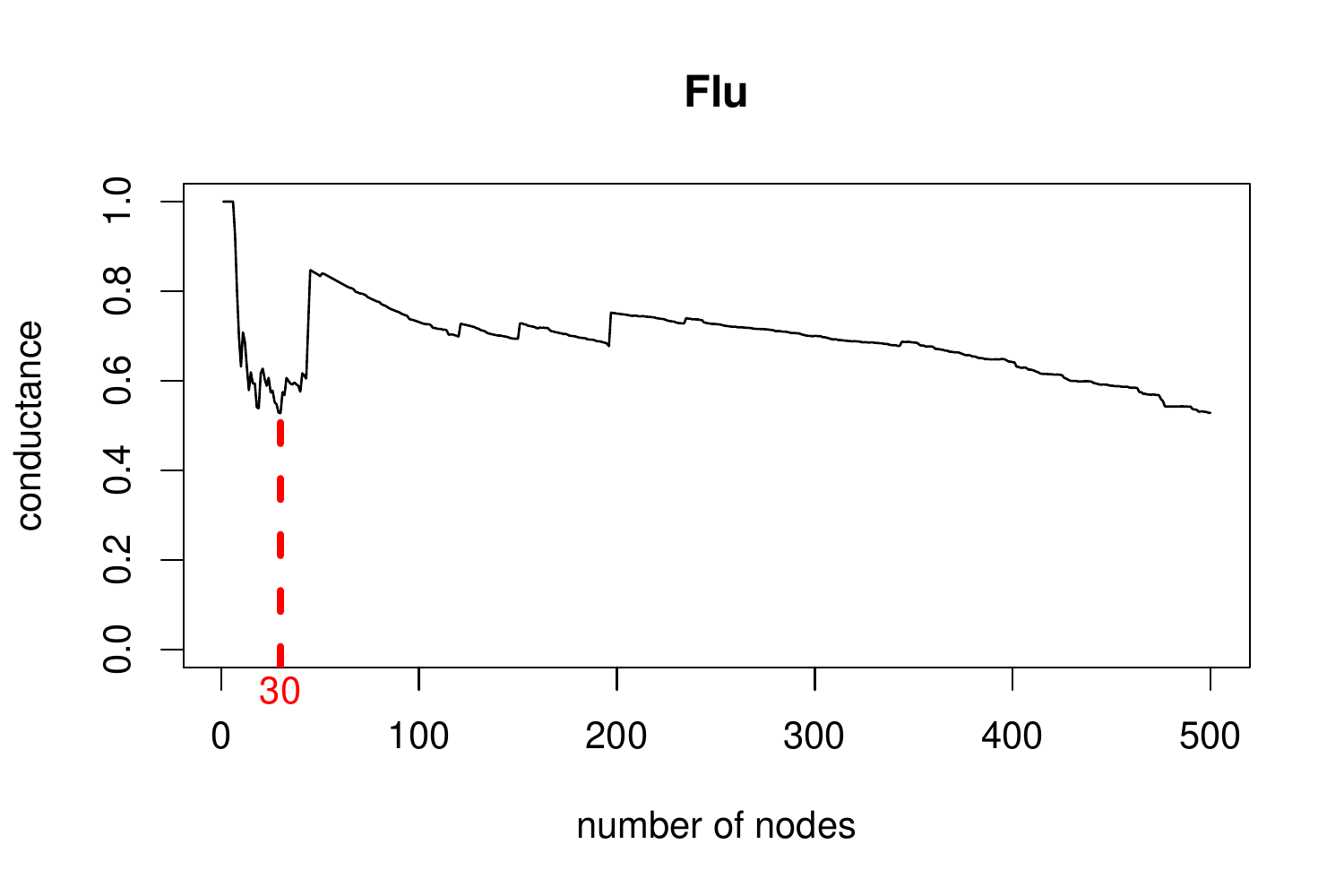}
     \caption{} 
\end{subfigure}
\caption{The conductance plots for each topic: (a) single-cell transcriptomics, (b) labor economics, and (c) flu. The x-axis is truncated at 500 for visual clarity.}\label{fig:conductance_topic}
\end{figure}


For each topic, we construct the preference vector by Eq~\eqref{eq:pi_vector}. We choose the threshold $t$ based on the citation counts from the topic papers to the source papers. For the topic ``single-cell" and ``labor economics", the top papers receive more than 90 citations; we set $t = 10$. For the topic ``flu", the highest citation count is less than 90, and we set $t = 5$. 

The conductance plot for each topic is shown in Figure~\ref{fig:conductance_topic}. In most cases, there is an obvious local minimum leading to a reasonable community size. In (b), we choose the first minimum occurring after $n \geq 10$ for a more plausible subnetwork size and clearer interpretation of result.


\subsection{Additional case studies}
\label{sec:more_cases}
\spacingset{1}
\begin{table}[h]
\centering
\footnotesize
\begin{tabular}{c|r|c | c}
\hline\hline
Topic & Size & Graph density  & Average clustering coefficient \\ \hline
Covid-19 (BIO) & 39 & 0.054  & 0.309    \\ \hline
Covid-19 (SOC) & 12 & 0.115  & 0     \\ \hline
\end{tabular}
\caption{Basic information for the networks in Figure~\ref{fig:case_2}.}
\end{table}

\noindent\textbf{Covid-19 (BIO \& SOC)} 

As a rapidly emerging topic of wide scientific interests and public relevance, we apply our local clustering procedure to the topic Covid-19. We search papers with the word ``covid" in their abstracts and either have the category label BIO or SOC. We treat these two sets of papers separately as we expect them to focus on different aspects of the pandemic in their studies.

Note that as we finished the data collection process before 2020 December, most of the research related to vaccines or new Covid-19 strains (e.g., the Delta variant)  had not yet appeared.
For Covid papers in BIO, a considerable number of papers found by the clustering procedure are on survival analysis. In particular, we can observe a hub at node 24 \citep{fine1999proportional} in Figure~\ref{fig:case_covid_bio}. This paper proposed the Fine-Gray method, which is popular in competing risk analysis (a type of survival analysis). The cause-specific hazard functions with explanatory covariates are commonly used in this type of analysis, but they often lack interpretations. As a result, clinicians prefer the cumulative incidence functions that are the marginal probability of certain events. \cite{fine1999proportional} modeled the cumulative incidence function by a proportional hazards model, which helps analysts measure the effect of covariates.  Another hub is node 39 \citep{benjamini1995controlling}, pointing to the need for multiple testing in many analyses in this field.

On the other hand, the papers found for Covid-19 SOC focus more on the societal reactions under a pandemic.  In Figure~\ref{fig:case_covid_soc},  we can observe a hub centered at node 10 \citep{abadie2010synthetic}. This paper analyzed the effects of  California Proposition 99 (Tobacco Tax and Health Protection Act of 1988) through the synthetic control method that is commonly used to evaluate the effect of an intervention in comparative case studies. At the beginning of the pandemic, the effect of various quarantine measures became of great public concern. Synthetic control methods are used to study the outcomes of different quarantine  policies.

\spacingset{1}
\begin{figure}
    \begin{subfigure}[b]{1\textwidth}
    \caption{Covid-19 (BIO)} \label{fig:case_covid_bio}
    \vspace{-0.2cm}
     \scalebox{0.5}{\includegraphics{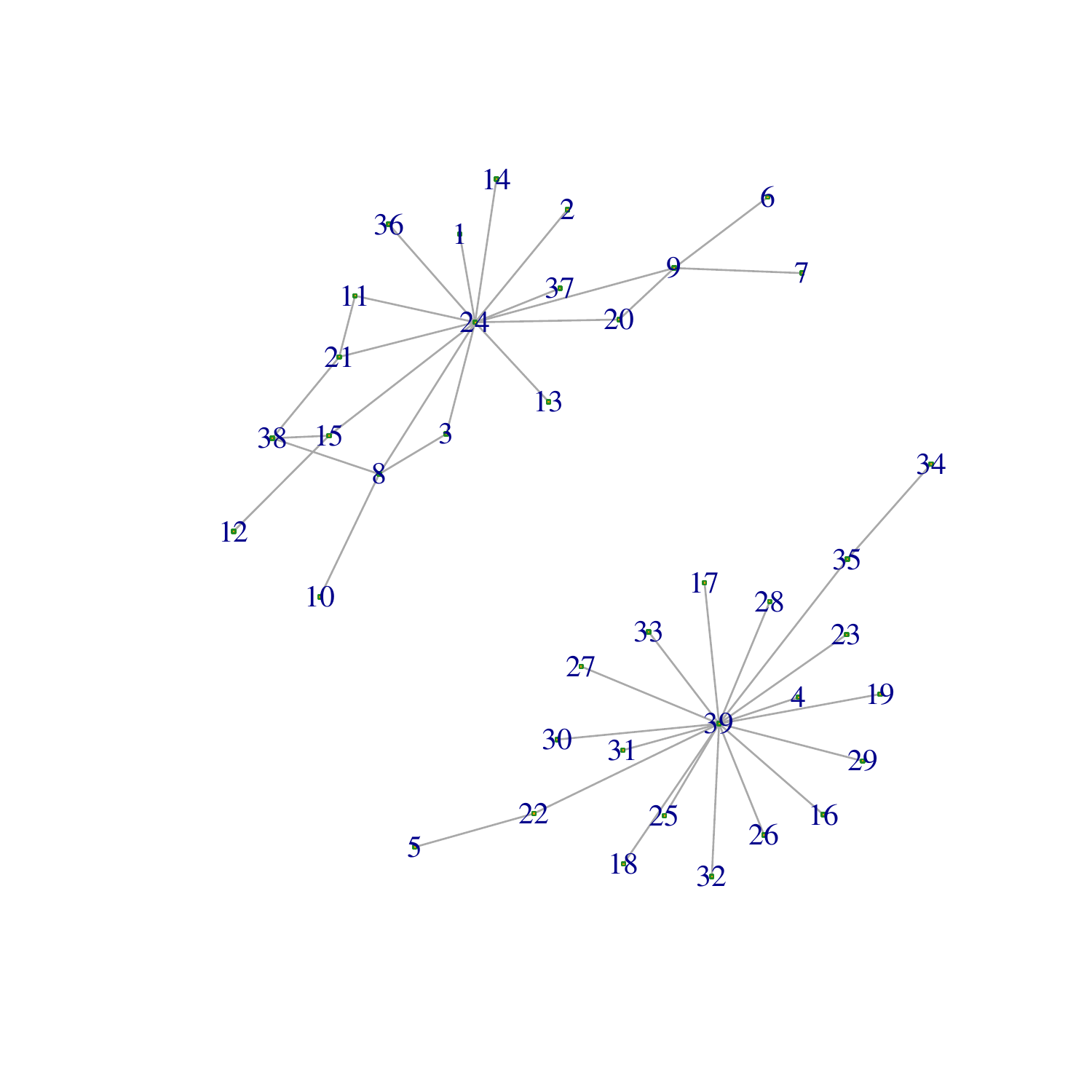}}
    \scalebox{0.45}{\includegraphics{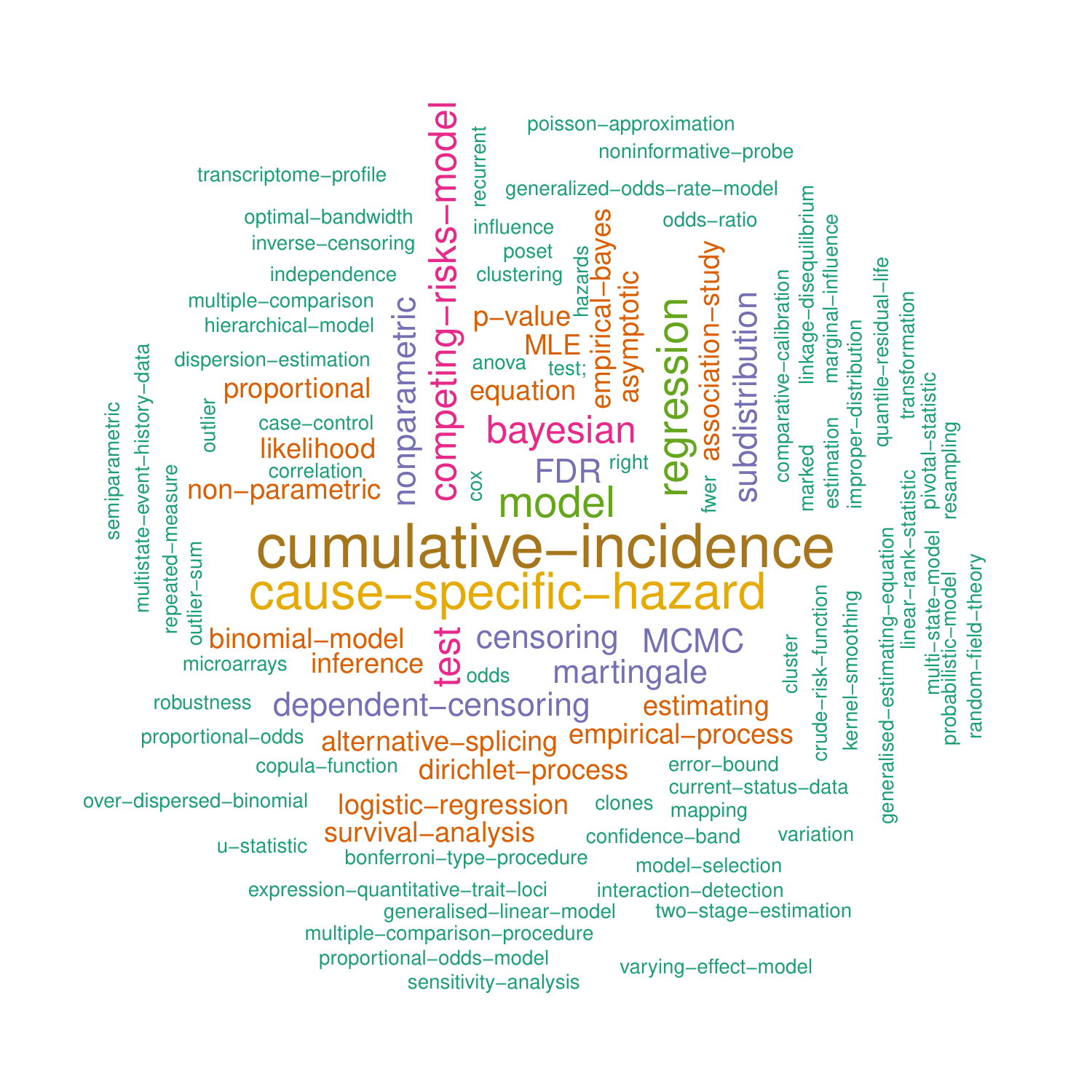}}
    \end{subfigure}
 \vspace{1cm}   
 \begin{subfigure}[b]{1\textwidth}
 \caption{Covid-19 (SOC)} \label{fig:case_covid_soc}
 \vspace{-0.2cm}
     \scalebox{0.5}{\includegraphics{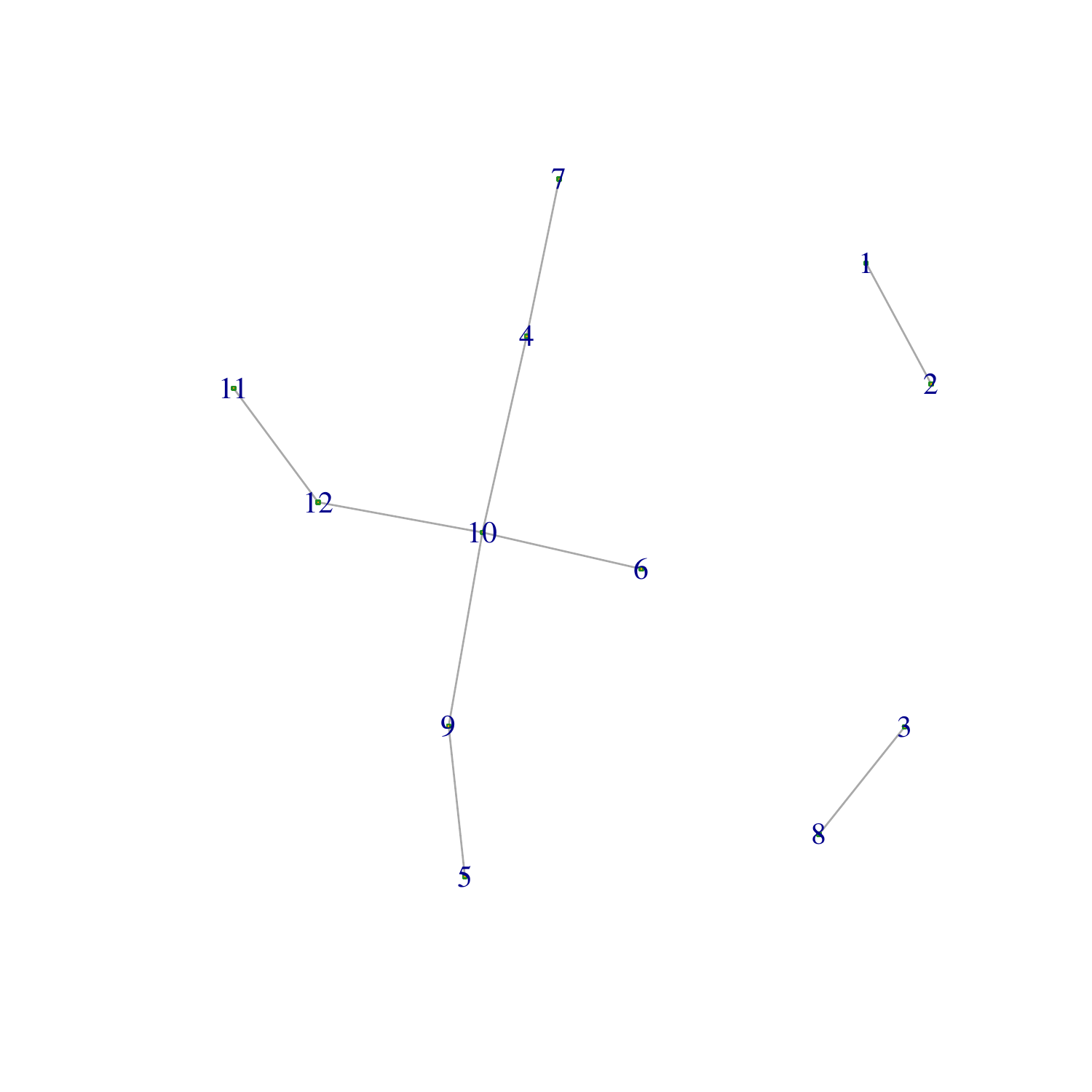}}
    \scalebox{0.45}{\includegraphics{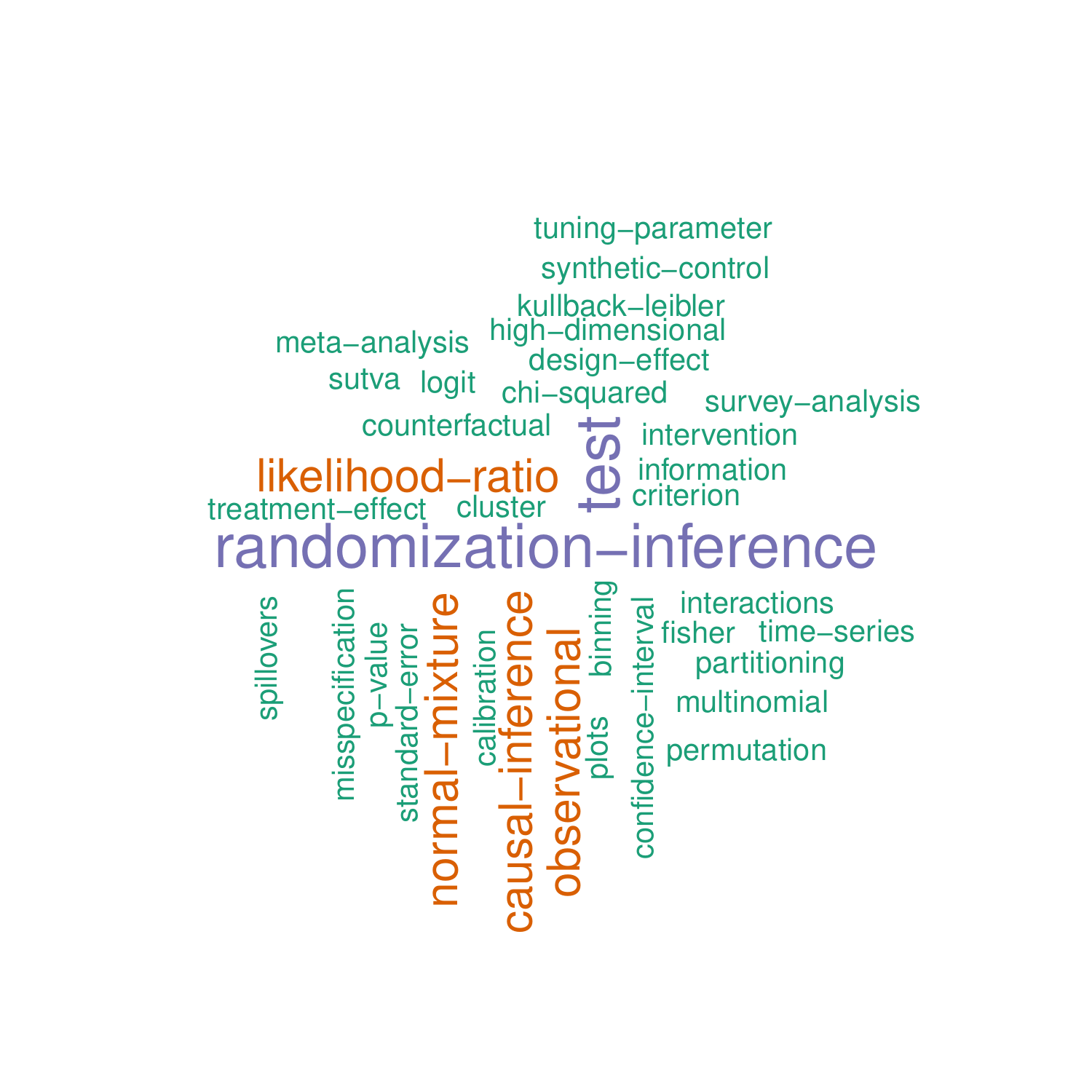}} \vspace{-1cm}
    \end{subfigure}
\caption{Networks and word clouds generated from the source papers found by local clustering for each topic.}\label{fig:case_2}
\end{figure}

\begin{figure}[h]
\centering
\begin{subfigure}[b]{0.48\textwidth}
    \centering
   \includegraphics[width=\textwidth]{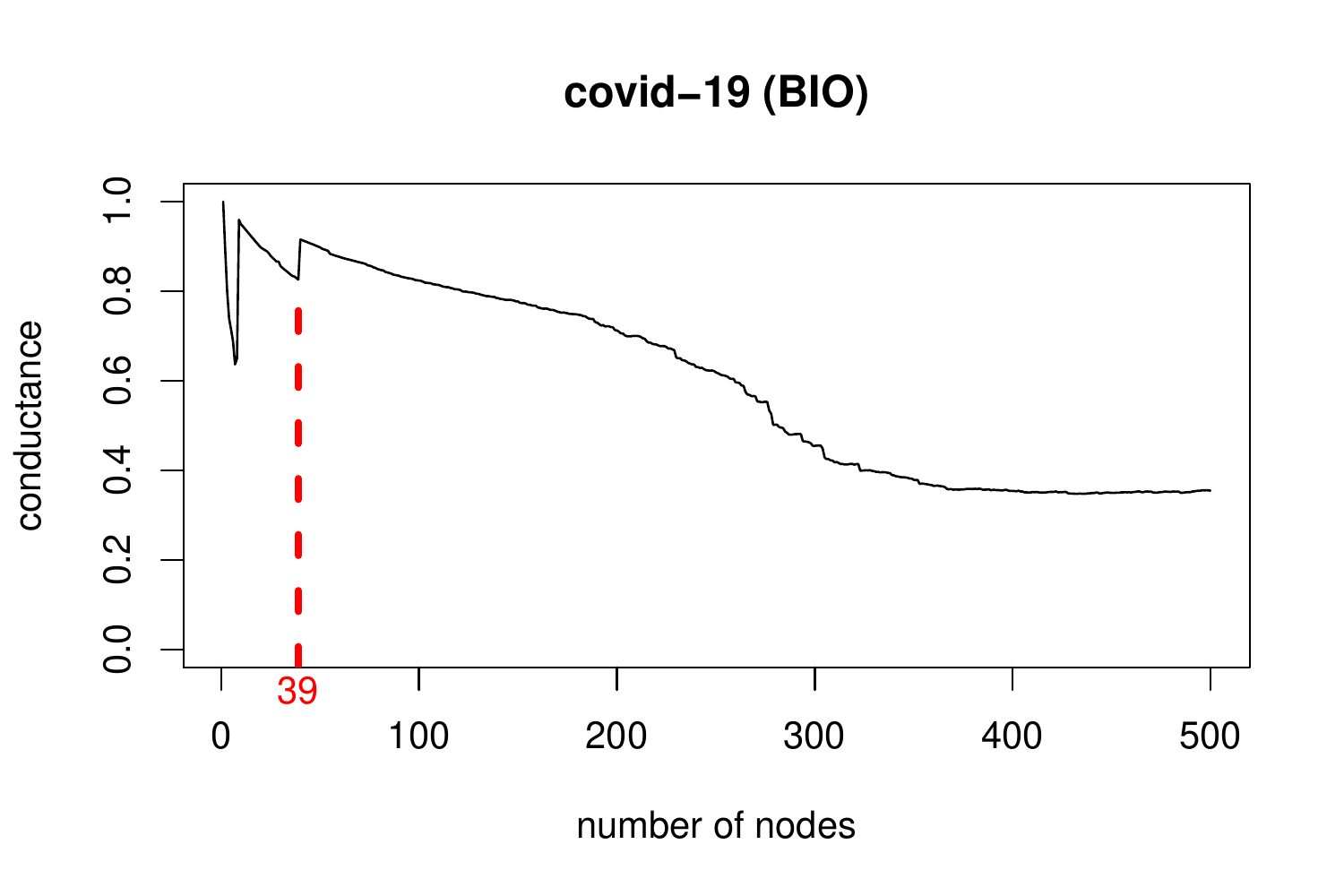}
     \caption{} 
\end{subfigure}
\begin{subfigure}[b]{0.48\textwidth}
    \centering
   \includegraphics[width=\textwidth]{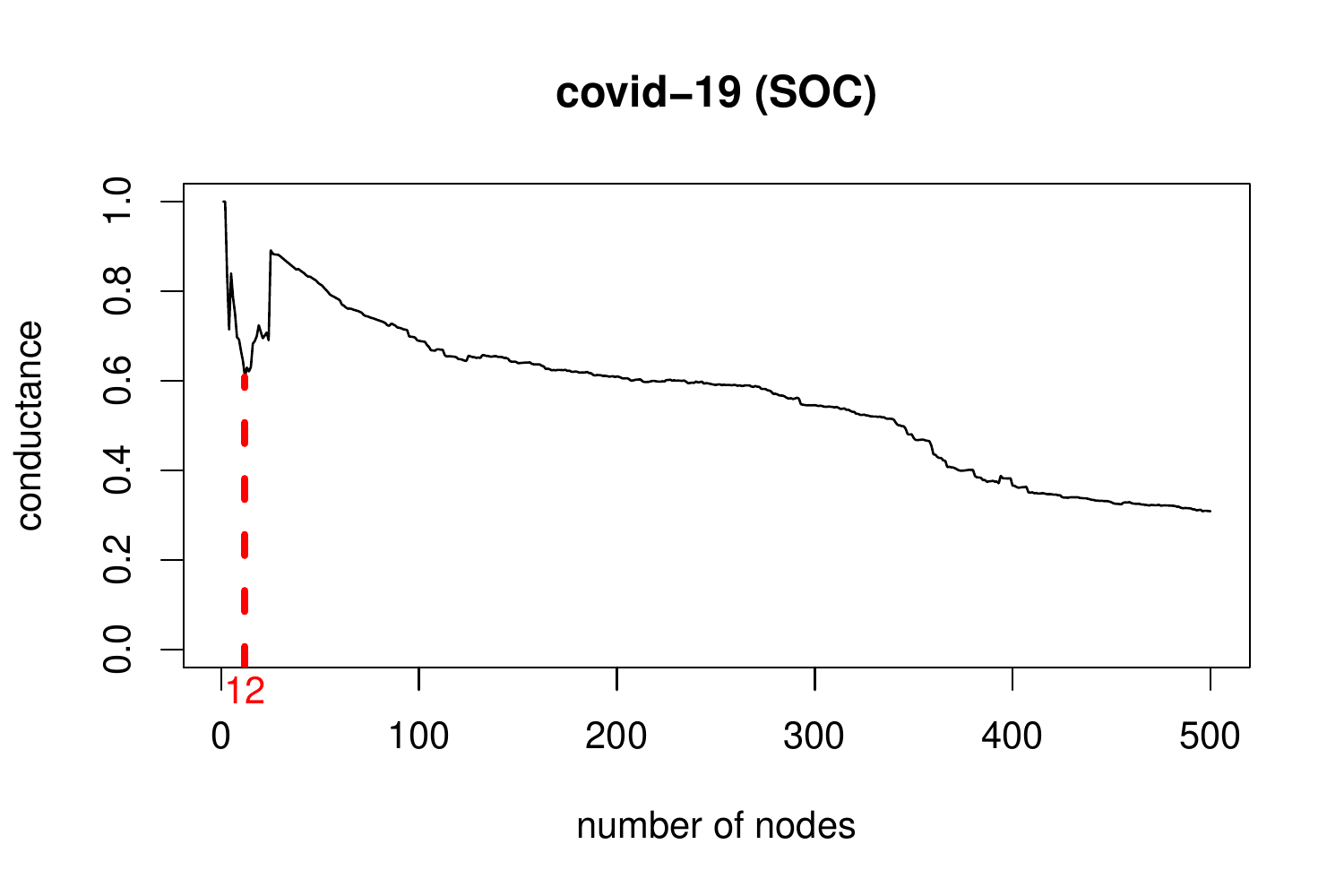}
     \caption{} 
\end{subfigure}
\caption{Conductance plots for: (a) Covid-19 BIO; (b) Covid-19 SOC.}
\end{figure}

\end{document}